\documentclass[10pt]{IEEEtran}

\usepackage{amsmath}
\usepackage{mathtools}
\usepackage{graphicx}
\usepackage{enumerate}
\usepackage{enumitem}
\usepackage{url}
\usepackage{subfigure}
\usepackage{algorithm}
\usepackage{algpseudocode}
\usepackage{stfloats}
\usepackage{lineno}
\usepackage{hyperref}
\usepackage{bm}
\usepackage{bbm}
\usepackage{amssymb}
\usepackage{xcolor}
\usepackage{float}
\usepackage{cite}
\usepackage{tabularx}
\usepackage{tabu}
\usepackage{multicol}
\usepackage{multirow}
\usepackage{colortbl,booktabs,threeparttable}
\usepackage{dcolumn}
\usepackage{flushend}
\usepackage{soul}
\usepackage{ragged2e}
\usepackage{relsize}

\graphicspath{{Figures/}}

\newcommand{\rscl}[1]{\mathrm{#1}}  
\renewcommand{\vec}[1]{\bm #1}
\newcommand{\rvec}[1]{\mathbf{#1}}
\newcommand{\mat}[1]{\bm #1}

\renewcommand{\math}[1]{\hat{\bm #1}}

\newcommand{\vecb}[1]{\bar{\bm #1}}
\newcommand{\rvecb}[1]{\bar{\mathbf{#1}}}

\renewcommand{\cal}[1]{\mathcal{#1}}

\newcommand{\R}{\mathbb{R}}
\newcommand{\C}{\mathbb{C}}
\newcommand{\E}{\mathbb{E}}

\renewcommand{\th}{\text{th}}

\newcommand{\T}{\mathsf{T}}
\renewcommand{\H}{\mathsf{H}}

\newcommand{\round}[1]{\langle{#1}\rangle}

\newcommand{\opreal}{\operatorname{real}}
\newcommand{\opimag}{\operatorname{imag}}
\newcommand{\KL}{\operatorname{KL}}

\DeclareMathOperator*{\argmin}{argmin}

\newcommand{\defeq}{\coloneqq}
\newcommand{\stp}{\hfill $\square$}     
\newcommand{\captext}[1]{\texorpdfstring{#1}{}} 

\newcommand{\tabincell}[2]{\begin{tabular}{@{}#1@{}}#2\end{tabular}}
\newcommand{\quotemark}[1]{``#1”}

\definecolor{hl-bg-color}{RGB}{255,255,215}
\sethlcolor{hl-bg-color} 
\definecolor{new-magenta}{RGB}{255,0,255}
\soulregister\cite7
\soulregister\citep7
\soulregister\citet7
\soulregister\ref7
\soulregister\eqref7
\ifdefined\HIGHLIGHT
    \newcommand{\red}[1]{{\color{red}{#1}}}
    \newcommand{\blue}[1]{{\color{blue}{#1}}}

    \newcommand{\hlblue}[1]{\blue{\hl{#1}}}

\else
    \renewcommand{\hl}[1]{{#1}}
    \newcommand{\red}[1]{{#1}}
    \newcommand{\blue}[1]{{#1}}

    \newcommand{\hlblue}[1]{{#1}}

\fi

\newcommand*{\IEEE}{}
\ifdefined\IEEE
    \newtheorem{theorem}{{Theorem}}
    \newtheorem{proposition}{{Proposition}}
    
    \newtheorem{lemma}{{Lemma}}
    \newtheorem{method}{{Method}}

    \newtheorem{remark}{{Remark}}

    \newtheorem{example}{{Example}}
    
    \newtheorem{myrule}{{Rule}}

\fi

\begin{document}
\newpage
\title{
A New Particle Filter for Target Tracking in MIMO OFDM Integrated Sensing and Communications
}

\author{Shixiong Wang,~
        Wei Dai,
        and Geoffrey Ye Li,~\IEEEmembership{Fellow,~IEEE}
\thanks{S. Wang is with the School of Mathematics and Statistics, Xi’an Jiaotong University, Xi’an 710049, China (E-mail: s.wang@xjtu.edu.cn; s.wang@u.nus.edu). W. Dai and G. Y. Li are with the Department of Electrical and Electronic Engineering, Imperial College London, London SW7 2AZ,
United Kingdom (E-mails: wei.dai1@imperial.ac.uk; geoffrey.li@imperial.ac.uk). 
(\textit{Corresponding Author: S. Wang.})
}
}

\maketitle

\begin{abstract}
Target tracking using multi-input multi-output (MIMO) pulse-Doppler radars faces three long-standing obstacles: a) the absence of reliable likelihood models for raw radar data; b) the computational and statistical complications that arise when nuisance parameters (e.g., complex path gains) are augmented into state vectors; and c) the prohibitive computational burden of extracting noisy measurements of range, Doppler, and angles from snapshots. Motivated by an optimization-centric interpretation of Bayes' rule, this article addresses these challenges by proposing a new particle filtering framework that evaluates each hypothesized state using a tailored cost function, rather than relying on an explicit likelihood relation. In addition, we study the integration of the proposed particle filter into MIMO orthogonal frequency-division multiplexing (OFDM) systems, aiming to equip modern communication infrastructure with integrated sensing and communications (ISAC) capabilities. Experiments suggest that the proposed particle filter achieves high tracking performance and computational efficiency in MIMO-OFDM ISAC systems, particularly when wide bandwidth, extended on-target time, and large antenna aperture are utilized.
\end{abstract}


\section{Introduction} \label{sec:introdction}
\IEEEPARstart{I}{n} emerging wireless systems, integrated sensing and communications (ISAC) has been envisioned as a promising solution for enabling both communication and sensing functions on a common hardware platform, leveraging shared resources in power, time, space, and frequency, such as transmit and receive antennas, radio frequency chains, digital signal processing units, and transmit waveforms \cite{mishra2019toward,liu2023seventy}. This integration can reduce hardware volume, improve resource efficiency, and provide mutual benefits between sensing and communications \cite{liu2020radar,dong2025communication}. Among many sensing aspects (e.g., detection, measuring, tracking, and imaging \cite{richards2022fundamentals}), target tracking is fundamental and plays a critical role in a broad range of ISAC applications, such as sensing-assisted communications \cite{liu2020radar} and intelligent transportation \cite{zhang2025sensing}. This article is concerned with target tracking problems in ISAC. Particularly, the multi-input multi-output orthogonal frequency-division multiplexing (MIMO-OFDM) systems are investigated for three reasons. 
First, MIMO technology enables spatial and waveform diversity to enhance communication performance in terms of throughput and reliability \cite{lu2014overview,li1999channel,li2002mimo}, as well as radar sensing performance in terms of angle resolution and spatial interference suppression \cite{stoica2007probing,xu2008target,cui2013mimo}. 
Second, OFDM techniques allow frequency-domain signal processing using fast Fourier transform (FFT) if the duration of the cyclic prefix is sufficient to cover all round-trip target delays \cite[p.~19]{koivunen2024multicarrier}, \cite{zhang2024input}, which brings significant computational benefits to real-world radar operations \cite{sturm2011waveform}, \cite{xiao2024novel}, \cite[Eqs.~(62)-(63)]{keskin2023monostatic}, especially compared to time-domain integration-based matched filtering for ranging \cite{mercier2020comparison,san2007mimo}. 
Third, modern mature wireless communication systems are built upon MIMO-OFDM, such as 4G LTE, 5G NR, Wi-Fi 5/6/6E/7, and even upcoming 6G schemes; the adoption of MIMO-OFDM in 6G has been confirmed by the 3GPP RAN1 Committee in August 2025. Therefore, it is natural to equip these widespread systems with sensing capabilities without requiring large-scale infrastructure modification and software upgradation \cite{sturm2011waveform,koivunen2024multicarrier}. Research on MIMO ISAC, OFDM ISAC, and MIMO-OFDM ISAC is flourishing. Remarkable progress in this and related areas includes the following: 1) information-theoretic analyses of systematic and methodological performances \cite{xiong2023fundamental,olson2023coverage,ahmadipour2022information,liu2025cpofdm}; 2) signal processing methods in target detection and parameter estimation \cite{berger2010signal,zheng2017super,wu2022super,xiao2024novel}, waveform design \cite{hu2022low,zhang2024input,li2024mimo}, beamforming and precoding \cite{liu2021cramer,lu2024random}, and pulse shaping and modulation \cite{du2024reshaping,liu2025uncovering}; 3) and practical developments of real-world systems \cite{xu2022experimental}. Target tracking for MIMO-OFDM ISAC systems is also reported, although sparsely, such as \cite{gu2023beam}, where only angles are tracked.

The target tracking problem for a single object can be formulated as \cite{li2003survey,herbert2017mmse}
\begin{equation}\label{eq:target-trcking}
    \left\{
        \begin{array}{cl}
           \rvec x_k  &=  \vec f_k(\rvec x_{k-1}, \rvec w_{k-1}) \\
           \rvec y_k  &=  \vec h_k(\rvec x_k, \rvec v_k),
        \end{array}
    \right.
\end{equation}
where $k = 0, 1, 2, \ldots$ denotes the discrete time index, $\rvec x_k \in \R^{d_x}$ the target state vector, $\rvec y_k \in \R^{d_y}$ (or $\rvec y_k \in \C^{d_y}$) the measurement vector, $\vec f_k(\cdot, \cdot)$ the state transition function, $\rvec w_k  \in \R^{d_w}$ the process noise vector, $\vec h_k(\cdot, \cdot)$ the state measurement function, and $\rvec v_k  \in \R^{d_v}$ (or $\rvec v_k  \in \C^{d_v}$) the measurement noise vector. Mathematically, a target tracking problem is stated as using the system model \eqref{eq:target-trcking} and the measurement history
$\cal Y_{k} \defeq \{\rvec y_0, \rvec y_1, \ldots, \rvec y_k\}$, up to and including time $k$, to infer the unknown target state $\rvec x_k$. Depending on whether the system is linear or nonlinear, state estimation approaches, such as Kalman or particle filtering \cite{simon2006optimal,arulampalam2002tutorial,godsill2019particle}, serve as the inference engine. 
The state vector $\rvec x_k$, as per specific applications, can consist of different variables. On the one hand, in Cartesian coordinates, $\rvec x_k$ can denote positions, velocities, and accelerations of the target of interest (ToI) in two- or three-dimensional spaces \cite{li2003survey}. On the other hand, in polar coordinates, $\rvec x_k$ can include range, radial speed (or Doppler frequency shift), angles (e.g., elevation, azimuth), and angular speeds of the ToI in the radar's view field \cite{ratpunpairoj2015particle,huleihel2013optimal}. At times, statistical signal processing methods mandate $\rvec x_k$ to contain some nuisance parameters, for example, the turning rate in a coordinated-turn maneuvering model \cite{li2003survey}, the model index in a multi-model tracking system \cite{wang2023distributionally}, or the real and imaginary components of complex path gains \cite{huleihel2013optimal,ratpunpairoj2015particle,gu2023beam,herbert2017mmse}. 
According to the level of radar signal processing, the following two cases for the measurement vector $\rvec y_k$ are standard:
\begin{itemize}
    \item \textit{Information-Level Tracking}: In this scheme, $\rvec y_k \in \R^{d_y}$ denotes the radar's measurements of range, radial velocity, and angle(s) that are to be assimilated by the target's motion mechanism \eqref{eq:target-trcking}; see, e.g., \cite[Chap.~7]{richards2022fundamentals}. At times, $\rvec y_k \in \R^{d_y}$ can also denote the unfiltered positions and velocities of the ToI, calculated using raw range-Doppler-angle measurements \cite{li2001survey}, \cite[p.~8]{stone2013bayesian}. This is the most common way in real-world radar target tracking due to its simplicity in operation. When detection is jointly considered, this paradigm is widely referred to as \textit{detect-then-track} in the literature \cite{davey2013snr}.

    \item \textit{Signal-Level Tracking}: In this scheme, $\rvec y_k \in \C^{d_y}$ denotes the radar's base-band raw in-phase/quadrature (I/Q) snapshots \cite{ratpunpairoj2015particle,huleihel2013optimal,herbert2017mmse,herbert2018computationally}. Note that within a target tracking period $k$, a radar can collect several snapshots in fast and slow times; see \cite[Sec.~1.5.1]{richards2022fundamentals}. In this case, the measurement model $\rvec y_k = \vec h_k(\rvec x_k, \rvec v_k)$ is severely nonlinear in $\rvec x_k$,\footnote{For example, see \eqref{eq:general-radar-measurement}-\eqref{eq:measurements}, where the measurement function is severely nonlinear in ToI's state parameters $(\tau_0, \nu_0, \theta_0)$.} so particle filtering is chosen over Kalman filtering and its variants, such as extended and unscented Kalman filtering, for higher tracking accuracy; see, e.g., \cite[Figs.~4-5]{gu2023beam}, \cite{ratpunpairoj2015particle,huleihel2013optimal,herbert2017mmse}. When detection is jointly considered, this paradigm is known as \textit{track-before-detect}, which outperforms detect-then-track in the low signal-to-noise ratio (SNR) regimes \cite{davey2007comparison,davey2013track}.
\end{itemize}
However, three drawbacks of the existing schemes for radar target tracking have to be outlined:
\begin{itemize}
    \item In information-level tracking, at each time $k$, obtaining measurements $\vec y_k$ and calibrating their uncertainty quantification parameters (i.e., distribution type, mean, covariance, and other statistics of $\rvec v_k$) are computationally prohibitive, especially when advanced optimization techniques, such as compressed sensing and tensor decomposition, are utilized; see \cite{zheng2017super,wu2022super}, \cite[Table~II]{zhang2021overview}. Even for matched-filter-based methods that are considered computationally efficient \cite{xiao2024novel}, searching for the peaks of the filter outputs in a validation gate is time-consuming as well. The computational burden further increases when particle filtering is afterward employed to link these measurements $\vec y_k$ with the tracking model \eqref{eq:target-trcking}. This issue motivates the methodological shift from information-level tracking to signal-level tracking.

    \item In signal-level tracking, when raw snapshots $\vec y_k$ are obtained, the next key step is to calculate the likelihoods $q(\vec y_k | \vec x_k)$ of hypothesized states $\vec x_k$; note that the likelihood function $q(\vec y_k | \vec x_k)$ is induced by a nonlinear measurement model $\rvec y_k  =  \vec h_k(\rvec x_k, \rvec v_k)$. The challenge is that the nonlinear measurement model $\rvec y_k  =  \vec h_k(\rvec x_k, \rvec v_k)$ is unknown (or uncertain) in practice, so is the probabilistic model $q(\vec y_k | \vec x_k)$, because the information about environmental scatterers (including targets of no interest) and channel noises is lacking. Assuming inexact distribution families or statistics for $q(\vec y_k | \vec x_k)$ cannot offer trustworthy likelihood information.
    
    \item To handle unknown parameters involved in the state transition and measurement mechanisms, we usually need to include the nuisance parameters into the state vector $\rvec x_k$, e.g., the turning rate of the coordinated turn model \cite{li2003survey}, the real and imaginary components of the complex path gain due to the ToI \cite{ratpunpairoj2015particle,huleihel2013optimal}. However, this operation imposes nontrivial complications because the statistical and computational efficiency of particle filters significantly deteriorate when the dimension $d_x$ of the state vector $\rvec x_k$ increases \cite{doucet2001sequential,snyder2008obstacles}. Even though the Rao-Blackwellized (i.e., marginalized) particle filtering \cite{schon2005marginalized,cappe2007overview} has the potential to mitigate such complications, it remains inadequate for tracking rapidly jumping nuisance parameters if an accurate state-transition model for these parameters is not available. For example, complex path gains can change abruptly across consecutive tracking periods $k$ for fast-fading channels, so the usual first-order Gauss–Markov models (e.g., random walk and constant-velocity models) are not appropriate; see the experiment section for a specific discussion.
\end{itemize} 
To address the above issues of radar target tracking, this article proposes a new particle filtering framework for signal-level target tracking (PF-SLTR), by leveraging the optimization-centric interpretation of Bayes' rule, i.e., the Gibbs posterior \cite{germain2016pac,bissiri2016general,knoblauch2022optimization}. The advantages of this new framework are threefold:
\begin{itemize}
    \item We do not require an operationally separate, and computationally expensive, phase to obtain raw radar measurements (e.g., range, Doppler, angles) and calibrate their uncertainty quantification.

    \item We do not mandate the knowledge of the probabilistic model $q(\vec y_k | \vec x_k)$ to calculate $\vec y_k$-data evidence (i.e., likelihoods) of particles (i.e., hypothesized state values).

    \item We do not augment the parameters of the complex path gain, due to the ToI, into the state vector.
\end{itemize}
Technically, the pivotal yet challenging step is to devise an efficient cost function $\rho(\vec y_k, \vec x_k)$ to quantify the $\vec y_k$-data evidence of a hypothesized state $\vec x_k$, which requires problem-specific investigation for signal-level target tracking using MIMO pulse-Doppler radar. Existing general principles and methodologies in applied statistics and machine learning are not directly applicable since the design of $\rho$ is domain-specific and there is no one-size-fits-all solution; cf., e.g., \cite{germain2016pac,bissiri2016general,baek2023generalized,knoblauch2022optimization,martin2022direct}.

Moreover, we adapt the proposed PF-SLTR to MIMO-OFDM ISAC systems, where MIMO-OFDM communication waveforms are reused for high-accuracy target tracking. Radar signal processing for target tracking in both the time domain and the frequency domain is investigated, corresponding to the pulsed radar scheme \cite{lellouch2014processing} and the continuous-wave radar scheme \cite{koivunen2024multicarrier,zhang2024input}, respectively. 

In summary, the contributions of this article can be outlined as follows: 1) A new particle filtering framework for signal-level target tracking is studied, which frees practitioners from obtaining raw range-Doppler-angle measurements, specifying the exact likelihood model, and augmenting path gain parameters into the state vector; 2) The usage of the proposed PF-SLTR in MIMO-OFDM ISAC systems is investigated, by designing waveform structures and signal processing chains.

\textit{Notations}: Let $\R^d$ and $\C^d$ denote the $d$-dimensional spaces of real and complex numbers, respectively. We use lowercase letters for vectors (e.g., $\vec x$) and uppercase ones for matrices (e.g., $\mat X$), both in boldface; scalars are not boldfaced. Random quantities are written in Roman fonts (e.g., $\rvec x$, $\rvec X$), while deterministic ones are in italic fonts (e.g., $\vec x$, $\vec X$). Let $\round{x}$ denote the rounding operation of real number $x$ to the nearest integer, and $[N] \defeq \{0, 1, 2, \ldots, N-1\}$ the running index set starting from zero for integer $N$. Let $\mat I_d$ denote the $d$-dimensional identity matrix, and $\mat X^\T$ and $\mat X^\H$ the transpose and conjugate transpose of matrix $\mat X$, respectively.

\section{Measurement Model and Problem Statements}\label{sec:radar-model}
\subsection{Measurement Model of Signal-Level Target Tracking}
In this section, we review the MIMO multi-carrier pulse-Doppler radar model that is suitable for various transmit waveforms, including reused MIMO-OFDM communication waveforms. For conciseness of presentation and without loss of generality, this article specifically investigates a two-dimensional planar target tracking problem so that only one direction-of-arrival (DoA) angle for each target is to be examined. We consider a mono-static ISAC station that is
\begin{itemize}
    \item Located at the origin $[0, 0]$ and staying unmoving.
    
    \item Operating with center frequency $f_c$, wavelength $\lambda \defeq C/f_c$, bandwidth $B$, sampling time period $T_s \defeq 1/B$, and per-antenna power budget $P_t$; $C$ is the light speed.

    \item Equipped with a transmit antenna array that includes $N_t$ collocated omnidirectional elements; the transmit steering vector is $\vec a(\theta)$ in direction $\theta$, which is determined by the array geometry.

    \item Equipped with a receive antenna array that includes $N_r$ collocated omnidirectional elements; the receive steering vector is $\vec b(\theta)$ in direction $\theta$. The receive array and transmit array can operate in either the time-division duplex mode (cf. pulsed radars) or the full-duplex mode (cf. continuous-wave radars with pulse-Doppler processing).

    \item Configured with target tracking period $T_t$, i.e., the time interval between $k-1$ and $k$ in  \eqref{eq:target-trcking}.
\end{itemize}

The waveform structures in terms of radar sensing, as illustrated in Fig. \ref{fig:waveform-structure-pulsed}, are specified as follows.
\begin{itemize}
    \item Each target tracking period $T_t$ includes one coherent processing interval (CPI) $T_i$, within which the motion parameters of the ToI remain unchanged and can be resolved. Note that $T_i \le T_t$.
    \item Each CPI constitutes $N_p$ pulses (i.e., $N_p$ slow times) and each pulse has a time duration $T_p$. Hence, the number of snapshots within a pulse is $L_{\text{ss}} \defeq \round{T_p/T_s}$. 
    \item The pulse repetition interval (PRI) is $T_r$, and therefore, the pulse repetition frequency (PRF) is $f_r \defeq 1/T_r$. Note that $T_p \le T_r \le T_i \le T_t$.
    \item Each PRI consists of $L$ sampling points, so $T_r = L T_s$. The set $[L]$ defines $L$ fast-time range bins; NB: $L_{\text{ss}} \le L$.
\end{itemize}
\begin{figure}[!htbp]
    \centering
    \includegraphics[width=6cm]{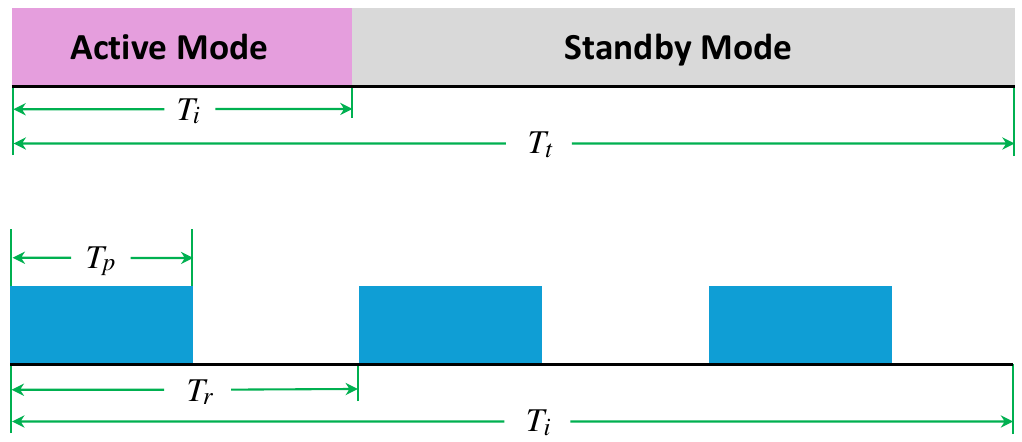}
    \caption{The time structure of a radar waveform. In the active mode, wireless signals are used for target tracking, while in the standby mode, no tracking task is performed. For illustration, three pulses are shown within one CPI $T_i$. When $T_r > T_p$ (usually $T_r \gg T_p$), the pulsed radar scheme applies; if $T_r = T_p$, the continuous-wave radar scheme applies. For both schemes, pulse-Doppler signal processing pipelines hold \cite{richards2022fundamentals,zhang2024input}.}
    \label{fig:waveform-structure-pulsed}
\end{figure}


The above pulse-Doppler radar configurations imply that
\begin{itemize}
    \item We have the CPI $T_i = N_p T_r$.
    
    \item The resolutions of round-trip delay and Doppler frequency are $1/B$ and $1/T_i$, respectively.
    
    \item The inter-pulse silent time is $T_r - T_p$.
    
    \item In time-division duplex mode, the maximum detectable round-trip delay is $\tau_{\text{max-d}} \defeq T_r - T_p$, and the minimum detectable round-trip delay is $\tau_{\text{min-d}} \defeq T_p$. In full-duplex mode, $\tau_{\text{max-d}} \defeq T_r$ and $\tau_{\text{min-d}} \defeq 0$. As a result, the maximum and minimum detectable distances are $R_{\text{max-d}} \defeq C\tau_{\text{max-d}}/2$ and $R_{\text{min-d}} \defeq C\tau_{\text{min-d}}/2$, respectively.
    
    \item The minimum and maximum unambiguous Doppler frequency shifts are $\nu_{\text{min-ua}} \defeq -f_r/2$ and $\nu_{\text{max-ua}} \defeq f_r/2$, respectively.
\end{itemize}
In addition, suppose that the ToI within the tracking time period $k$ is characterized by 
\begin{itemize}
    \item The true round-trip delay $\tau_0$, and hence, the true range $R_0 \defeq C\tau_0/2$. (Note that $R_{\text{min-d}} \le R_0 \le R_{\text{max-d}}$.)
    \item The true Doppler frequency shift $\nu_0$, and therefore, the true radial speed $v_0 \defeq \nu_0 \lambda/2$. (Note that $\nu_{\text{min-ua}} \le \nu_0 \le \nu_{\text{max-ua}}$.)
    \item The true DoA $\theta_0$ relative to the array boresight. (The transmit array and receive array are supposed to have the same DoA $\theta_0$ for the ToI.)
\end{itemize}

With the above system settings, in a given target tracking period $k$, the measuring mechanism of narrowband MIMO radars\footnote{When a wireless system is narrowband in terms of radar sensing, it may still be treated as broadband in terms of communications; for determination criteria, see \cite[Eq.~(3)]{san2007mimo}, \cite[p.~439]{heath2016overview}.} can be modeled as \cite[Eq.~(11)]{san2007mimo}, \cite[p.~439]{heath2016overview}, \cite{zhang2021overview},
\begin{equation}\label{eq:general-radar-measurement}
    \begin{array}{cl}
        \rvec y(t) &= \beta_0 \vec b(\theta_0) \vec a^\H(\theta_0) \rvec s(t - \tau_0) e^{j 2 \pi \nu_0 t} + \\
        & \quad \quad \displaystyle \sum^{N_s}_{i = 1} \beta_{i,0} \vec b(\theta_{i,0}) \vec a^\H(\theta_{i,0}) \rvec s(t - \tau_{i,0}) e^{j2 \pi \nu_{i,0} t} +  \rvec n(t),
    \end{array}
\end{equation}
where $t$, $0 \le t \le T_i \le T_t$, denotes the radar's logical time of signal processing in a CPI; $\beta_0$ is the path gain of the ToI; $N_s$ is the total number of scatterers (including, e.g., targets of no interest); $\beta_{i,0}$, $\theta_{i, 0}$, $\tau_{i, 0}$, and $\nu_{i,0}$ are the parameters of the $i^\th$ scatterer; $\rvec s(t) \in \C^{N_t}$ denotes the transmitted waveform at time $t$, $\rvec y(t) \in \C^{N_r}$ the received signal, and $\rvec n(t) \in \C^{N_r}$ the zero-mean circularly-symmetric complex Gaussian channel noise vector. For a target tracking problem aiming to resolve target parameters $(\tau_0, \nu_0, \theta_0)$, \eqref{eq:general-radar-measurement} can be compactly written as
\begin{equation}\label{eq:general-radar-measurement-compact}
    \begin{array}{cl}
        \rvec y(t) &= \beta_0 \vec b(\theta_0) \vec a^\H(\theta_0) \rvec s(t - \tau_0) e^{j 2 \pi \nu_0 t} + \rvec n_c(t),
    \end{array}
\end{equation}
where $\rvec n_c(t) \in \C^{N_r}$ denotes the channel interference-plus-noise (IPN) vector, which is not necessarily Gaussian. Supposing that the Doppler frequency shift is piece-wise constant across different PRIs, \eqref{eq:general-radar-measurement-compact} can be sampled at $t = pT_r + lT_s$, leading to the discrete-time representation 
\begin{equation}\label{eq:general-radar-measurement-compact-discret}
    \begin{array}{cl}
        \rvec y_p(l) &= \beta_0 \vec b(\theta_0) \vec a^\H(\theta_0) \rvec s_p(l - \tau_0) e^{j 2 \pi \nu_0 p T_r} + \rvec n_{c,p}(l),
    \end{array}
\end{equation}
where $p \in [N_p]$ is the slow-time Doppler-bin index and $l \in [L]$ is the fast-time range-bin index; $\rvec s_p(l)$ is the sampled transmitted waveform at $p^\th$ PRI and $l^\th$ range bin, $\rvec y_p(l)$ the sampled received signal, and $\rvec n_{c,p}(l)$ the sampled channel IPN vector. To avoid notational clutter, $\rvec s_p(l - \tau_0)$ is shorthand for $\rvec s_p(l - \round{\tau_0/T_s})$. The phase variation term $e^{j 2 \pi \nu_0 l T_s}$ is ignored.

At times, the frequency-domain representation of \eqref{eq:general-radar-measurement-compact-discret} would be computationally useful. Suppose that we have $N_{\text{cr}}$ FFT points for ranging. Let the frequency-bin width be $\Delta \defeq B/N_{\text{cr}}$. Then, the frequency-domain representation of \eqref{eq:general-radar-measurement-compact-discret} is
\begin{equation}\label{eq:general-radar-measurement-compact-freq}
    \begin{array}{cl}
        \rvecb y_p(n) &= \beta_0 \vec b(\theta_0) \vec a^\H(\theta_0) \rvecb s_p(n) e^{-j 2 \pi n \Delta \tau_0} e^{j 2 \pi p T_r \nu_0} + \\ 
        & \quad \quad \quad \quad \rvecb n_{c,p}(n),
    \end{array}
\end{equation}
where $n \in [N_{\text{cr}}]$ denotes the frequency-bin index, $\rvecb y_p$ the FFT of $\rvec y_p$, $\rvecb s_p$ the FFT of $\rvec s_{p}$, and $\rvecb n_{c,p}$ the FFT of $\rvec n_{c,p}$; the $n$-independent phase term $e^{-j 2 \pi [f_c - \frac{(N_{\text{cr}} - 1)}{2} \Delta] \tau_0}$ has been absorbed by $\beta_0$. In matched filtering for ranging, it is common that $N_{\text{cr}} = L$. However, the computational complexity in the time domain is $\cal O(L^2)$, while that in the frequency domain is $\cal O(L \log_2 (L))$ due to the FFT technique.

In signal-level target tracking (SLTR), \eqref{eq:general-radar-measurement-compact-discret} or \eqref{eq:general-radar-measurement-compact-freq} defines the state measurement equation in the tracking model \eqref{eq:target-trcking}. To be specific, $\rvec y_k$ is a data set constructed as
\begin{equation}\label{eq:measurements}
    \rvec y_k \defeq 
    \begin{cases}
        \{\rvec y_p(l)\}_{\forall p, l}, &\text{if time-domain processing}, \\
        \{\rvecb y_p(n)\}_{\forall p, n}, &\text{if frequency-domain processing},
    \end{cases}
\end{equation}
where the notational dependence of the right-hand-side terms on the tracking period index $k$ is omitted. 


Under the Bayesian statistical inference paradigm, the target tracking problem using the system model \eqref{eq:target-trcking}, in a specific time $k$, can be stated as finding the posterior distribution
\begin{equation}\label{eq:bayesian-tracking}
    q(\vec x_k | \cal Y_k) \propto q(\vec y_k | \vec x_k) q(\vec x_k | \cal Y_{k-1})
\end{equation}
of the unknown state vector $\rvec x_k$ by utilizing the likelihood function $\vec x_k \mapsto q(\vec y_k | \vec x_k)$ and the prior distribution $q(\vec x_k | \cal Y_{k-1})$; at every $k$, the unknown true state vector $\vec x_{0,k}$ is one-to-one corresponded to the true target parameters $(\tau_{0,k}, \nu_{0,k}, \theta_{0,k})$. It is preferable if $q(\vec x_k | \cal Y_k)$ can concentrate around the true state $\vec x_{0,k}$. The prior distribution $q(\vec x_k | \cal Y_{k-1})$ is induced by the state transition equation, while the likelihood function $q(\vec y_k | \vec x_k)$ is defined by the state measurement equation. In this section, for notational simplicity, we suppress the dependence on the time index $k$ and the measurement history $\cal Y_{k-1}$. Eq. \eqref{eq:bayesian-tracking} is therefore shorthanded as 
\begin{equation}\label{eq:bayesian-tracking-lite}
    q(\vec x | \vec y) \propto q(\vec y | \vec x) q(\vec x).
\end{equation}
Focusing on the target parameters, \eqref{eq:bayesian-tracking-lite} particularizes into 
\begin{equation}\label{eq:bayesian-tracking-lite-polar}
    q(\vec \varphi | \vec y) \propto q(\vec y | \vec \varphi) q(\vec \varphi),
\end{equation}
where $\vec \varphi \defeq (\tau, \nu, \theta) \in \R^3$ and the likelihood function $\vec \varphi \mapsto q(\vec y | \vec \varphi)$ is induced by either \eqref{eq:general-radar-measurement-compact-discret} or \eqref{eq:general-radar-measurement-compact-freq}; cf. \eqref{eq:measurements}.

\subsection{Problem Statements (Research Gaps in SLTR)}
In applying the Bayes' rule \eqref{eq:bayesian-tracking-lite-polar} to signal-level target tracking, the real-world operational dilemma is twofold:
\begin{itemize}
    \item The explicit and exact mathematical expression of the likelihood function $\vec \varphi \mapsto q(\vec y | \vec \varphi)$ is difficult to obtain because the probabilistic models of $\beta_0$, $\rvec n_{c, p}(l)$, and $\rvecb n_{c, p}(n)$ are challenging to specify in practice. Even if the probabilistic form of $q(\vec y | \vec \varphi)$ is (assumed to be) known, for example, Gaussian or exponential family, its statistical information, such as mean and covariance, is still to be meticulously calibrated by leveraging the domain knowledge and historical data. The complication arises if $q(\vec y | \vec \varphi)$ is time-varying across different tracking periods due to the significant maneuvers of the ToI and scatterers and to the change of channel conditions.

    \item To simplify the modeling of $q(\vec y | \vec \varphi)$, some statistical inference methods, such as conventional particle filters, treat $\beta_0$ as a nuisance parameter to be jointly estimated from the collected data \cite{ratpunpairoj2015particle,huleihel2013optimal}. This scheme leads to the augmented parameter vector $\vec \varphi \defeq (\opreal(\beta), \opimag(\beta), \tau, \nu, \theta) \in \R^5$ of the ToI. However, although convenient for modeling, this augmentation significantly introduces additional computational burdens to statistical inference because the larger the dimension of the state vector $\rvec x_k$ is, the lower the computational and sample efficiency of any statistical inference method would be. This is specifically the case for particle filters where more particles, and therefore more computational resources, must be employed to guarantee a satisfactory accuracy for $(\tau_0, \nu_0, \theta_0) \in \R^3$ \cite{doucet2001sequential,snyder2008obstacles}. Although the Rao-Blackwellized particle filtering can be used to reduce such computational burdens \cite{schon2005marginalized}, it is inadequate if $\beta_0$ abruptly varies across consecutive tracking periods $k$: Note that to track an unknown parameter, it should remain almost constant or slowly vary over time $k$ \cite[p.~397]{simon2006optimal}.
\end{itemize}

To address the above two issues, this article revolutionizes the existing particle filtering framework for signal-level target tracking using MIMO pulse-Doppler radar and proposes a new PF-SLTR scheme.

\section{A New Particle Filter for Signal-Level Target Tracking}\label{sec:particle-filter}

\subsection{Optimization-Centric Interpretation of Bayes' Rule \captext{\eqref{eq:bayesian-tracking-lite-polar}}}
In the statistical machine learning and general statistics communities, a well-established optimization-centric interpretation of the Bayes' rule \eqref{eq:bayesian-tracking-lite-polar} is as follows:
\begin{equation}\label{eq:bayes-rule-opt-view}
    q(\vec \varphi | \vec y) = \displaystyle \argmin_{u(\vec \varphi)} \E_{\vec \varphi \sim u(\vec \varphi)} [-\ln{q(\vec y | \vec \varphi)}] + \KL[{u(\vec \varphi) \| q(\vec \varphi)}],
\end{equation}
where $u(\vec \varphi)$ is a variable distribution of $\vec \varphi$ and $\KL[{u(\vec \varphi) \| q(\vec \varphi)}]$ defines the Kullback--Leibler divergence of $u(\vec \varphi)$ from $q(\vec \varphi)$. Eq. \eqref{eq:bayes-rule-opt-view} indicates that the posterior distribution $q(\vec \varphi | \vec y)$ of $\vec \varphi$, given $\vec y$, \textit{maximizes} the expected log-likelihood evaluated at $\vec y$, and simultaneously, \textit{minimizes} the dissimilarity from the prior distribution $q(\vec \varphi)$. This interpretation aligns well with our intuition: the posterior distribution should mostly match the \quotemark{data evidence} at hand (i.e., maximum likelihood), and at the same time, mostly match the \quotemark{prior knowledge} in the brain (i.e., minimum deviation from the prior).

Eq. \eqref{eq:bayes-rule-opt-view} can be generalized to the following formulation
\begin{equation}\label{eq:bayes-rule-opt-view-h}
    q(\vec \varphi | \vec y) \defeq \displaystyle \argmin_{u(\vec \varphi)} \E_{\vec \varphi \sim u(\vec \varphi)} \xi \rho(\vec y, \vec \varphi) + \KL[{u(\vec \varphi) \| q(\vec \varphi)}],
\end{equation}
where $\rho(\vec y, \vec \varphi)$ is a general cost function that encodes the $\vec y$-data evidence of $\vec \varphi$, beyond the well-adopted negative log-likelihood cost function $-\ln q(\vec y | \vec \varphi)$, and $\xi \ge 0$ is a tradeoff parameter that balances between $\vec y$-data evidence and $q(\vec \varphi)$-prior information. In \eqref{eq:bayes-rule-opt-view-h}, with a slight abuse of notation, we reuse $q(\vec \varphi \mid \vec y)$ to denote a generalized posterior distribution; cf. \eqref{eq:bayes-rule-opt-view}, where $q(\vec \varphi \mid \vec y)$ refers to the conventional Bayesian posterior. The generalization from \eqref{eq:bayes-rule-opt-view} to \eqref{eq:bayes-rule-opt-view-h} is rooted in the following two motivations:
\begin{itemize}
    \item Different applications prefer different utility metrics, so we can manipulate how the available data $\vec y$ supports the utility of the solution (or decision) $\vec \varphi$. The basic principle to design $\rho$ is that, if $\vec \varphi_1$ is favored over $\vec \varphi_2$ given $\vec y$, we should have $\rho(\vec y, \vec \varphi_1) \le \rho(\vec y, \vec \varphi_2)$. Illustrating instances in machine learning applications can be seen in \cite{germain2016pac}.

    \item In practice, the likelihood function $q(\vec y | \vec \varphi)$ may not be exactly specified, so we can steer the cost function $\rho$ to achieve robustness against this likelihood-model mismatch. Examples in this direction can be seen in \cite{knoblauch2022optimization}.
\end{itemize}

A desired property of \eqref{eq:bayes-rule-opt-view-h} is that it admits a closed-form solution for every cost function $\rho(\vec y, \vec \varphi)$ and every prior distribution $q(\vec \varphi)$, which is beneficial in real-time signal processing, such as online target tracking.

\begin{lemma}[see \cite{bissiri2016general}]\label{lemma:bayes-rule-h}
    The generalized posterior distribution $q(\vec \varphi | \vec y)$ solving \eqref{eq:bayes-rule-opt-view-h} is given by
    \begin{equation}\label{eq:Gibbs}
        q(\vec \varphi | \vec y) \propto q(\vec \varphi) e^{- \xi \rho(\vec y, \vec \varphi)},
    \end{equation}
     if the normalizer of the right-hand-side term exists. \stp
\end{lemma}

Eq. \eqref{eq:Gibbs} is called the \textit{Gibbs posterior} in the literature. Focusing back on the Bayesian filtering for signal-level target tracking (e.g., PF-SLTR), since the exact likelihood function $\vec \varphi \mapsto q(\vec y | \vec \varphi)$ is unavailable, we are motivated to construct an appropriate cost function $\rho(\vec y, \vec \varphi)$ to solve the two issues stated before, that is, without relying on the exact knowledge of the likelihood function $\vec \varphi \mapsto q(\vec y | \vec \varphi)$ and without augmenting nuisance parameters $(\opreal(\beta), \opimag(\beta))$ into $\vec \varphi$.

\subsection{A New Framework of PF-SLTR}
Particle filtering, which utilizes particle representations of probability density functions, is a computational method of Bayesian filtering \cite{godsill2019particle,arulampalam2002tutorial}. To be specific, in particle filtering, the prior distribution $q(\vec \varphi)$ is represented by a set of discrete particles $\{\vec \varphi_0, \vec \varphi_1, \ldots, \vec \varphi_{N_{\text{par}}-1}\}$:
\begin{equation}\label{eq:par-representation-prior}
    q(\vec \varphi) = \sum^{N_{\text{par}} - 1}_{i = 0} \eta_i \delta(\vec \varphi - \vec \varphi_i),
\end{equation}
where $\eta_i \in [0, 1]$ is the prior weight associated with the particle $\vec \varphi_i$, $\sum^{N_{\text{par}} - 1}_{i = 0} \eta_i = 1$, and $\delta(\cdot)$ is the Dirac delta distribution. Then, under the Gibbs rule \eqref{eq:Gibbs}, the posterior distribution $q(\vec \varphi | \vec y)$ given $\vec y$ can be represented as
\begin{equation}\label{eq:par-representation-posterior}
    q(\vec \varphi | \vec y) = \sum^{N_{\text{par}} - 1}_{i = 0} u_i \delta(\vec \varphi - \vec \varphi_i),
\end{equation}
where $u_i \in [0, 1]$ is the posterior weight associated with the particle $\vec \varphi_i$, $\sum^{N_{\text{par}} - 1}_{i = 0} u_i = 1$, and for every $i \in [N_{\text{par}}]$, we have
\begin{equation}\label{eq:par-representation-posterior-weight}
    u_i \propto \eta_i e^{- \xi \rho(\vec y, \vec \varphi_i)}.
\end{equation}

\begin{example}\label{eq:usual-particle-filter}
If the likelihood function can be exactly specified, i.e., $\rho(\vec y, \vec \varphi) \defeq -\ln q(\vec y | \vec \varphi)$, by setting $\xi \defeq 1$, \eqref{eq:par-representation-posterior-weight} particularizes to 
\begin{equation}\label{eq:par-representation-posterior-weight-usual}
    u_i \propto \eta_i \cdot q(\vec y | \vec \varphi_i) = \eta_i \cdot e^{-[-\ln q(\vec y | \vec \varphi_i)]},
\end{equation}
which is the usual particle filter. \stp
\end{example}

Eqs. \eqref{eq:par-representation-prior}-\eqref{eq:par-representation-posterior-weight} form the foundation of our new particle filtering framework for signal-level target tracking. The enabling yet challenging step is to design a cost function $\rho(\vec y, \vec \varphi)$ without relying on the true likelihood function $q(\vec y | \vec \varphi)$. In addition, this cost function satisfies the following rule.

\begin{myrule}\label{rule:cost-func-h}
At data $\vec y$, it holds that
\begin{equation}\label{eq:rule-1}
    \rho(\vec y, \vec \varphi_1) \le \rho(\vec y, \vec \varphi_2),
\end{equation}
or equivalently,
\begin{equation}\label{eq:rule-2}
e^{-\rho(\vec y, \vec \varphi_1)} \ge e^{-\rho(\vec y, \vec \varphi_2)},
\end{equation}
if $\vec \varphi_1$ is preferred over $\vec \varphi_2$. \stp
\end{myrule}

Aligning with Rule \ref{rule:cost-func-h}, in what follows, this article suggests four concrete cost functions for PF-SLTR, which are conceptually and operationally simple for MIMO pulse-Doppler radars described by \eqref{eq:general-radar-measurement}$-$\eqref{eq:measurements}.

\subsubsection{Cost Functions in PF-SLTR}
The first two cost functions are motivated by the least-squares estimation.

\begin{method}\label{method:cost-func-time-domain}
If we work with the time-domain signals in \eqref{eq:general-radar-measurement-compact-discret}, we define
    \begin{equation}\label{eq:cost-func-time-domain}
        \begin{array}{l}
            \rho_1(\vec y, \vec \varphi) \defeq \\ 
            \displaystyle \min_{\gamma \in \C} \sum^{N_p - 1}_{p = 0} \sum^{L - 1}_{l = 0} \Big\|\vec y_p(l) - \gamma \vec b(\theta) \vec a^\H(\theta) \vec s_p(l - \tau) e^{j 2 \pi p T_r \nu} \Big\|^2_2,
        \end{array}
    \end{equation}
where $\vec y$ is a realized sample of the radar received signal $\rvec y_k$ in \eqref{eq:measurements}, $\vec \varphi \defeq (\tau, \nu, \theta) \in \R^3$ is a hypothesized ToI-parameter vector, and $\vec y_p(l)$ is a realized sample of $\rvec y_p(l)$ in \eqref{eq:general-radar-measurement-compact-discret}. \stp
\end{method}

In $\rho_1(\vec y, \vec \varphi)$, the nuisance parameter $\gamma$ is marginalized out through the least-squares principle (i.e., the minimization operation). Different from Rao-Blackwellization, this operation can serve as an alternative approach to dimensionality reduction in particle filtering. The cost function $\rho_1(\vec y, \vec \varphi)$ can be computationally simplified.

\begin{proposition}\label{prop:h1}
If the transmit array has the empirical spatial correlation $\math R_s$, then $\rho_1(\vec y, \vec \varphi)$ equals
\begin{equation}\label{eq:h1-simplified}
    \begin{array}{l}
    \rho_1(\vec y, \vec \varphi) = 
    \displaystyle \sum^{N_p - 1}_{p = 0} \sum^{L - 1}_{l = 0} \vec y^\H_p(l) \vec y_p(l) \quad - \\ 
            \quad \frac{
            \displaystyle
            \left| \sum^{N_p - 1}_{p = 0} \sum^{L - 1}_{l = 0} \vec y^\H_p(l) \vec b(\theta) \vec a^\H(\theta) \vec s_p(l - \tau) e^{j 2 \pi p T_r \nu} \right|^2
            }{\displaystyle N_r N_p L_{\text{ss}} \vec a^\H(\theta) \math R_s \vec a(\theta)},
    \end{array}
\end{equation}
where $L_{\text{ss}}$ denotes the number of snapshots contained in the pulse $\{\vec s_p(l)\}_{\forall l \in [L_{\text{ss}}]}$, for every $p \in [N_p]$. When $\math R_s = P_t \mat I_{N_t}$, the denominator becomes $P_t N_t N_r N_p L_{\text{ss}}$, which is $\vec \varphi$-independent.
\end{proposition}
\begin{proof}
See Appendix A in the online supplementary materials. 
\stp
\end{proof}

\begin{method}\label{method:cost-func-freq-domain}
If we work with the frequency-domain signals in \eqref{eq:general-radar-measurement-compact-freq}, we define
    \begin{equation}\label{eq:cost-func-freq-domain}
        \begin{array}{cl}
            \rho_2(\vec y, \vec \varphi) &\defeq \displaystyle \min_{\gamma \in \C} \sum^{N_p - 1}_{p = 0} \sum^{N_{\text{cr}} - 1}_{n = 0} \Big\|\vecb y_p(n) - \Big.\\ 
            & \quad \quad \Big. \gamma \vec b(\theta) \vec a^\H(\theta) \vecb s_p(n) e^{-j 2 \pi n \Delta \tau} e^{j 2 \pi p T_r \nu} \Big\|^2_2,
        \end{array}
    \end{equation}
where $\vec y$ is a realized sample of the radar received signal $\rvec y_k$ in \eqref{eq:measurements}, $\vec \varphi \defeq (\tau, \nu, \theta) \in \R^3$ is a hypothesized ToI-parameter vector, and $\vecb y_p(n)$ is a realized sample of $\rvecb y_p(n)$ in \eqref{eq:general-radar-measurement-compact-freq}. \stp
\end{method}

Due to the minimization operation over $\gamma$, by employing $\rho_2(\vec y, \vec \varphi)$, dimensionality reduction in particle filtering is implicitly conducted; cf. Rao-Blackwellization. The cost function $\rho_2(\vec y, \vec \varphi)$ can be computationally simplified as follows.

\begin{proposition}\label{prop:h2}
If the transmit array has the empirical spatial correlation $\math R_s$, then $\rho_2(\vec y, \vec \varphi)$ equals
\begin{equation}\label{eq:h2-simplified}
    \begin{array}{l}
    \rho_2(\vec y, \vec \varphi) = 
    \displaystyle \sum^{N_p - 1}_{p = 0} \sum^{N_{\text{cr}} - 1}_{n = 0} \vecb y^\H_p(n) \vecb y_p(n) \quad - \\ 
            \quad \frac{
            \displaystyle
            \left| \sum^{N_p - 1}_{p = 0} \sum^{N_{\text{cr}} - 1}_{n = 0} \vecb y^\H_p(n) \vec b(\theta) \vec a^\H(\theta) \vecb s_p(n) e^{-j 2 \pi n \Delta \tau} e^{j 2 \pi p T_r \nu} \right|^2
            }{\displaystyle N_r N_p N_{\text{cr}} \vec a^\H(\theta) \math R_s \vec a(\theta)},
    \end{array}
\end{equation}
where $N_{\text{cr}}$ denotes the number of frequency bins (i.e., FFT points) to represent $\{\vecb s_p(n)\}_{\forall n \in [N_{\text{cr}}]}$, for every $p \in [N_p]$. 
When $\math R_s = P_t \mat I_{N_t}$, the denominator in the second line becomes $P_t N_t N_r N_p N_{\text{cr}}$, which is $\vec \varphi$-independent.
\end{proposition}
\begin{proof}
    Similar to the proof of Proposition \ref{prop:h1}. \stp 
\end{proof}

As we can see from Propositions \ref{prop:h1} and $\ref{prop:h2}$, to expect small values of $\rho_1(\vec y, \vec \varphi)$ and $\rho_2(\vec y, \vec \varphi)$ when given $\vec y$, it suffices to have large values of
    \begin{equation}\label{eq:h1-simplified-one-term}
        \begin{array}{l}
            \frac{
            \displaystyle
            \left| \sum^{N_p - 1}_{p = 0} \sum^{L - 1}_{l = 0} \vec y^\H_p(l) \vec b(\theta) \vec a^\H(\theta) \vec s_p(l - \tau) e^{j 2 \pi p T_r \nu} \right|^2
            }{\displaystyle N_r N_p L_{\text{ss}} \vec a^\H(\theta) \math R_s \vec a(\theta)}
        \end{array}
    \end{equation}
and
    \begin{equation}\label{eq:h2-simplified-one-term}
        \begin{array}{l}
            \frac{
            \displaystyle
            \left| \sum^{N_p - 1}_{p = 0} \sum^{N_{\text{cr}} - 1}_{n = 0} \vecb y^\H_p(n) \vec b(\theta) \vec a^\H(\theta) \vecb s_p(n) e^{-j 2 \pi n \Delta \tau} e^{j 2 \pi p T_r \nu} \right|^2
            }{\displaystyle N_r N_p N_{\text{cr}} \vec a^\H(\theta) \math R_s \vec a(\theta)},
        \end{array}
    \end{equation}
in \eqref{eq:h1-simplified} and \eqref{eq:h2-simplified}, respectively, because $\rho_1(\vec y, \vec \varphi) \ge 0$ and $\rho_2(\vec y, \vec \varphi) \ge 0$ naturally hold for all $(\vec y, \vec \varphi)$. The two numerators are exactly the ambiguity functions (i.e., the matched filters) for MIMO pulse-Doppler radars in the time domain and frequency domain, respectively \cite[Eq.~(33)]{san2007mimo}, \cite[p.~131]{michael2013mimo}. However, the two denominators introduce modifiers that depend on the spatial
correlation matrix $\math R_s$ and the DoA $\theta$.

In what follows, for brevity in presentation and without loss of generality, we assume that $\math R_s \defeq P_t \mat I_{N_t}$. That is, the transmit waveforms are uncorrelated across antennas. In this case, it suffices to have large values for the numerators in \eqref{eq:h1-simplified-one-term} and \eqref{eq:h2-simplified-one-term}, respectively, because the two denominators degenerate to constants. This observation inspires the next two cost functions.

\begin{method}\label{method:cost-func-time-domain-AF}
If we work with the time-domain signals in \eqref{eq:general-radar-measurement-compact-discret}, we define
    \begin{equation}\label{eq:cost-func-time-domain-AF}
        \begin{array}{l}
            \rho_3(\vec y, \vec \varphi) \defeq \\ 
            \quad \displaystyle - \ln \left[\left| \sum^{N_p - 1}_{p = 0} \sum^{L - 1}_{l = 0} \vec y^\H_p(l) \vec b(\theta) \vec a^\H(\theta) \vec s_p(l - \tau) e^{j 2 \pi p T_r \nu} \right|^2\right],
        \end{array}
    \end{equation}
where $\vec y$ is a realized sample of the radar received signal $\rvec y_k$ in \eqref{eq:measurements}, $\vec \varphi \defeq (\tau, \nu, \theta) \in \R^3$ is a hypothesized ToI-parameter vector, and $\vec y_p(l)$ is a realized sample of $\rvec y_p(l)$ in \eqref{eq:general-radar-measurement-compact-discret}. \stp
\end{method}


Using Method \ref{method:cost-func-time-domain-AF}, \eqref{eq:par-representation-posterior-weight} becomes  
\begin{equation}\label{eq:Gibbs-par-weight-time-AF}
    u_i \propto \eta_i \cdot \left| \sum^{N_p - 1}_{p = 0} \sum^{L - 1}_{l = 0} \vec y^\H_p(l) \vec b(\theta_i) \vec a^\H(\theta_i) \vec s_p(l - \tau_i) e^{j 2 \pi p T_r \nu_i} \right|^{2\xi},
\end{equation}
where $(\tau_i, \nu_i, \theta_i)$ denotes the $i^\th$ particle. Eq. \eqref{eq:Gibbs-par-weight-time-AF} aligns well with our intuition because, the closer the parameter pair $(\tau_i, \nu_i, \theta_i)$ is to its true value $(\tau_0, \nu_0, \theta_0)$, the larger the value of the ambiguity function evaluated at $(\tau_i, \nu_i, \theta_i)$, and thus the more the collected radar data $\vec y$ supports the $i^\th$ particle. If $\math R_s \neq P_t \mat I_{N_t}$, the denominator ${\displaystyle N_r N_p L_{\text{ss}} \vec a^\H(\theta) \math R_s \vec a(\theta)}$ should be retained in the right-hand side of \eqref{eq:Gibbs-par-weight-time-AF}.

\begin{method}\label{method:cost-func-freq-domain-AF}
If we work with the frequency-domain signals in \eqref{eq:general-radar-measurement-compact-freq}, we define
    \begin{equation}\label{eq:cost-func-freq-domain-AF}
        \begin{array}{cl}
            \rho_4(\vec y, \vec \varphi) &\defeq \displaystyle - \ln \Bigg[\Bigg| \sum^{N_p - 1}_{p = 0} \sum^{N_{\text{cr}} - 1}_{n = 0} \Bigg. \Bigg. \\
            & \quad \Bigg.\Bigg.\vecb y^\H_p(n)  \vec b(\theta) \vec a^\H(\theta) \vecb s_p(n) e^{-j 2 \pi n \Delta \tau} e^{j 2 \pi p T_r \nu} \Bigg|^2\Bigg],
        \end{array}
    \end{equation}
where $\vec y$ is a realized sample of the radar received signal $\rvec y_k$ in \eqref{eq:measurements}, $\vec \varphi \defeq (\tau, \nu, \theta) \in \R^3$ is a hypothesized ToI-parameter vector, and $\vecb y_p(n)$ is a realized sample of $\rvecb y_p(n)$ in \eqref{eq:general-radar-measurement-compact-freq}. \stp
\end{method}


Using Method \ref{method:cost-func-freq-domain-AF}, \eqref{eq:par-representation-posterior-weight} becomes  
\begin{equation}\label{eq:Gibbs-par-weight-freq-AF}
    \begin{array}{rl}
        u_i &\propto \eta_i \cdot \Bigg| \displaystyle \sum^{N_p - 1}_{p = 0} \sum^{N_{\text{cr}} - 1}_{n = 0} \Bigg. \Bigg. \\
        & \quad \quad \quad \displaystyle \Bigg.\Bigg.\vecb y^\H_p(n)  \vec b(\theta_i) \vec a^\H(\theta_i) \vecb s_p(n) e^{-j 2 \pi n \Delta \tau_i} e^{j 2 \pi p T_r \nu_i} \Bigg|^{2 \xi}.
    \end{array}
\end{equation}
Similarly, the closer the parameter pair $(\tau_i, \nu_i, \theta_i)$ is to its true value $(\tau_0, \nu_0, \theta_0)$, the larger the value of the ambiguity function evaluated at $(\tau_i, \nu_i, \theta_i)$, and thus the more the collected radar data $\vec y$ supports the $i^\th$ particle. If $\math R_s \neq P_t \mat I_{N_t}$, the denominator ${\displaystyle N_r N_p N_{\text{cr}} \vec a^\H(\theta) \math R_s \vec a(\theta)}$ should be retained in the right-hand side of \eqref{eq:Gibbs-par-weight-freq-AF}.

The cost functions $\rho_3(\vec y, \vec \varphi)$ and $\rho_4(\vec y, \vec \varphi)$ are conceptually consistent with $\rho_1(\vec y, \vec \varphi)$ and $\rho_2(\vec y, \vec \varphi)$ when $\math R_s = P_t \mat I_{N_t}$; cf. Propositions \ref{prop:h1} and $\ref{prop:h2}$. However, $\rho_3(\vec y, \vec \varphi)$ and $\rho_4(\vec y, \vec \varphi)$ are computationally more lightweight than $\rho_1(\vec y, \vec \varphi)$ and $\rho_2(\vec y, \vec \varphi)$, especially when implementing the Gibbs-posterior-based PF-SLTR; see \eqref{eq:Gibbs-par-weight-time-AF} and \eqref{eq:Gibbs-par-weight-freq-AF}. To limit the computational complexity, hereafter, this article focuses on the cost functions $\rho_3(\vec y, \vec \varphi)$ and $\rho_4(\vec y, \vec \varphi)$. In this sense, hence, the cost functions $\rho_1(\vec y, \vec \varphi)$ and $\rho_2(\vec y, \vec \varphi)$ in this article are more analytically motivational than practically instructive: They motivate the proposals of $\rho_3(\vec y, \vec \varphi)$ and $\rho_4(\vec y, \vec \varphi)$, respectively.

\subsubsection{Property Analysis}
In this subsection, we analyze the properties of the proposed cost functions. As elucidated before, we are particularly interested in $\rho_3(\vec y, \vec \varphi)$ and $\rho_4(\vec y, \vec \varphi)$ for their computational advantages. Since $\rho_3(\vec y, \vec \varphi)$ and $\rho_4(\vec y, \vec \varphi)$ are mathematically equivalent under FFT and inverse FFT (IFFT), it suffices to analyze the property of $\rho_4(\vec y, \vec \varphi)$ only. The proposition below shows that $\rho_4(\vec y, \vec \varphi)$ aligns with Rule \ref{rule:cost-func-h} in a well-specified tracking gate. A tracking gate (or validation gate) is a hypothesized region in which the ToI is assumed to be located; that is, it is a delimited subspace of $\R^3$ in which the true target parameters $(\tau_0, \nu_0, \theta_0)$ are included.

\begin{proposition}\label{prop:h4-property}
    Suppose that the transmitted signals are spatially uncorrelated with $\mat R_s \defeq \E \rvecb s_p(n) \rvecb s^\H_p(n) = P_t \mat I_{N_t}$, for every $p \in [N_p]$ and $n \in [N_{\text{cr}}]$, and scatterers are well-separated from the ToI.\footnote{Well-separation means that the energy from scatterers that leaks into the ToI's tracking gate through the matched filter’s sidelobes is negligible.} 
    Additionally, we suppose that the transmitter and the receiver are uniform linear arrays with half-wavelength spacing: that is,
    \begin{equation}\label{eq:str-vec-a}
        \vec a (\theta) \defeq [1, e^{-j \pi \sin \theta}, \ldots, e^{-j \pi (N_{t} - 1) \sin \theta}]^\T,
    \end{equation}
    and
    \begin{equation}\label{eq:str-vec-b}
        \vec b (\theta) \defeq [1, e^{-j \pi \sin \theta}, \ldots, e^{-j \pi (N_{r} - 1) \sin \theta}]^\T.
    \end{equation}
    Let $N_{tr} \defeq \max\{N_t, N_r\}$. Then, within the tracking gate
    \begin{equation}\label{eq:tracking-gate}
        \begin{array}{l}
            \displaystyle \Big[\tau_0 - \frac{1}{B},~ \tau_0 + \frac{1}{B}\Big] \times \Big[\nu_0 - \frac{1}{T_i},~ \nu_0 + \frac{1}{T_i}\Big] \times \\
            
            \quad \displaystyle  \Big[\arcsin\Big(\sin \theta_0 - \frac{2}{N_{tr}}\Big),~ \arcsin\Big(\sin \theta_0 + \frac{2}{N_{tr}}\Big)\Big],
        \end{array}
    \end{equation}
    the cost function 
    \begin{equation}\label{eq:cost-func-freq-domain-AF-expected}
        \begin{array}{l}
        -\ln \Bigg[\Bigg| \displaystyle  \sum^{N_p - 1}_{p = 0} \sum^{N_{\text{cr}} - 1}_{n = 0} \E_{\rvecb y_p(n), \rvecb s_p(n)} \bigg[ \Bigg. \Bigg. \bigg. \\
         \quad \quad \quad \quad \Bigg. \Bigg. \bigg. \rvecb y^\H_p(n) \vec b(\theta) \vec a^\H(\theta) \rvecb s_p(n) e^{-j 2 \pi n \Delta \tau} e^{j 2 \pi p T_r \nu} \bigg] \Bigg|^2 \Bigg],
        \end{array}
    \end{equation}
    induced by $\rho_4(\vec y, \vec \varphi)$ in \eqref{eq:cost-func-freq-domain-AF}, through inner-expectation over $\rvecb y_p(n)$ and $\rvecb s_p(n)$, satisfies Rule \ref{rule:cost-func-h}.
\end{proposition}
\begin{proof}
    See Appendix B in the online supplementary materials.
    \stp
\end{proof}

Due to the randomness of $\rvecb y_p(n)$ and $\rvecb s_p(n)$, we resort to the expected version in \eqref{eq:cost-func-freq-domain-AF-expected} for property analysis, to be specific, the analysis on the statistical property instead of the single-sample property. In the practice of PF-SLTR, however, \eqref{eq:cost-func-freq-domain-AF} should be used to update the weights of the particles; see \eqref{eq:Gibbs-par-weight-freq-AF}. 

In addition, the assumptions in Proposition \ref{prop:h4-property} are introduced primarily to obtain a clean statistical characterization of the proposed cost function. The separated-scatterer assumption means that the leakage from other scatterers into the ToI’s tracking gate is negligible, so that the ambiguity-function response around the ToI is dominated by the ToI component. If this condition is not satisfied, nearby scatterers may distort the cost-function landscape and mislead the tracking filter. Similarly, the uncorrelated-waveform assumption $\mat R_s \defeq P_t \mat I_{N_t}$ is used to simplify the cost functions. To be specific, if $\mat R_s \neq P_t \mat I_{N_t}$, the corresponding denominators involving $\vec a^\H(\theta)\math R_s\vec a(\theta)$ should be retained, as in \eqref{eq:h1-simplified-one-term} and \eqref{eq:h2-simplified-one-term}. In summary, when $\mat R_s$, $\vec a$, and $\vec b$ take more complicated forms, the tracking gate and cost functions should be calculated accordingly. For clarity of statistical analysis, Proposition \ref{prop:h4-property} is established under these controlled assumptions to provide an interpretable analytical result.
 



Next, we investigate the role of the parameter $\xi$ in \eqref{eq:Gibbs-par-weight-time-AF} and \eqref{eq:Gibbs-par-weight-freq-AF}. Let $\vec \alpha \defeq [\alpha_0, \alpha_1, \ldots, \alpha_{N_{\text{par}}-1}] > 0$ denote the $N_{\text{par}}$-dimensional data-evidence vector of particles, which is normalized so that its sum equals one; hence, $\vec \alpha$ can be regarded as a discrete probability vector. Taking \eqref{eq:Gibbs-par-weight-time-AF} as an example, up to a normalizing constant, we can define
\[
\alpha_i \propto \left| \sum^{N_p - 1}_{p = 0} \sum^{L - 1}_{l = 0} \vec y^\H_p(l) \vec b(\theta_i) \vec a^\H(\theta_i) \vec s_p(l - \tau_i) e^{j 2 \pi p T_r \nu_i} \right|^2,
\]
resulting in the following rule for generating posterior weights:
\begin{equation}\label{eq:posterior-weights}
u_i \propto \eta_i \alpha^\xi_i, ~~~ \forall i \in [N_{\text{par}}].
\end{equation}
Let $\vec \alpha^{(\xi)} \defeq [\alpha^{\xi}_0, \alpha^{\xi}_1, \ldots, \alpha^{\xi}_{N_{\text{par}}-1}]/C_{\xi}$ denote a new distribution powered by $\xi \ge 0$ from $\vec \alpha$, where $C_{\xi} \defeq \sum^{N_{\text{par}}-1}_{i = 0} \alpha^{\xi}_i$ is the normalization factor. We have the following result regarding the entropy values of $\vec \alpha$ and $\vec \alpha^{(\xi)}$. Note that the larger (resp. smaller) the entropy value of a discrete distribution, the more uniform (resp. polarized) this distribution is.

\begin{theorem}\label{thm:role-xi}
    It holds that
    \[
        - \displaystyle \sum^{N_{\text{par}} - 1}_{i = 0} \alpha_i \ln \alpha_i
        \begin{cases}
           \le - \displaystyle \sum^{N_{\text{par}} - 1}_{i = 0} \alpha^{(\xi)}_i \ln \alpha^{(\xi)}_i,     &\text{if } 0 \le \xi \le 1; \\

           \ge - \displaystyle \sum^{N_{\text{par}} - 1}_{i = 0} \alpha^{(\xi)}_i \ln \alpha^{(\xi)}_i,     &\text{if } \xi \ge 1.
        \end{cases}
    \]
\end{theorem}       
\begin{proof}
    See Appendix C in the online supplementary materials, where the entropy function $\xi \mapsto - \sum^{N_{\text{par}} - 1}_{i = 0} \alpha^{(\xi)}_i \ln \alpha^{(\xi)}_i$ is shown to be monotonically decreasing in $\xi \ge 0$. \stp
\end{proof}

Theorem \ref{thm:role-xi} implies that when $0 \le \xi < 1$, we make the $\vec y$-data evidence (i.e., the  powered evidence values $\vec \alpha^{(\xi)}$) more balanced across all particles $\{\vec x_i\}_{\forall i \in [N_{\text{par}}]}$. 
In contrast, when $\xi > 1$, we let the $\vec y$-data evidence be more concentrated on particles that have large evidence values.  
An illustrative example is given below.

\begin{example}\label{ex:role-xi}
Let $\vec \alpha \defeq [0.2, 0.5, 0.3]$. If $\xi = 0.5$, we have $\vec \alpha^{(0.5)} = [0.2628, 0.4154, 0.3218]$; if $\xi = 2$, we have $\vec \alpha^{(2)} = [0.1053, 0.6579, 0.2368]$. As indicated, $\vec \alpha^{(0.5)}$ is more balanced (i.e., having smaller differences) than $\vec \alpha$, while $\vec \alpha^{(2)}$ is more polarized (i.e., having larger differences) than $\vec \alpha$.
\stp
\end{example}

Since the conventional Bayes' rule is $u_i \propto \eta_i \alpha_i$, for all $i \in [N_\text{par}]$, where $\{\alpha_i\}$ denotes the likelihood and $\{\eta_i\}$ the prior, the posterior $\{u_i\}$ is shaped by the relative dispersion of these two terms. To be specific, if the likelihood $\{\alpha_i\}$ is more balanced (i.e., more uniform) than the prior $\{\eta_i\}$, the posterior remains closer to the prior. In the limiting case where all $\{\alpha_i\}$ are equal, the likelihood carries no information and does not impact the posterior. In contrast, if the likelihood $\{\alpha_i\}$ is more concentrated (i.e., more polarized) than the prior, the posterior is pulled toward the likelihood. In the limiting case where one of $\{\alpha_i\}$ equals one and all others are zeros, the posterior is entirely determined by the likelihood. Therefore, in the modified Bayes' rule \eqref{eq:posterior-weights}, i.e., $u_i \propto \eta_i \alpha^\xi_i$, $\xi$ plays a key role in balancing the influence of the prior and the data evidence, as it directly controls the dispersion of the data-evidence vector. This dispersion-based interpretation provides an alternative perspective on the role of $\xi$ in trading off between $\{\eta_i\}$ and $\{\alpha_i\}$, distinct from but complementary to the optimization-based interpretation in \eqref{eq:bayes-rule-opt-view-h}.

On the other hand, balanced prior and data-evidence vectors lead to balanced posterior weights, which in turn can benefit combating particle degeneracy, provided that particles are diverse and informative. In particle filtering, to ensure diversity and informativity, particles should be drawn from a well-designed proposal distribution \cite[pp.~178-180]{arulampalam2002tutorial}.

\section{PF-SLTR for MIMO-OFDM ISAC}\label{sec:MIMO-OFDM-ISAC}

The proposed time structure of OFDM ISAC waveforms for PF-SLTR is shown in Fig. \ref{fig:OFDM_ISAC_frame}, which applies to every transmit antenna in MIMO. A target tracking period $T_t$, which is usually in dozens of milliseconds (ms; e.g., 25ms, 50ms), contains one ISAC period $T_i$ to sense the moving ToI. Tens or hundreds of microseconds ($\mu$s) are, in practice, sufficient for an ISAC period. To save computing resources and without compromising the necessary accuracy of target tracking, there is no need to use the entire $T_t$ time window for radar sensing. The remaining $T_t - T_i$ seconds can be employed for pure communications to improve, e.g., throughput and reliability. Depending on how we use an MIMO-OFDM ISAC station, we can employ two types of waveform structures:
\begin{itemize}
    \item \textbf{Pulsed Waveforms}: In this scheme, silent times are included between consecutive pulses. The transmit array and the receive array operate in either a time-division duplex or full-duplex manner, corresponding to the non-blind minimum round-trip delay of $T_p$ or $0$, respectively. During these silent times, the radar receiver can listen to \textit{far} target echoes and resolve target parameters. This scheme is identical to that of the typical pulse-Doppler radars: The maximum detectable range is determined by the duration of the silence (i.e., $T_r - T_p$) in the time-division duplex mode or by the pulse repetition interval $T_r$ in the full-duplex mode, where $T_p$ is the pulse duration. Within each pulse duration, several OFDM symbols, together with their cyclic prefixes (CPs), are arranged back-to-back. This is a \textit{sensing-centric} solution where the ranging capacity is improved by adding silent times and sacrificing the communications throughput.

    \item \textbf{Continuous-Wave (CW) Waveforms}: In this scheme, no silent times are contained between consecutive pulses. The transmit array and the receive array are supposed to be well isolated, and they operate in a full-duplex manner. The maximum resolvable range is determined by the length of the CP, which ensures that the \textit{near} ToI can be treated as a scatterer. As a result, the CP length fully covers all possible path delays so that the circular convolution property is preserved for FFT-based signal processing, for example, frequency-domain channel estimation for communications and matched filtering for ranging. This is a \textit{communications-centric} solution where the CP duration limits the ranging capacity, while the communication data streams remain unaffected.
\end{itemize}

\begin{figure}[!htbp]
    \centering

    \subfigure[An ISAC period $T_i$ is part of a tracking time window $T_t$]{
        \centering
        \includegraphics[width=7.55cm]{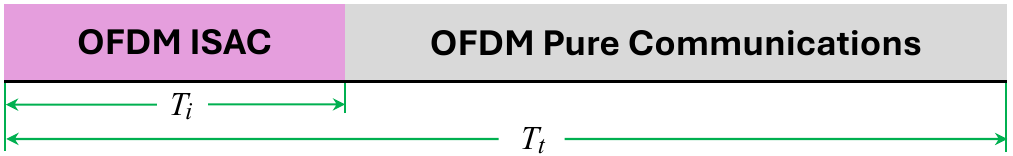}
    }

    \subfigure[Pulsed: A CPI $T_i$ with silent times defines an ISAC period]{
        \centering
        \includegraphics[width=7.5cm]{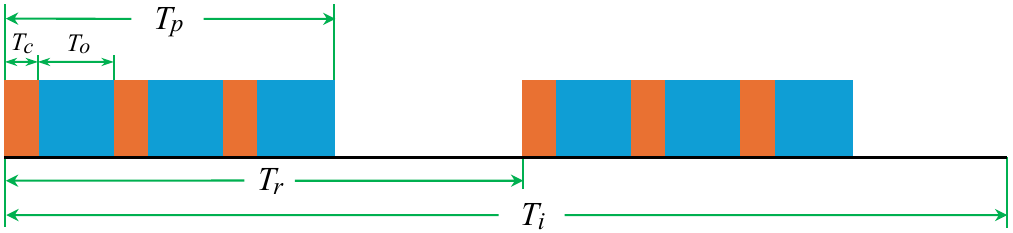}
    }

    \subfigure[CW: A CPI $T_i$ without silent times defines an ISAC period]{
        \centering
        \includegraphics[width=7.5cm]{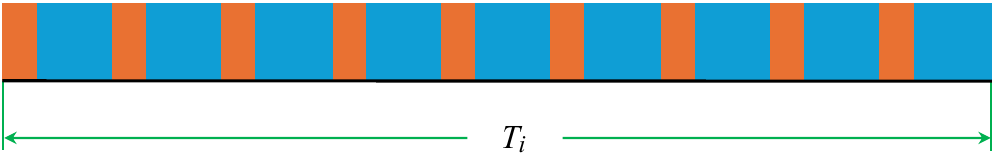}
    }
    
    \caption{The time structure of OFDM ISAC waveforms. Cyclic prefixes (CPs) are visualized by orange rectangles, while data payloads are indicated by blue rectangles. (b): pulsed OFDM ISAC waveforms; (c): continuous-wave (CW) OFDM ISAC waveforms. The tracking time period is $T_t$, and the ISAC period is $T_i$. For pulsed waveforms in (b), within a coherent processing interval (CPI) $T_i$, two pulses repeat with the pulse repetition interval (PRI) $T_r$. Each pulse has the time duration $T_p$. The silent time between two consecutive pulses is therefore $T_r - T_p$. Within each pulse, three CP-appended OFDM symbols are included; the duration of each OFDM symbol is $T_o$, whereas that of each CP is $T_c$. For CW waveforms in (c), we can treat $T_r = T_p$ so that no silent times are inserted. The Doppler processing for both structures is algorithmically similar.} 
    \label{fig:OFDM_ISAC_frame}
\end{figure}

Note that in the pulsed case, the circular convolution property is not inherently preserved for sensing signal processing, because the ToI, as a scatterer in the sensing channel, may not fall within the CP duration. Therefore, FFT-based frequency-domain methods require additional treatment compared to the CW case, as to be elaborated on in the following subsections. However, the scattering effect of the ToI may not impact the communication channel, since in the pulsed case, targets are typically far from the station, whereas communication users are located much closer. If the sensing ToI and communication users are both nearby, the CW scheme—where the CP can cover all path delays—should be adopted.

\subsection{Pulsed MIMO-OFDM ISAC}\label{subsec:pulsed-MIMO-OFDM}
In this mode, MIMO-OFDM ISAC works like a pulsed radar, and the only distinction is that the OFDM waveforms are employed at transmit antennas, instead of the canonical radar waveforms (e.g., phase-coded and chirp signals). Suppose that $M$ OFDM symbols are included in each pulse, and the number of subcarriers in OFDM is $N_c$. We have $T_p = M \cdot (T_c + T_o)$ and $T_o = 1/\Delta_f = N_c/B = N_c T_s$, where $\Delta_f$ is the subcarrier spacing. Let $\rvec c_{p,m}(n) \in \C^{N_t}$ denote the constellation points on the $n^\th$ subcarrier in the $p^\th$ pulse and $m^\th$ OFDM symbol, where $p \in [N_p]$, $m \in [M]$, and $n \in [N_c]$. The time-domain waveform $\rvec s_p(l)$ to be transmitted in the $p^\th$ pulse is
\begin{equation}\label{eq:pulsed-OFDM-tx-signal}
    \rvec s_p(l) = 
    \begin{cases}
        \rvec s_{p, \text{eff}}(l),     &\textbf{ if } 0 \le l \le \round{T_p/T_s} - 1, \\
        \vec 0, &\textbf{ if } \round{T_p/T_s} \le l \le \round{T_r / T_s} - 1,
    \end{cases}
\end{equation}
where $\rvec s_{p, \text{eff}}(l)$ for $0 \le lT_s \le T_p$ is the effective (i.e., non-zero) transmit waveform at the logical time $l$ in the $p^\th$ pulse, consisting of $M$ back-to-back CP-appended OFDM time-domain symbols; $\{\rvec s_{p, \text{eff}}(l)\}_{l:~0 \le l \le \round{T_p/T_s}-1}$ is determined by $\{\rvec c_{p,m}(n)\}_{\forall p \in [N_p], m \in [M], n \in [N_c]}$ via IFFT and CP addition; see, e.g., \cite[Eq.~(3)]{zhang2024input}, \cite[Eq.~(2)]{sturm2011waveform}, \cite[Eq.~(1)]{koivunen2024multicarrier}. Note that within the time window $[T_p, T_r]$, the transmitter keeps silent. Note also that $N_{\text{cr}} \ne N_c$: the former is the FFT length of the whole transmit signal $\{\rvec s_{p}(l)\}_{\forall l \in [L]}$ in which multiple CP-added OFDM symbols and a silent time slot are contained [cf. \eqref{eq:general-radar-measurement-compact-freq} and \eqref{eq:cost-func-freq-domain-AF}], while the latter is the FFT length of one single CP-excluded OFDM symbol. In the practice of pulsed radars, $N_c \ll N_{\text{cr}} = L$.

The transmit MIMO-OFDM time-domain signals $\rvec s_p$ defined in \eqref{eq:pulsed-OFDM-tx-signal} are then propagated through the MIMO channel \eqref{eq:general-radar-measurement-compact-discret} and received by the radar receive array as $\{\rvec y_p(l)\}_{\forall l \in [L]}$. The data evidence for a hypothesized ToI parameter pair $\vec \varphi \defeq (\tau, \nu, \theta) \in \R^3$ is quantified by the cost functions in Methods \ref{method:cost-func-time-domain} and \ref{method:cost-func-time-domain-AF}, to be fed into the Gibbs-posterior-based PF-SLTR \eqref{eq:par-representation-posterior-weight}. If frequency-domain Methods \ref{method:cost-func-freq-domain} and \ref{method:cost-func-freq-domain-AF} are to be used, FFT from $\{\rvec s_{p}(l)\}_{\forall l \in [L]}$ to $\{\rvecb s_{p}(n)\}_{\forall n \in [N_{\text{cr}}]}$ must be explicitly conducted \textit{without removing CPs}, because the circular convolution property of FFT was destroyed, so that CPs are also part of effective sensing waveforms and $\{\rvecb s_{p}(n)\}_{\forall n \in [N_{\text{cr}}]}$ do not equal $\{\rvec c_{p,m}(n)\}_{\forall m \in [M], n \in [N_c]}$.


\subsection{Continuous-Wave MIMO-OFDM ISAC}\label{subsec:CW-MIMO-OFDM}
The pulse-Doppler radar signal processing can also be applied to CW MIMO-OFDM ISAC systems, just treating $T_r = T_p = M \cdot (T_c + T_o)$. Because we assumed that the maximum target round-trip delay is covered by the CP duration, the circular convolution property of the FFT can be preserved. As a result, efficient frequency-domain processing is applicable. To be specific, \textit{after CP removal}, the MIMO channel propagation can be written as
\begin{equation}\label{eq:CW-OFDM-Waveform}
    \begin{array}{cl}
        \rvecb y_{p,m}(n) &= \beta_0 \vec b(\theta_0) \vec a^\H(\theta_0) \rvec c_{p,m}(n) e^{-j 2 \pi n \Delta_f \tau_0} e^{j 2 \pi p T_r \nu_0} + \\ 
        & \quad \quad \quad \quad \rvecb n_{c,p,m}(n),
    \end{array}
\end{equation}
for every $p \in [N_p]$, $m \in [M]$, and $n \in [N_c]$ (not $[N_{\text{cr}}]$). Note that within the same pulse, the Doppler frequency shift is assumed to be constant so that the Doppler phase term $e^{j 2 \pi p T_r \nu_0}$ does not depend on $m$. Therefore, the data evidence for a hypothesized ToI parameter pair $\vec \varphi \defeq (\tau, \nu, \theta) \in \R^3$ is quantified by the cost functions in frequency-domain Methods \ref{method:cost-func-freq-domain} and \ref{method:cost-func-freq-domain-AF}. Note that $\rho_2$ and $\rho_4$ need to be slightly modified to encompass the summations over $m \in [M]$ and $n \in [N_c]$:
    \begin{equation}\label{eq:cost-func-freq-domain-2}
        \begin{array}{cl}
            \rho^\prime_2(\vec y, \vec \varphi) &\defeq \displaystyle \min_{\vec \gamma \in \C} \sum^{N_p - 1}_{p = 0} \sum^{M - 1}_{m = 0}  \sum^{N_c - 1}_{n = 0} \Big\|\vecb y_{p,m}(n) - \Big.\\ 
            & \quad \quad \Big. \vec \gamma \vec b(\theta) \vec a^\H(\theta) \rvec c_{p,m}(n) e^{-j 2 \pi n \Delta_f \tau} e^{j 2 \pi p T_r \nu} \Big\|^2_2,
        \end{array}
    \end{equation}
and
    \begin{equation}\label{eq:cost-func-freq-domain-AF-2}
        \begin{array}{cl}
            \rho^\prime_4(\vec y, \vec \varphi) &\defeq \displaystyle - \ln \Bigg[\Bigg| \sum^{N_p - 1}_{p = 0} \sum^{M - 1}_{m = 0} \sum^{N_c - 1}_{n = 0} \Bigg. \Bigg. \\
            & \Bigg.\Bigg.\vecb y^\H_{p,m}(n)  \vec b(\theta) \vec a^\H(\theta) \rvec c_{p,m}(n) e^{-j 2 \pi n \Delta_f \tau} e^{j 2 \pi p T_r \nu} \Bigg|^2\Bigg].
        \end{array}
    \end{equation}

Numerically evaluating \eqref{eq:cost-func-freq-domain-2} and \eqref{eq:cost-func-freq-domain-AF-2} in the frequency domain, especially when the FFT technique is applied, can be significantly faster than evaluating \eqref{eq:cost-func-time-domain} and \eqref{eq:cost-func-time-domain-AF} in the time domain. This highlights the computational benefit of the CW MIMO-OFDM ISAC systems, however, coming with a limited ranging capability set by the CP duration. In contrast, pulsed MIMO-OFDM ISAC systems can support much larger ranging regions without considering the time length of CPs, however, at the cost of a higher computational burden\footnote{Although this burden can be algorithmically reduced by FFT tricks; cf. \eqref{eq:cost-func-time-domain-AF} and \eqref{eq:cost-func-freq-domain-AF}; NB: in practice, $M N_c \le N_{\text{cr}}$ and $N_c \ll N_{\text{cr}}$.} and lower communications throughput. 

\subsection{Overall PF-SLTR Algorithm}
In this section, we summarize the PF-SLTR for MIMO-OFDM ISAC systems in Algorithm \ref{algo:PF-SLTR}, which is adapted from the canonical particle filter for dynamic systems \eqref{eq:target-trcking}; see \cite[Algorithm 4]{arulampalam2002tutorial}, \cite[p.~468]{simon2006optimal}. Some algorithmic details and operational remarks, such as the basics of particle filtering and the notions of resampling and effective sample size, can be seen in \cite{arulampalam2002tutorial}, \cite[Chap.~15]{simon2006optimal}, \cite{godsill2019particle}. For extensive discussions on PF-SLTR, see Section \ref{sec:discussions}.

\begin{algorithm}[htbp]
    \caption{PF-SLTR in MIMO-OFDM ISAC}
    \label{algo:PF-SLTR}
    \begin{flushleft}
        \justifying
        \textbf{Definition}: $k$ is the discrete tracking-time index; $N_{\text{par}}$ is the number of state particles; $\vec x_{k, i}$ is the posterior state particles at $k$ and $u_{k,i}$ the associated weights, $\forall i \in [N_{\text{par}}]$; $\vec \varphi_{k,i}$ is the range-Doppler-angle representation of $\vec x_{k, i}$; $N_{\text{eff}}$ is the effective sample size and $N_{\text{thres}}$ its threshold. \\
        \textbf{Remark}: For illustration and computational efficiency, we only show frequency-domain processing using \eqref{eq:cost-func-freq-domain-AF} and \eqref{eq:cost-func-freq-domain-AF-2}. In addition, we suppose that $\math R_s \defeq P_t \mat I_{N_t}$; otherwise, cost functions can be modified by incorporating denominators as in \eqref{eq:h1-simplified-one-term} and \eqref{eq:h2-simplified-one-term}.\\
        \textbf{Initialization}: $N_{\text{par}}$, $N_{\text{thres}}$, $\xi$, and $\{\vec x_{0,i}, u_{0,i}\}_{i \in [N_{\text{par}}]}$.
    \end{flushleft}
    \begin{algorithmic}[1]
        \Require Radar received signals $\bm y_k,~k=1,2,3,...$; see \eqref{eq:measurements}
        \State // \textit{Step 1: Generate Prior State Particles; See \eqref{eq:target-trcking}}
        \For {$i=0:N_{\text{par}}-1$}
            \State Sample $\vec w_{k-1, i}$ from the distribution of $\rvec w_{k-1}$
            \State $\vec x_{k, i} \leftarrow \bm f_k (\vec x_{k-1, i}, \vec w_{k-1,i})$
            \State Translate $\vec x_{k, i}$ to $\vec \varphi_{k,i}$ using geometric relations
            \State $\eta_{k,i} \leftarrow u_{k-1,i}$
        \EndFor
        
        \State // \textit{Step 2: Evaluate $\vec y_k$-Data Evidence for Every $\vec \varphi_{k,i}$}
        \For {$i=0:N_{\text{par}}-1$}
            \If{Pulsed Scheme}
                \State Calculate cost $\rho(\vec y_k, \vec \varphi_{k,i})$ using \eqref{eq:cost-func-freq-domain-AF}
            \ElsIf{CW Scheme}
                \State Calculate cost $\rho(\vec y_k, \vec \varphi_{k,i})$ using \eqref{eq:cost-func-freq-domain-AF-2}
            \EndIf
        \EndFor

        \State // \textit{Step 3: Generate Posterior Particles $\vec x_{k,i}$; See \eqref{eq:par-representation-posterior-weight}}
        \For {$i=0:N_{\text{par}}-1$}
            \State Keep $\vec x_{k,i}$ unchanged
            \State Update weight by $u_{k,i} \leftarrow \eta_{k,i} \cdot e^{-\xi \rho(\vec y_k, \vec \varphi_{k,i})}$
        \EndFor
        \State Normalize weights $\{u_{k,i}\}_{\forall i \in[N_{\text{par}}]}$
        
        \State  // \textit{Step 4: Resampling}
        \State $N_{\text{eff}} \leftarrow 1/\sum^{N_{\text{par}} - 1}_{i=0} u^2_{k,i}$
        \If {$N_{\text{eff}} < N_{\text{thres}}$}
            \State Resample $N_{\text{par}}$ times from $\{\bm x_{k,i},u_{k,i}\}_{\forall i\in [N_{\text{par}}]}$
            \State $u_{k,i} \leftarrow 1/N_{\text{par}},~\forall i \in [N_{\text{par}}]$
        \EndIf 


        \Ensure Posterior particles and weights $\{\vec x_{k,i}, u_{k,i}\}_{\forall i\in [N_{\text{par}}]}$.
    \end{algorithmic}
\end{algorithm}

The computational complexity, in time, of PF-SLTR can be analyzed as follows:
\begin{itemize}
    \item In Step 1, if $\vec f_k$ is linear in $\rvec x_{k-1}$, which is the case for the usual constant-velocity and constant-acceleration models \cite{li2003survey}, the complexity is $\cal O(N_{\text{par}} d^2_{x})$, where $d_{x}$ denotes the dimension of $\rvec x_{k-1}$ (assuming fixed across different $k$'s).

    \item In Step 2, for the pulsed scheme, the complexity is $\cal O(N_{\text{par}} N_p N_{\text{cr}} (N_t + N_r))$; for the CW scheme, the complexity is $\cal O(N_{\text{par}} N_p M N_c (N_t + N_r))$.

    \item In Step 3, the complexity is $\cal O(N_{\text{par}})$ because the $\vec y_k$-data evidence $\{\rho(\vec y_k, \vec \varphi_{k,i})\}_{\forall i \in [N_{\text{par}]}}$ of the particles have already been evaluated in Step 2. 

    \item In Step 4, the complexity is $\cal O(N_{\text{par}})$, if the systematic resampling is adopted.
\end{itemize}
Overall, the computational complexity in time of PF-SLTR is $\cal O(N_{\text{par}} d^2_{x} + N_{\text{par}} N_p N_{\text{cr}} (N_t + N_r))$ for the pulsed scheme, and is $\cal O(N_{\text{par}} d^2_{x} + N_{\text{par}} N_p M N_c (N_t + N_r))$ for the CW scheme.

\section{\captext{Discussions on PF-SLTR}}\label{sec:discussions}
\hl{In this section, we discuss the assumptions, advantages, and limitations of PF-SLTR.}

\textit{State Estimation}: 
In Bayesian target tracking, target existence detection, track initialization, and state estimation are related but distinct inference tasks. This article assumes that the ToI is present in the surveillance region and that track initialization has already been performed.  This assumption is standard in the target tracking literature when the objective is to evaluate the performance of a state estimation algorithm in isolation \cite{li2003survey,stone2013bayesian}, i.e., to demonstrate that the proposed particle filtering mechanism yields more accurate and stable state estimates under raw radar measurements. Therefore, we specifically focus on the recursive Bayesian estimation of the target kinematic state $\rvec x_k$ from raw radar measurements $\rvec y_k$; recall \eqref{eq:target-trcking} and \eqref{eq:measurements}. Note that operating directly on raw radar measurements does not require the target-existence decision to be jointly considered; see, e.g., \cite{ratpunpairoj2015particle,huleihel2013optimal,herbert2017mmse,herbert2018computationally}. 

\textit{Existence Detection}: 
When target tracking considers existence detection, track-before-detect \red{particle filtering (TBD PF)} is standard \cite{davey2007comparison,davey2013track,kwon2019particle}. However, \red{TBD PF} relies heavily on the accuracy of the likelihood model $q(\vec y_k | \vec \varphi_k, e_k = 1)$ and the IPN model $q(\vec y_k | e_k = 0)$, where $\vec y_k$ denotes the radar raw data at time $k$ [see \eqref{eq:measurements}], $\vec \varphi_k$ defines the ToI parameter [see \eqref{eq:bayesian-tracking-lite-polar}], and $e_k$ is a Boolean that equals one if the ToI presents and zero otherwise. In the radar target tracking literature (e.g., \red{TBD PF}), these two probabilistic models are typically constructed using simplified assumptions (e.g., Gaussian IPN, single radio path, or single target), which may be inaccurate under realistic conditions such as dense multipath, multiple targets, or non-Gaussian interference. To be specific, when there exist scatterers (including targets of non-interest) and the statistics of these scatterers are unknown, the measurement model $q(\vec y_k | \vec \varphi_k, e_k = 1)$ and the IPN model $q(\vec y_k | e_k = 0)$ may deviate from Gaussianity, since the IPN term $\rvec n_c(t)$ in \eqref{eq:general-radar-measurement-compact} is generally non-Gaussian; cf. \eqref{eq:general-radar-measurement}. This particularly affects signal-level particle filters that rely on explicit likelihood evaluation, for example, those that assume the Gaussianity of the IPN $\rvec n_c(t)$. Consequently, \red{TBD PF} remains sensitive to likelihood model mismatch, particularly in challenging environments where accurate probabilistic modeling and statistical characterization of the measurements are difficult. Hence, existence detection using TBD PF is not considered.


\textit{Multi-Target Tracking}: Extending the proposed PF-SLTR framework to the multi-target tracking setting is not a straightforward generalization, as it introduces several nontrivial challenges that are fundamentally different in nature from the single-target case. In particular, the following three difficulties are immediate \cite{vermaak2005monte,ito2020multi,wang2017survey}: a) data association that determines which path component of the raw measurement $\rvec y_k$ is generated by which target [cf. \eqref{eq:general-radar-measurement}, \eqref{eq:general-radar-measurement-compact}, and \eqref{eq:measurements}]; b) track management that includes targets' initiation and termination; and c) IPN modeling where the IPN $\rvec n_c(t)$ depends on other targets' parameters. Handling these technical aspects typically requires a different modeling paradigm, such as Random Finite Set (RFS)-based formulations, and would significantly increase the scope and complexity of the problem. For this reason, and consistent with common practice in the development of new filtering methodologies, we restrict attention to the single-target case and assume that the ToI has already been initialized. On the other hand, it is worth noting that the proposed PF-SLTR can be extended to multi-target tracking. To be specific, it can be incorporated into existing multi-target Bayesian frameworks by replacing the state- and measurement-update steps for each target with the proposed PF-SLTR, particularly when data association and track management are handled separately. The benefit is that, unlike conventional PF and \red{TBD PF} methods, PF-SLTR does not explicitly require the accurate modeling of the IPN $\rvec n_c(t)$: recall that in Algorithm \ref{algo:PF-SLTR}, the cost functions \eqref{eq:cost-func-freq-domain-AF} and \eqref{eq:cost-func-freq-domain-AF-2} do not rely on $\rvec n_c(t)$. 
For example, supposing we have two targets in the surveillance region and no other scatterers exist, \eqref{eq:general-radar-measurement} simplifies into
\begin{equation}\label{eq:2-target-radar-measurements}
    \begin{array}{cl}
        \rvec y(t) &= \beta_1 \vec b(\theta_1) \vec a^\H(\theta_1) \rvec s(t - \tau_1) e^{j 2 \pi \nu_1 t} + \\
        & \quad \quad \beta_2 \vec b(\theta_2) \vec a^\H(\theta_2) \rvec s(t - \tau_2) e^{j 2 \pi \nu_2 t} +  \rvec n(t),
    \end{array}
\end{equation}
where $(\tau_1, \nu_1, \theta_1)$ denotes the parameter of the first target, $(\tau_2, \nu_2, \theta_2)$ that of the second target, and $\rvec n(t)$ the Gaussian environmental noise. When we track the first target, the radar measurement model becomes 
\begin{equation}\label{eq:2-target-radar-measurements-first}
    \begin{array}{cl}
        \rvec y(t) &= \beta_1 \vec b(\theta_1) \vec a^\H(\theta_1) \rvec s(t - \tau_1) e^{j 2 \pi \nu_1 t} + \rvec n_c(t),
    \end{array}
\end{equation}
where the IPN term $\rvec n_c(t) \defeq \beta_2 \vec b(\theta_2) \vec a^\H(\theta_2) \rvec s(t - \tau_2) e^{j 2 \pi \nu_2 t} + \rvec n(t)$. Instead, when we track the second target, the radar measurement model becomes 
\begin{equation}\label{eq:2-target-radar-measurements-second}
    \begin{array}{cl}
        \rvec y(t) &= \beta_2 \vec b(\theta_2) \vec a^\H(\theta_2) \rvec s(t - \tau_2) e^{j 2 \pi \nu_2 t} + \rvec n_c(t),
    \end{array}
\end{equation}
where the IPN term $\rvec n_c(t) \defeq \beta_1 \vec b(\theta_1) \vec a^\H(\theta_1) \rvec s(t - \tau_1) e^{j 2 \pi \nu_1 t} + \rvec n(t)$. As indicated, when data association and track management are handled separately, which is a common practice, multi-target tracking can be decomposed into a set of single-target tracking problems: Different targets are to be tracked by separate particle filters. 
Therefore, within the context of target tracking using MIMO-OFDM radar without explicit IPN handling, focusing on single-target state estimation does not compromise the technical contributions of this article.

\textit{Applicability of PF-SLTR}: \blue{Below, we compare the proposed PF-SLTR with the conventional likelihood-based PF and} \red{TBD PF}\blue{, under various operating conditions.} All three frameworks are formulated at the signal level; namely, they process raw radar data $\rvec y_k$ directly [see \eqref{eq:measurements}], rather than extracted range–Doppler–angle measurements of $(\tau_{0, k}, \nu_{0, k}, \theta_{0, k})$. Typically, the following conditions are investigated.
\begin{itemize}
    \item \textit{Large Data Dimensions in Time, Frequency, and Space}: When the transmit antenna number $N_t$, receive antenna number $N_r$, bandwidth $B$, or CPI length $T_i$ are large, extracting range-Doppler-angle measurements via matched filtering or advanced optimization becomes computationally intensive. \hl{This makes information-level tracking pipelines (i.e., detect-then-track) computationally expensive, especially when combined with particle filtering. In such cases, signal-level approaches (e.g., PF-SLTR) are preferred to reduce computational burden.}

    \item \textit{Model Mismatch in Likelihood Function $\vec \varphi_k \mapsto q(\vec y_k \mid \vec \varphi_k)$}: The performance of the conventional likelihood-based PF and \red{TBD PF} relies heavily on the accuracy of the likelihood model $q(\vec y_k \mid \vec \varphi_k)$. \hl{If non-negligible mismatches exist in $q(\vec y_k \mid \vec \varphi_k)$, PF-SLTR} \hlblue{can be considered first} \hl{since it does not mandate the accurate modeling of the IPN term $\rvec n_c(t)$ and the likelihood relation $q(\vec y_k \mid \vec \varphi_k)$}, whereas the likelihood-based PF and \red{TBD PF} do.
    
    \item \textit{Low Signal-to-Interference-Plus-Noise Ratio (SINR)}: When the received ToI signal is weak relative to IPN, \red{TBD PF} is generally preferred, provided that the measurement model $q(\vec y_k | \vec \varphi_k, e_k = 1)$ and the IPN model $q(\vec y_k | e_k = 0)$ can be accurately specified. However, the performance of \red{TBD PF} is highly sensitive to the errors in the likelihood and IPN models. Therefore, when a significant model mismatch is present, PF-SLTR can be preferable across different SINR regimes.

    \item \textit{Dense Clutter (Due to Targets of Non-Interest and Scatterers)}: Dense clutter complicates IPN modeling (i.e., determining the statistical properties of $\rvec n_c(t)$), leading to likelihood mismatch and degraded performance of likelihood-based and TBD PF methods. \hl{In contrast, PF-SLTR} \hlblue{tends to be} \hl{more robust in such scenarios because it avoids explicit likelihood and IPN modeling.}

    \item \textit{Fast-Fading (i.e., Rapidly Varying Path Gains $\beta_{0,k}$)}: Standard Bayesian filtering frameworks require augmenting nuisance parameters, such as complex path gains $\beta_{0,k}$, into the state vector, increasing its dimensionality and thus degrading statistical and computational efficiency. When these parameters vary abruptly across tracking periods $k$, standard state-transition models (e.g., random walk or constant-velocity) become ineffective. Consequently, likelihood-based PF (including the variant Rao-Blackwellized PF) and \red{TBD PF} methods become less effective. In contrast, PF-SLTR does not explicitly model path gains and is therefore robust to their rapid variations.
\end{itemize}
\hl{In summary, likelihood-based PF is most effective when accurate statistical models are available and the state dimension is moderate.} \red{TBD PF} \hl{methods are advantageous in low-SINR regimes when reliable likelihood models can be constructed. The proposed PF-SLTR is particularly beneficial in scenarios involving significant model uncertainty or rapidly varying path-gain parameters where explicit likelihood modeling becomes unreliable or computationally burdensome.}

\textit{Cost Functions}: In this article, the cost functions $\rho_3$ in \eqref{eq:cost-func-time-domain-AF}, $\rho_4$ in \eqref{eq:cost-func-freq-domain-AF}, and $\rho^\prime_4$ in \eqref{eq:cost-func-freq-domain-AF-2} are studied and recommended; see Algorithm \ref{algo:PF-SLTR}. This choice is motivated by the minimization-based marginalization trick in \eqref{eq:cost-func-time-domain} and \eqref{eq:cost-func-freq-domain}, which handles the nuisance parameter $\beta_0$; cf. Propositions \ref{prop:h1} and \ref{prop:h2}. For all radars operating under the measurement model \eqref{eq:general-radar-measurement} [see also \eqref{eq:general-radar-measurement-compact}-\eqref{eq:measurements}], the cost functions designed in \eqref{eq:cost-func-time-domain-AF}, \eqref{eq:cost-func-freq-domain-AF}, and \eqref{eq:cost-func-freq-domain-AF-2} are applicable within the proposed PF-SLTR framework, regardless of the waveform structure or parameterization of the wireless signal $\rvec s(t)$ and noise $\rvec n(t)$. Specifically, $\rho_3$ in \eqref{eq:cost-func-time-domain-AF} can be adopted when time-domain processing is employed for pulsed radars, $\rho_4$ in \eqref{eq:cost-func-freq-domain-AF} when frequency-domain processing is used for pulsed radars, and $\rho^\prime_4$ in \eqref{eq:cost-func-freq-domain-AF-2} when frequency-domain processing is applied to continuous-wave OFDM radars. Given that the radar measurement model \eqref{eq:general-radar-measurement} is general, the proposed PF-SLTR in Algorithm \ref{algo:PF-SLTR}, equipped with the cost functions $\rho_3$, $\rho_4$, or $\rho^\prime_4$, is broadly applicable to a wide range of MIMO radar systems, including both conventional radars and ISAC radars. In this article, MIMO-OFDM communication stations serve as radars, as raw OFDM communication waveforms are reused for sensing. Indeed, other types of cost functions may also be developed, provided that they comply with the fundamental rule in Rule \ref{rule:cost-func-h}. However, discovering these new cost functions, derived from alternative motivations, can be explored in future research. As an initiating article in this research direction, focusing on the proposed $\rho_3$, $\rho_4$, and $\rho^\prime_4$ does not compromise the technical contributions of this article and the general applicability of the proposed PF-SLTR algorithm.

\textit{Likelihood Adaptation}: In the literature, such as \cite[Algo.~1]{conti2019soft}, learning the likelihood relation $\vec x_k \mapsto q(\vec y_k | \vec x_k)$ from historical data pairs $(\vec y_k, \vec x_k)$ is suggested; cf. \eqref{eq:target-trcking}. Within the problem setting of this article, adapting this idea would require constructing or learning the likelihood relation $\vec \varphi_k \mapsto q(\vec y_k | \vec \varphi_k)$ associated with the raw radar measurement $\vec y_k$ and the target parameter $\vec \varphi_k \defeq (\tau_{0,k}, \nu_{0,k}, \theta_{0,k})$. This is difficult for the following three reasons:
\begin{itemize}
    \item The required labeled training data pairs $(\vec y_k, \vec \varphi_k)$ are generally unavailable in online radar target tracking problems. Data collected from other scenarios may not ensure sufficient cross-scenario generalizability because the radio propagation environment and scatterer (including targets of no interest) conditions can change significantly from scenario to scenario.
    
    \item The signal-level measurement $\vec y_k$ is high-dimensional in MIMO pulse-Doppler radar systems, as it may contain samples across receive antennas, fast-time range bins, slow-time Doppler bins, OFDM symbols, and subcarriers; cf. \eqref{eq:general-radar-measurement-compact-discret}, \eqref{eq:general-radar-measurement-compact-freq}, \eqref{eq:CW-OFDM-Waveform}, and \eqref{eq:measurements}. Hence, learning a reliable likelihood relation $\vec \varphi_k \mapsto q(\vec y_k | \vec \varphi_k)$ in such high-dimensional spaces would require substantial training data and computational resources, which is difficult to support within a short tracking period $T_t$; cf. Fig. \ref{fig:OFDM_ISAC_frame}. 
    
    \item The likelihood relation may vary across tracking periods $k$ because nuisance parameters and IPN statistics can be time-varying. In particular, the ToI path gain $\beta_{0,k}$ may change rapidly across different tracking periods $k$, and the IPN term may also evolve with time-varying scatterers. Therefore, a likelihood model learned in one period or scenario may not remain valid in another.
\end{itemize}

For these difficulties, \cite{conti2019soft} assumes an off-line training phase in which a likelihood-like relation is learned from sufficient training data, and this learned relation is assumed to remain valid in the online deployment scenarios. To relax this assumption, this article takes a different route. Instead of learning or adapting the explicit likelihood relation $\vec \varphi_k \mapsto q(\vec y_k | \vec \varphi_k)$, PF-SLTR uses domain-motivated data-evidence cost functions within a Gibbs-posterior particle filtering framework.

\section{Experiments}\label{sec:experiment}
In this section, we use experiments to validate the efficiency of the proposed PF-SLTR (i.e., Algorithm \ref{algo:PF-SLTR}) in MIMO-OFDM ISAC systems. Let the state vector $\rvec x_k \defeq [R_k, \varsigma_k, \theta_k, \omega_k]^\T$ contain the ToI range $R_k$, radial speed $\varsigma_k$, DoA angle $\theta_k$, and DoA angular velocity $\omega_k$. 

For the state transition equation $\rvec x_k  =  \vec f_k(\rvec x_{k-1}, \rvec w_{k-1})$ in \eqref{eq:target-trcking}, we use the constant-velocity (CV) model \cite[Eq.~(14)]{li2003survey}, which is linear in $\rvec x_k$, because the CV model naturally captures the kinematic relations that the radial speed and DoA angular velocity are the time derivatives of the range and DoA angle, respectively. We have
\begin{equation}\label{eq:track-model-process}
    \rvec x_k = \operatorname{diag}(\mat F_{k-1}, \mat F_{k-1}) \rvec x_{k-1} + \operatorname{diag}(\mat G_{k-1}, \mat G_{k-1}) \rvec w_{k-1},
\end{equation}
where $\operatorname{diag}$ forms a block-diagonal matrix from its arguments,
\[
\mat F_{k-1} \defeq 
\left[
\begin{array}{cc}
   1  &  T_t \\
   0  &   1
\end{array}
\right],~\text{and }
\mat G_{k-1} \defeq 
\left[
\begin{array}{c}
   T^2_t/2 \\
   T_t 
\end{array}
\right].
\]

For the measurement model $\rvec y_k  =  \vec h_k(\rvec x_{k}, \rvec v_{k})$ in \eqref{eq:target-trcking}, it particularizes into \eqref{eq:general-radar-measurement-compact-discret} in the time domain, \eqref{eq:general-radar-measurement-compact-freq} in the frequency domain, or \eqref{eq:CW-OFDM-Waveform} for the CW MIMO-OFDM case. Since this measurement mechanism is nonlinear in $\rvec x_k$, linear state estimation approaches, such as Kalman filtering, are not applicable. In addition, because the nonlinearity is severe, nonlinear Kalman variants, e.g., extended and unscented Kalman filtering, often lead to unreliable estimates \cite[p.~433,~p.~462]{simon2006optimal}. Therefore, particle filtering is widely adopted in the literature of signal-level target tracking; see, e.g., \cite[Figs.~4-5]{gu2023beam}, \cite{ratpunpairoj2015particle,huleihel2013optimal,herbert2017mmse}. However, employing conventional particle filters requires additional modeling assumptions due to the unknown statistical information of the IPN terms $\rvec n_{c,p}(l)$, $\rvecb n_{c,p}(n)$, and $\rvecb n_{c,p,m}(n)$ in \eqref{eq:general-radar-measurement-compact-discret}, \eqref{eq:general-radar-measurement-compact-freq}, and \eqref{eq:CW-OFDM-Waveform}, respectively.

The following filtering schemes are implemented for performance comparison.

    \textit{KF-ILTR}: Kalman filtering (KF) for information-level target tracking (ILTR). First, we obtain the measurements of $(\tau_{0}, \nu_{0}, \theta_{0})$ at each tracking time $k$, using the $3$-dimensional (3D) range-Doppler-DoA matched filter. For demonstration purposes only, other advanced methods, e.g., \cite{zheng2017super,wu2022super,xiao2024novel}, \cite[Table~II]{zhang2021overview}, are not implemented because they are even more computationally prohibitive for real-time use (e.g., within the tracking time window $T_t = 50$ms). The computational complexity, in time, of the 3D range-Doppler-DoA matched filters \eqref{eq:cost-func-freq-domain-AF} and \eqref{eq:cost-func-freq-domain-AF-2} is $\cal O(N^3_{\text{tg}} N_p N_{\text{cr}} (N_t + N_r))$ and $\cal O(N^3_{\text{tg}} N_p M N_{\text{c}} (N_t + N_r))$, respectively, where $N_{\text{tg}}$ is the length of the tracking gate in each dimension. To clarify, to search for the peak of the matched filter output in a $N_{\text{tg}} \times N_{\text{tg}} \times N_{\text{tg}}$ grid of range-Doppler-DoA, we need to evaluate the ambiguity function \eqref{eq:cost-func-freq-domain-AF} or \eqref{eq:cost-func-freq-domain-AF-2} $N^3_{\text{tg}}$ times. As indicated, even $N_{\text{tg}} = 30$, we need to evaluate ambiguity functions $27000$ times, and the complexity of each time is $\cal O(N_p N_{\text{cr}} (N_t + N_r))$ or $\cal O(N_p M N_{\text{c}} (N_t + N_r))$. For systems with large $N_p$, $N_{\text{cr}}$ (i.e., $L$), $M$, $N_c$, $N_t$, and $N_r$, even the traditional 3D range-Doppler-DoA matched filtering is unaffordable, no mention advanced approaches \cite{zheng2017super,wu2022super}, \cite[Table~II]{zhang2021overview}. In this case, the measurement model $\rvec y_k  =  \vec h_k(\rvec x_{k}, \rvec v_{k})$ in \eqref{eq:target-trcking} simplifies to
    \begin{equation}\label{eq:track-model-measurement-ILTR}
        \rvec y_k = 
        \left[ 
        \begin{array}{cccc}
           1  &  0 & 0 & 0 \\
           0  &  1 & 0 & 0 \\
           0  &  0 & 1 & 0
        \end{array}
        \right]
        \rvec x_{k} + \rvec v_{k},
    \end{equation}
    where $\rvec y_k$ contains the measured values of the range $\bar R_k$, the radial velocity $\bar \varsigma_k$, and the DoA $\bar \theta_k$ that are calculated using the matched-filter estimates of $(\tau_{0}, \nu_{0}, \theta_{0})$. Since both the state transition equation \eqref{eq:track-model-process} and state measurement equation \eqref{eq:track-model-measurement-ILTR} are linear, Kalman filtering is applicable.

    \textit{UKF-SLTR-A}: The unscented Kalman filtering (UKF) for signal-level target tracking (SLTR) that augments (A) the real and imaginary components of the complex gain $\beta_{0,k}$ (due to the ToI) into the state vector $\rvec x_k$ \cite[Sec.~14.3]{simon2006optimal}. In this case, the state measurement model is \eqref{eq:general-radar-measurement-compact-discret}, \eqref{eq:general-radar-measurement-compact-freq}, or \eqref{eq:CW-OFDM-Waveform}. Because we do not know the statistical information of the IPN terms $\rvec n_{c,p}(l)$, $\rvecb n_{c,p}(n)$, and $\rvecb n_{c,p,m}(n)$, it is assumed that the conditional distribution $q(\vec y_k | \vec x_k)$ is complex Gaussian and there are no scatterers. To be specific, at the tracking time $k$, for every $p,m,n$, we have
    \begin{equation}\label{eq:likelihood-PF-SLTR-A-pulsed}
        q(\vecb y_{k,p}(n) | \beta_0, \tau_0, \nu_0, \theta_0) \defeq \cal{CN}(\vec \epsilon_1,~\sigma^2_{\text{cn}} \mat I_{N_r})
    \end{equation}
    where $\vec \epsilon_1 \defeq \beta_0 \vec b(\theta_0) \vec a^\H(\theta_0) \rvecb s_p(n) e^{-j 2 \pi n \Delta \tau_0} e^{j 2 \pi p T_r \nu_0}$, for the pulsed MIMO-OFDM paradigm [see \eqref{eq:general-radar-measurement-compact-freq}], or 
    \begin{equation}\label{eq:likelihood-PF-SLTR-A-CW}
    q(\vecb y_{k,p,m}(n) |  \beta_0, \tau_0, \nu_0, \theta_0) \defeq \cal{CN}(\vec \epsilon_2,~\sigma^2_{\text{cn}} \mat I_{N_r})
    \end{equation}
    where $\vec \epsilon_2 \defeq \beta_0 \vec b(\theta_0) \vec a^\H(\theta_0) \rvec c_{p,m}(n) e^{-j 2 \pi n \Delta_f \tau_0} e^{j 2 \pi p T_r \nu_0}$ for the CW MIMO-OFDM paradigm [see \eqref{eq:CW-OFDM-Waveform}]; $\cal{CN}(\vec \epsilon,~\sigma^2_{\text{cn}} \mat I_{N_r})$ denotes the complex Gaussian distribution with mean $\vec \epsilon$, covariance $\sigma^2_{\text{cn}} \mat I_{N_r}$, and zero pseudo-covariance; $\sigma^2_{\text{cn}}$ is the assumed-known power of the channel noise. Since the measurement model is severely nonlinear in $\rvec x_k$ (see the definitions of $\vec \epsilon_1$ and $\vec \epsilon_2$ above) and the likelihood models \eqref{eq:likelihood-PF-SLTR-A-pulsed} and \eqref{eq:likelihood-PF-SLTR-A-CW} can be misleading, the UKF-SLTR-A often gives unreliable state estimates \cite[p.~462]{simon2006optimal}.

    \textit{PF-SLTR-A}: The traditional particle filtering (PF) for signal-level target tracking (SLTR) that augments (A) the real and imaginary components of the complex gain $\beta_{0,k}$ (due to the ToI) into the state vector $\rvec x_k$ \cite{ratpunpairoj2015particle,huleihel2013optimal,herbert2017mmse}. In this case, the UKF is replaced with the traditional PF and the remaining settings, e.g., the likelihood models \eqref{eq:likelihood-PF-SLTR-A-pulsed} and \eqref{eq:likelihood-PF-SLTR-A-CW}, remain unchanged. As is frequently observed in practice, the convergence of PF-SLTR-A is difficult to ensure, especially when the path gain $\beta_{0,k}$ is significantly time-varying across different $k$'s or when scatterers exist, which makes the likelihoods in \eqref{eq:likelihood-PF-SLTR-A-pulsed} and \eqref{eq:likelihood-PF-SLTR-A-CW} largely unreliable.

    \textit{RBPF-SLTR-A}: The Rao-Blackwellized particle filtering (RBPF) \cite{schon2005marginalized} for signal-level target tracking (SLTR) that augments (A) the real and imaginary components of the complex gain $\beta_{0,k}$ into the state vector $\rvec x_k$. Since the state transition equation is linear in  $\rvec x_{k-1}$ and the measurement equation \eqref{eq:general-radar-measurement} [cf. \eqref{eq:general-radar-measurement-compact-discret} and \eqref{eq:general-radar-measurement-compact-freq}] is linear in $\beta_{0,k}$, the RBPF that marginalizes out $\beta_{0,k}$ using a Kalman filter can improve the statistical and computational efficiency, compared to the traditional PF-SLTR-A; see \cite[Sec. II-G]{cappe2007overview}. 

The performance of each method is measured by both the running times and the target-tracking mean-squared errors (MSEs). The experiments were performed in MATLAB R2025b on a Lenovo desktop with an Intel$^\copyright$ Core$^\text{\tiny TM}$ i7-14700F processor (2.10 GHz), 64 GB of RAM, and a 64-bit operating system. All the source data and codes are available online at GitHub: \url{https://github.com/Spratm-Asleaf/PF-SLTR-MIMO-OFDM}.



In this section, we consider a continuous-wave MIMO-OFDM ISAC system with center frequency $f_c = 10$GHz for tracking a near and slow ToI. For supplementary experiments (e.g., the pulsed case for far and fast targets), see Appendix D in the online supplementary materials.

The experimental scenario is displayed in Fig. \ref{fig:CW_trajectory}. Therein, the ToI is moving along a curved trajectory, and five static scatterers are assumed, with ranges $[30,45,100,80,130]$ and DoAs $[20,-30,35,-8,15]$, respectively. The communication-centric waveform structure in Fig. \ref{fig:OFDM_ISAC_frame}(c) is used, and the cost function in \eqref{eq:cost-func-freq-domain-AF-2} is employed in PF-SLTR to track the ToI. For communications waveforms, we assume that each OFDM subcarrier contains random symbols uniformly from the quadrature amplitude modulation (QAM) constellation of order $64$. Since during target tracking, the communication infrastructure and waveforms are reused for sensing, the communication performance remains unaffected.

\begin{figure}[!htbp]
    \centering

    \subfigure[Trajectory]{
        \centering
        \includegraphics[height=3cm]{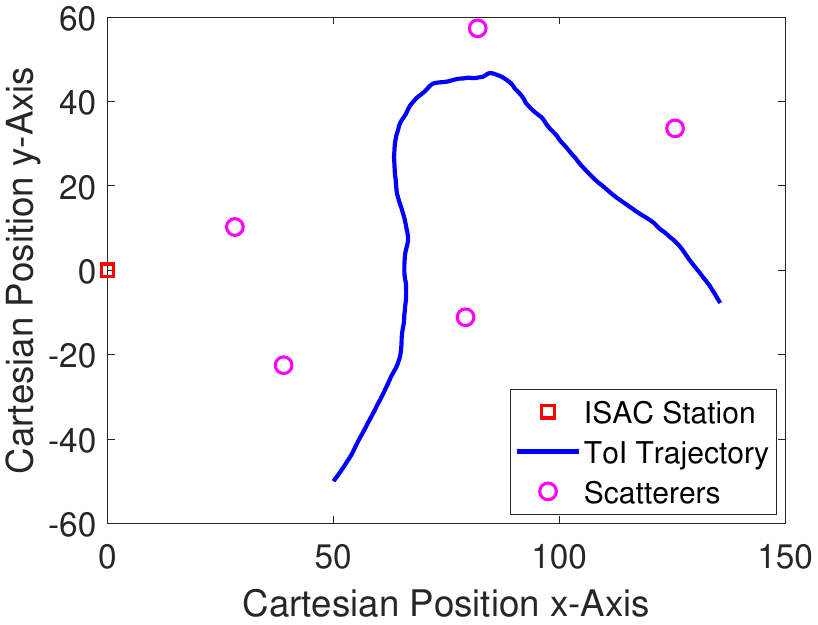}
    }
    \subfigure[Range]{
        \centering
        \includegraphics[height=3cm]{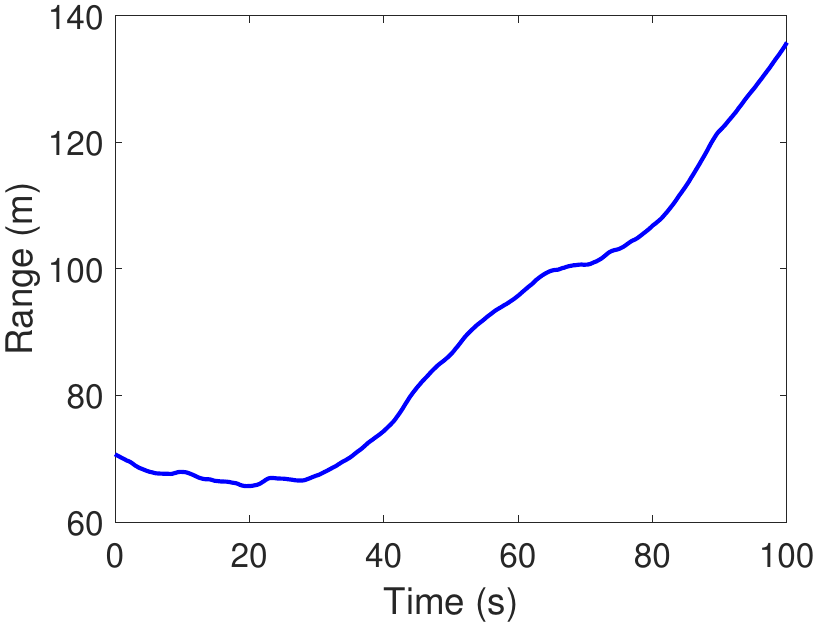}
    }

    \subfigure[Radial Velocity]{
        \centering
        \includegraphics[height=3cm]{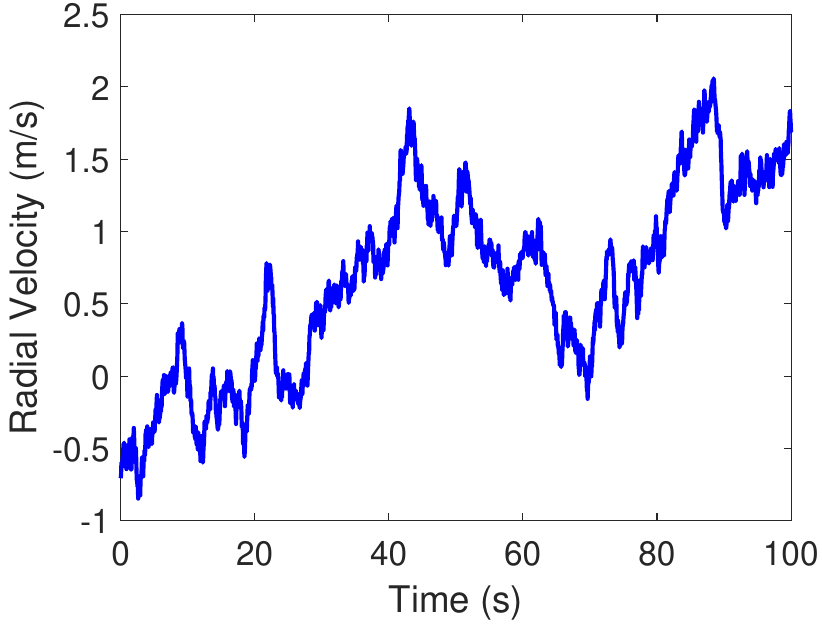}
    }
    \subfigure[DoA]{
        \centering
        \includegraphics[height=3cm]{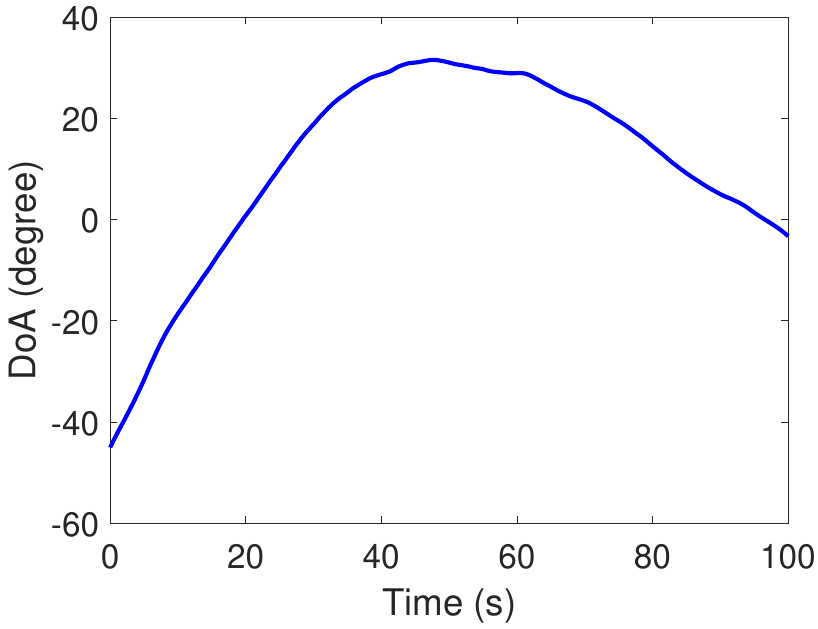}
    }
    
    \caption{Tracking the maneuvering ToI using CW MIMO-OFDM ISAC.}
    \label{fig:CW_trajectory}
\end{figure}

On the sensing side, the tracking time period is set to $T_t = 50$ms, and the total tracking time is $100$s. To cover the maximum range of $150$m, we let the CP duration be $T_c = 1 \mu$s. Suppose that the OFDM symbol duration (for data payloads) is five times of the CP length (i.e., $T_o = 5 \mu$s), the number of OFDM symbols in a pulse is $M=1$, the number of subcarriers is $N_c = 256$, the number of transmit antennas is $N_t = 64$, the number of receive antennas is $N_r = 64$, and the number of pulses in a coherent processing interval is $N_p = 1$. Since using only the range and DoA information to track the ToI is sufficient, it is safe to set $N_p = 1$. As a result, the pulse duration is $T_p = M(T_c + T_o) = 6\mu$s, the pulse repetition interval is $T_r = T_p = 6\mu$s, the coherent processing interval is $T_i = N_p T_r = 6\mu$s, the subcarrier spacing is $\Delta_f = 1/T_o = 0.2$MHz, the bandwidth is $B = N_c \Delta_f = 51.2$MHz, the sampling time interval and the round-trip delay resolution is $T_s = 1/B = 0.019531\mu$s (i.e., the range resolution is $2.93$m), the resolution of $\sin$-DoA (not DoA itself) is $2/64 = 1/32$. Let the normalized channel noise power be $\sigma^2_{\text{cn}} \defeq 1$ and the transmit SNR ($\text{SNR}_\text{tx}$) be $-10$dB, where 
\begin{equation}
    \text{SNR}_\text{tx} \defeq 10\log_{10}(P_t/\sigma^2_{\text{cn}}).
\end{equation}
The ToI path gain, normalized with respect to the initial ToI range $R_{0,0}$, is modeled as
\begin{equation}\label{eq:path-gain}
\beta_{0,k} \defeq \kappa \cdot \left(\frac{R_{0,0}}{R_{0,k}}\right)^2 \cdot \sqrt{\rscl s_{0,k}}
e^{j2\pi \rscl r_k},
\end{equation}
where $\kappa$ denotes the attenuation factor at the initial range $R_{0,0}$ and $R_{0,k}$ the ToI range at tracking period $k$; $\rscl s_{0,k}$ follows the exponential distribution with unit mean, and $\rscl r_k$ is uniformly distributed over $[0,1]$; $\rscl s_{0,k}$ models the normalized Swerling-I-type radar-cross-section fluctuation of the ToI, while $e^{j2\pi \rscl r_k}$ models its random scattering phase \cite[Chap.~2]{richards2022fundamentals}. The factor $({R_{0,0}}/{R_{0,k}})^2$ models the relative range-dependent amplitude attenuation. The receive SNR is determined by the transmit SNR and the path gain $\beta_{0,k}$ through \eqref{eq:general-radar-measurement}. Both the real and imaginary components of $\beta_{0,k}$ vary abruptly across tracking periods $k$. In the experiments, we set $\kappa \defeq 0.7$. Moreover, the number of particles is set to $N_{\text{par}} \defeq 200$, the effective particle number used in particle resampling is $N_{\text{thres}} \defeq N_{\text{par}} / 2$, and the power coefficient in \eqref{eq:par-representation-posterior-weight} is $\xi \defeq 1$. 

The tracking results for the 2-dimensional (2D) Cartesian positions (calculated using tracked ranges and DoAs) are shown in Table \ref{tab:scenario-cw} and Appendix D in the online supplementary materials. In Table \ref{tab:scenario-cw}, the 2D position MSEs are averaged over all different $k$'s and independent Monte-Carlo (MC) trials, so are the running times (Time); however, we found that the empirical performance variations are not obvious across different MC episodes because MSEs are averaged on a long time trajectory indexed by $k$. From Table \ref{tab:scenario-cw}, we can see that PF-SLTR is \textit{significantly} better than KF-ILTR, UKF-SLTR-A, PF-SLTR-A, and RBPF-SLTR-A, in terms of both computation and accuracy. This is because KF-ILTR must consume a huge amount of time to obtain the raw measurements of the ToI parameters $(\tau_{0,k}, \theta_{0,k})$ at each time $k$; NB: when $N_p = 1$, Doppler processing is disabled. Within a tracking time window of $T_t = 50$ms, spending $89.3$ms on computations is operationally intolerable because target tracking is a real-time signal processing problem. In contrast, the proposed PF-SLTR does not require this operation. In addition, the convergence of UKF-SLTR-A, PF-SLTR-A, and RBPF-SLTR-A is difficult to ensure for the following two reasons: First, it is challenging to accurately track the quickly-jumping path gain $\beta_{0,k}$ in \eqref{eq:path-gain};\footnote{Under the model \eqref{eq:path-gain}, both the real and imaginary components of $\beta_{0,k}$ vary abruptly across tracking periods $k$.} 
second, the likelihood functions in \eqref{eq:likelihood-PF-SLTR-A-pulsed} and \eqref{eq:likelihood-PF-SLTR-A-CW} are not sufficiently precise due to the presence of unknown scatterers. However, the proposed PF-SLTR in Algorithm \ref{algo:PF-SLTR} does not require estimating $\beta_{0,k}$ or relying on the exact likelihood function, resulting in superior performance.

\begin{table}[!htbp]
\centering
\caption{Tracking results for the scenario in Fig. \ref{fig:CW_trajectory}}
\label{tab:scenario-cw}
\begin{tabular}{lcccccccc}
\bottomrule[0.8pt]
                        & \textbf{PF-SLTR}   & \tabincell{l}{KF-\\ILTR}   &   \tabincell{l}{UKF-\\SLTR-A}  & \tabincell{l}{PF-\\SLTR-A}      & \tabincell{l}{RBPF-\\SLTR-A}\\
\specialrule{0.7pt}{0pt}{0pt}
MSE      & \textbf{0.0070}    & 1.5520   &   ({Diverged})   & ({Diverged})     & ({Diverged})\\
\hline
Time         & \textbf{29.2}      & 89.3    &   1.3   & 32.4           & 30.1\\
\toprule[0.8pt]
\multicolumn{6}{l}{
Note: Unit for MSE: m$^2$; unit for Time: ms.
}
\end{tabular}
\end{table}

In the following, we present the tracking results of PF-SLTR under different parameters. Results of other filters are not displayed because, as we can see, they are significantly dominated by PF-SLTR. When we vary one parameter, we keep the other the same as those used in Table \ref{tab:scenario-cw}.

The tracking results for different numbers of particles are given in Table \ref{tab:scenario-cw-par}. As expected, in particle filtering, due to the law of large numbers, the larger the number of particles, the less the average tracking error. However, as the number $N_{\text{par}}$ of particles increases, more computing resources are required.

\begin{table}[!htbp]
\centering
\caption{Tracking results of PF-SLTR when \captext{$N_{\text{par}}$} varies}
\label{tab:scenario-cw-par}
\begin{tabular}{lccccc}
\bottomrule[0.8pt]
$N_{\text{par}}$  & 12  & 25    & 50   & 100       & \textbf{200}\\
\specialrule{0.7pt}{0pt}{0pt}
MSE (unit: m$^2$) & (Diverged) & 0.1116    & 0.0421    & 0.0134    & \textbf{0.0070} \\
\hline
Time (unit: ms)   & 2.1 & 3.4     & 7.4      & 15.3   & \textbf{29.2} \\
\toprule[0.8pt]
\multicolumn{6}{l}{
Note: Boldfaced values are baselines in Table \ref{tab:scenario-cw}.
}
\end{tabular}
\end{table}

The tracking results for different transmit SNRs are given in Table \ref{tab:scenario-cw-snr}. As indicated, 
the larger value of the transmit SNR helps decrease the average tracking error. We do not report the running times because the value of the transmit SNR does not essentially affect the computational complexity.

\begin{table}[!htbp]
\centering
\caption{Tracking results of PF-SLTR when $\text{SNR}_\text{tx}$ varies}
\label{tab:scenario-cw-snr}
\begin{tabular}{lccccccccc}
\bottomrule[0.8pt]
$\text{SNR}_\text{tx}$   & -50 &   -48  & -45   & -42 & -40\\
\specialrule{0.7pt}{0pt}{0pt}
MSE (unit: m$^2$)  & (Diverged)  &   0.3238   & 0.0956  &  0.0591 & 0.0436\\
\bottomrule[0.8pt]
$\text{SNR}_\text{tx}$   & -30 & -20 & \textbf{-10}     &  0  & 10 \\
\specialrule{0.7pt}{0pt}{0pt}
MSE (unit: m$^2$) & 0.0119 & 0.0077 & \textbf{0.0070}     & 0.0065   & 0.0059 \\
\toprule[0.8pt]
\multicolumn{6}{l}{
Note: Boldfaced values are baselines in Table \ref{tab:scenario-cw}.
}
\end{tabular}
\end{table}


The tracking results for different numbers of antennas are given in Table \ref{tab:scenario-cw-tx}. As we can anticipate, the larger the number of antennas, the smaller the average tracking error, and the greater the computational burden. 

\begin{table}[!htbp]
\centering
\caption{Tracking results of PF-SLTR when $N_t$ and $N_r$ vary}
\label{tab:scenario-cw-tx}
\begin{tabular}{lccccccc}
\bottomrule[0.8pt]
$N_t = N_r =$         & 21  & 22   &  23  & 24  \\
\specialrule{0.7pt}{0pt}{0pt}
MSE (unit: m$^2$)       & (Diverged)   & 0.0346     &   0.0338  & 0.0273\\
\hline
Time (unit: ms)         & 11.4   & 12.7    &  13.2  & 13.4 \\
\bottomrule[0.8pt]
$N_t = N_r =$         &  28 & 32     &  48  & \textbf{64}\\
\specialrule{0.7pt}{0pt}{0pt}
MSE (unit: m$^2$)         &  0.0236  & 0.0165     &  0.0090  & \textbf{0.0070}\\
\hline
Time (unit: ms)          &  15.2  & 16.9       &  20.0  & \textbf{29.2} \\
\toprule[0.8pt]
\multicolumn{5}{l}{
Note: Boldfaced values are baselines in Table \ref{tab:scenario-cw}.
}
\end{tabular}
\end{table}

The tracking results for different numbers of subcarriers are given in Table \ref{tab:scenario-cw-sub}. It suggests that, given subcarrier spacing, the larger the number of subcarriers (i.e., the larger the bandwidth), the smaller the average tracking error. However, this performance improvement comes with an increase in computational time. 

\begin{table}[!htbp]
\centering
\caption{Tracking results of PF-SLTR when $N_c$ varies}
\label{tab:scenario-cw-sub}
\begin{tabular}{lcccccc}
\bottomrule[0.8pt]
$N_c$        & 32  & 48  & 64     &  128  & \textbf{256}\\
\specialrule{0.7pt}{0pt}{0pt}
MSE (unit: m$^2$)       & (Diverged)  & 0.2643  & 0.0797     & 0.0236   & \textbf{0.0070}\\
\hline
Time (unit: ms)         & 9.8   & 11.2    & 12.5      & 19.2  & \textbf{29.2} \\
\toprule[0.8pt]
\multicolumn{6}{l}{
Note: Boldfaced values are baselines in Table \ref{tab:scenario-cw}.
}
\end{tabular}
\end{table}

The tracking results for different numbers of pulses are given in Table \ref{tab:scenario-cw-Np}, suggesting that a larger number of pulses leads to a smaller average tracking error. This is because when $N_p \ge 2$, Doppler processing is enabled, and the additional information in the radial-velocity domain can be exploited. However, this performance improvement incurs substantial computational costs. Since using range and DoA information is sufficient to determine a target, reliable tracking can still be achieved even when Doppler processing is disabled.

\begin{table}[!htbp]
\centering
\caption{Tracking results of PF-SLTR when $N_p$ varies}
\label{tab:scenario-cw-Np}
\begin{tabular}{lccccc}
\bottomrule[0.8pt]
$N_p$        & \textbf{1} &   2  & 3 & 4  & 5   \\
\specialrule{0.7pt}{0pt}{0pt}
MSE (unit: m$^2$)       & \textbf{0.0070}  &   0.0063 & 0.0053 & 0.0051   & 0.0047    \\
\hline
Time (unit: ms)         & \textbf{29.2}   &   64.8 & 132.3 & 170.5  & 193.3  \\
\toprule[0.8pt]
\multicolumn{6}{l}{
Note: Boldfaced values are baselines in Table \ref{tab:scenario-cw}.
}
\end{tabular}
\end{table}

The tracking results for different values of $\xi$ in \eqref{eq:par-representation-posterior-weight} and \eqref{eq:posterior-weights} are given in Table \ref{tab:scenario-cw-xi}. We do not report the running times because the value of $\xi$ does not essentially affect the computational complexity. The results show that the value of $\xi$, which balances the relative importance of the $\vec y_k$-data-evidence compared to the prior belief [cf. \eqref{eq:bayes-rule-opt-view-h}], can potentially further improve the performance of PF-SLTR; see also Theorem \ref{thm:role-xi}, as well as the texts below it and Example \ref{ex:role-xi}. However, the empirically best value of $\xi$ is problem-specific and to be tuned case by case. For the specific scenario in Fig. \ref{fig:CW_trajectory}, we empirically found that $\xi = 6$ works best among tested values.

\begin{table}[!htbp]
\centering
\caption{Tracking results of PF-SLTR when $\xi$ varies}
\label{tab:scenario-cw-xi}
\begin{tabular}{lcccccccc}
\bottomrule[0.8pt]
$\xi$     & 0.50       & 0.75          & \textbf{1}       & 1.25        & 1.50 \\
\specialrule{0.7pt}{0pt}{0pt}
MSE (unit: m$^2$)       & 0.0105     & 0.0096        &  \textbf{0.0070} & 0.0059       &  0.0059 \\
\hline
\specialrule{0.7pt}{0pt}{0pt}
$\xi$     & 2.00        & 4.00          & 6.00             & 8.00         & 10.00  \\
\specialrule{0.7pt}{0pt}{0pt}
MSE (unit: m$^2$)       & 0.0057     & 0.0036        & 0.0032           & 0.0046       &  0.0052    \\
\toprule[0.8pt]
\multicolumn{6}{l}{
Note: Boldfaced values are baselines in Table \ref{tab:scenario-cw}.
}
\end{tabular}
\end{table}

\blue{For visualization of Tables \ref{tab:scenario-cw-par}-\ref{tab:scenario-cw-xi}, see Appendix
D in the online supplementary materials.}



\section{Conclusions}\label{sec:conclusion}
This article proposes a new particle filtering framework to address core challenges in signal-level target tracking using MIMO multi-carrier pulse-Doppler radar. By adopting an optimization-centric interpretation of Bayes’ rule, the proposed method eliminates the need for explicit likelihood models and avoids augmenting nuisance parameters (i.e., complex path gains) into the state vector. Instead, the filtering process is driven by evaluating the data evidence of each hypothesized target state using domain-tailored cost functions. This shift achieves significant performance improvements in both running times and tracking accuracy, relative to conventional approaches. Moreover, the proposed framework is adapted for implementation in MIMO-OFDM systems, enabling practical integration with existing communication infrastructures to realize joint sensing and communication capabilities. Both the pulsed paradigm for tracking far and fast targets and the continuous-wave scheme for tracking near and slow targets are discussed. Simulation results confirm the effectiveness of the proposed approaches in diverse MIMO-OFDM ISAC configurations, demonstrating the real-time feasibility and reliable tracking performance under unknown likelihood models.

\bibliographystyle{IEEEtran}
\bibliography{References}

@book{richards2022fundamentals,
  title={Fundamentals of Radar Signal Processing},
  author={Richards, Mark A.},
  edition={3rd},
  year={2022},
  publisher={McGraw-Hill New York}
}

@article{knoblauch2022optimization,
  title={An optimization-centric view on {Bayes}' rule: Reviewing and generalizing variational inference},
  author={Knoblauch, Jeremias and Jewson, Jack and Damoulas, Theodoros},
  journal={Journal of Machine Learning Research},
  volume={23},
  number={132},
  pages={1--109},
  year={2022}
}

@article{germain2016pac,
  title={{PAC-Bayesian} theory meets {Bayesian} inference},
  author={Germain, Pascal and Bach, Francis and Lacoste, Alexandre and Lacoste-Julien, Simon},
  journal={Advances in Neural Information Processing Systems},
  volume={29},
  year={2016}
}

@article{arulampalam2002tutorial,
  title={A tutorial on particle filters for online nonlinear/non-{Gaussian} {Bayesian} tracking},
  author={Arulampalam, M Sanjeev and Maskell, Simon and Gordon, Neil and Clapp, Tim},
  journal={IEEE Trans. Signal Processing},
  volume={50},
  number={2},
  pages={174--188},
  year={2002},
  publisher={IEEE}
}

@article{san2007mimo,
  title={{MIMO} radar ambiguity functions},
  author={San Antonio, Geoffrey and Fuhrmann, Daniel R and Robey, Frank C},
  journal={IEEE J. Sel. Topics Signal Processing},
  volume={1},
  number={1},
  pages={167--177},
  year={2007},
  publisher={IEEE}
}

@inbook{michael2013mimo,
    author = {Michael S. Davis},
    year = {2013},
    title = {{MIMO} Radar},
    booktitle = {Principles of Modern Radar: Advanced Techniques},
    chapter = {Chapter 4},
    pages = {119-145},
    doi = {10.1049/SBRA020E_ch4},
    publisher = {SciTech Publishing, Edison, NJ.}
}

@article{liu2023seventy,
  title={Seventy years of radar and communications: The road from separation to integration},
  author={Liu, Fan and Zheng, Le and Cui, Yuanhao and Masouros, Christos and Petropulu, Athina P and Griffiths, Hugh and Eldar, Yonina C},
  journal={IEEE Signal Processing Mag.},
  volume={40},
  number={5},
  pages={106--121},
  year={2023},
  publisher={IEEE}
}

@article{mishra2019toward,
  title={Toward millimeter-wave joint radar communications: A signal processing perspective},
  author={Mishra, Kumar Vijay and Shankar, MR Bhavani and Koivunen, Visa and Ottersten, Bjorn and Vorobyov, Sergiy A},
  journal={IEEE Signal Processing Mag.},
  volume={36},
  number={5},
  pages={100--114},
  year={2019},
  publisher={IEEE}
}

@article{dong2025communication,
  title={Communication-assisted sensing in {6G} networks},
  author={Dong, Fuwang and Liu, Fan and Lu, Shihang and Xiong, Yifeng and Zhang, Qixun and Feng, Zhiyong and Gao, Feifei},
  journal={IEEE J. Sel. Areas in Commun.},
  year={2025},
  volume={43},
  issue={4},
  pages={1371-1386},
  publisher={IEEE}
}

@article{liu2020radar,
  title={Radar-assisted predictive beamforming for vehicular links: Communication served by sensing},
  author={Liu, Fan and Yuan, Weijie and Masouros, Christos and Yuan, Jinhong},
  journal={IEEE Trans. Wireless Commun.},
  volume={19},
  number={11},
  pages={7704--7719},
  year={2020},
  publisher={IEEE}
}

@book{doucet2001sequential,
  title={Sequential Monte Carlo methods in Practice},
  author={Doucet, Arnaud and De Freitas, Nando and Gordon, Neil James and others},
  volume={1},
  number={2},
  year={2001},
  publisher={Springer}
}

@article{snyder2008obstacles,
  title={Obstacles to high-dimensional particle filtering},
  author={Snyder, Chris and Bengtsson, Thomas and Bickel, Peter and Anderson, Jeff},
  journal={Monthly Weather Review},
  volume={136},
  number={12},
  pages={4629--4640},
  year={2008}
}

@article{li2003survey,
  title={Survey of maneuvering target tracking: {Part I}. Dynamic models},
  author={Li, X Rong and Jilkov, Vesselin P},
  journal={IEEE Trans. Aerosp. Electron. Syst.},
  volume={39},
  number={4},
  pages={1333--1364},
  year={2003},
  publisher={IEEE}
}

@article{stoica2007probing,
  title={On probing signal design for {MIMO} radar},
  author={Stoica, Petre and Li, Jian and Xie, Yao},
  journal={IEEE Trans. Signal Processing},
  volume={55},
  number={8},
  pages={4151--4161},
  year={2007},
  publisher={IEEE}
}

@article{lu2014overview,
  title={An overview of massive {MIMO}: Benefits and challenges},
  author={Lu, Lu and Li, Geoffrey Ye and Swindlehurst, A Lee and Ashikhmin, Alexei and Zhang, Rui},
  journal={IEEE J. Sel. Topics Signal Processing},
  volume={8},
  number={5},
  pages={742--758},
  year={2014},
  publisher={IEEE}
}

@article{li1999channel,
  title={Channel estimation for {OFDM} systems with transmitter diversity in mobile wireless channels},
  author={Li, Ye and Seshadri, Nambi and Ariyavisitakul, Sirikiat},
  journal={IEEE J. Sel. Areas in Commun.},
  volume={17},
  number={3},
  pages={461--471},
  year={1999},
  publisher={IEEE}
}

@article{li2002mimo,
  title={{MIMO-OFDM} for wireless communications: Signal detection with enhanced channel estimation},
  author={Li, Y Geoffrey and Winters, Jack H and Sollenberger, Nelson R},
  journal={IEEE Trans. Commun.},
  volume={50},
  number={9},
  pages={1471--1477},
  year={2002},
  publisher={IEEE}
}

@article{xu2008target,
  title={Target detection and parameter estimation for {MIMO} radar systems},
  author={Xu, Luzhou and Li, Jian and Stoica, Petre},
  journal={IEEE Trans. Aerosp. Electron. Syst.},
  volume={44},
  number={3},
  pages={927--939},
  year={2008},
  publisher={IEEE}
}

@article{cui2013mimo,
  title={{MIMO} radar waveform design with constant modulus and similarity constraints},
  author={Cui, Guolong and Li, Hongbin and Rangaswamy, Muralidhar},
  journal={IEEE Trans. Signal Processing},
  volume={62},
  number={2},
  pages={343--353},
  year={2013},
  publisher={IEEE}
}

@article{koivunen2024multicarrier,
  title={Multicarrier {ISAC}: Advances in waveform design, signal processing, and learning under nonidealities},
  author={Koivunen, Visa and Keskin, Musa Furkan and Wymeersch, Henk and Valkama, Mikko and Gonz{\'a}lez-Prelcic, Nuria},
  journal={IEEE Signal Processing Mag.},
  volume={41},
  number={5},
  pages={17--30},
  year={2024},
  publisher={IEEE}
}

@article{sturm2011waveform,
  title={Waveform design and signal processing aspects for fusion of wireless communications and radar sensing},
  author={Sturm, Christian and Wiesbeck, Werner},
  journal={Proc. IEEE},
  volume={99},
  number={7},
  pages={1236--1259},
  year={2011},
  publisher={IEEE}
}

@article{mercier2020comparison,
  title={Comparison of correlation-based {OFDM} radar receivers},
  author={Mercier, Steven and Bidon, St{\'e}phanie and Roque, Damien and Enderli, Cyrille},
  journal={IEEE Trans. Aerosp. Electron. Syst.},
  volume={56},
  number={6},
  pages={4796--4813},
  year={2020},
  publisher={IEEE}
}

@article{keskin2023monostatic,
  title={Monostatic sensing with {OFDM} under phase noise: From mitigation to exploitation},
  author={Keskin, Musa Furkan and Wymeersch, Henk and Koivunen, Visa},
  journal={IEEE Trans. Signal Processing},
  volume={71},
  pages={1363--1378},
  year={2023},
  publisher={IEEE}
}

@article{zhang2025sensing,
  title={Sensing-Assisted Intelligent Transportation System With Adaptive Power Allocation and Automatic Beam Control},
  author={Zhang, Zhibo and Chen, Leyan and Xing, Jin and Liu, Kai and Chang, Qing},
  journal={IEEE Trans. Intell. Transp. Syst.},
  year={2025},
  pages={8820-8833},
  volume={26},
  issue={6},
  publisher={IEEE}
}

@book{stone2013bayesian,
  title={Bayesian Multiple Target Tracking},
  author={Stone, Lawrence D and Streit, Roy L and Corwin, Thomas L and Bell, Kristine L},
  year={2014},
  edition={2},
  publisher={Artech House}
}

@inproceedings{ratpunpairoj2015particle,
  title={A particle filter for objects tracking in cognitive {MIMO} system},
  author={Ratpunpairoj, Paopat and Kongprawechnon, Waree and Fukawa, Kazuhiko and Kaemarungsi, Kamol},
  booktitle={2015 6th International Conference of Information and Communication Technology for Embedded Systems (IC-ICTES)},
  pages={1--4},
  year={2015},
  organization={IEEE}
}

@article{huleihel2013optimal,
  title={Optimal adaptive waveform design for cognitive {MIMO} radar},
  author={Huleihel, Wasim and Tabrikian, Joseph and Shavit, Reuven},
  journal={IEEE Trans. Signal Processing},
  volume={61},
  number={20},
  pages={5075--5089},
  year={2013},
  publisher={IEEE}
}

@article{herbert2017mmse,
  title={{MMSE} adaptive waveform design for active sensing with applications to {MIMO} radar},
  author={Herbert, Steven and Hopgood, James R and Mulgrew, Bernard},
  journal={IEEE Trans. Signal Processing},
  volume={66},
  number={5},
  pages={1361--1373},
  year={2017},
  publisher={IEEE}
}

@article{wang2023distributionally,
  title={Distributionally robust state estimation for jump linear systems},
  author={Wang, Shixiong},
  journal={IEEE Trans. Signal Processing},
  volume={71},
  pages={3835--3851},
  year={2023},
  publisher={IEEE}
}

@article{herbert2018computationally,
  title={Computationally simple {MMSE} ({A}-Optimal) adaptive beam-pattern design for {MIMO} active sensing systems via a linear-{Gaussian} approximation},
  author={Herbert, Steven and Hopgood, James R and Mulgrew, Bernard},
  journal={IEEE Trans. Signal Processing},
  volume={66},
  number={18},
  pages={4935--4945},
  year={2018},
  publisher={IEEE}
}

@inproceedings{li2001survey,
  title={Survey of Maneuvering Target Tracking: {Part III}. Measurement Models},
  author={Li, X Rong and Jilkov, Vesselin P},
  booktitle={Signal and Data Processing of Small Targets 2001},
  volume={4473},
  pages={423--446},
  year={2001},
  organization={SPIE}
}

@article{davey2013track,
  title={Track-Before-Detect Techniques},
  author={Davey, Samuel J and Rutten, Mark G and Gordon, Neil J},
  journal={Integrated Tracking, Classification, and Sensor Management},
  pages={311--362},
  year={2013},
  publisher={Wiley Online Library}
}

@article{davey2007comparison,
  title={A comparison of detection performance for several track-before-detect algorithms},
  author={Davey, Samuel J and Rutten, Mark G and Cheung, Brian},
  journal={EURASIP Journal on Advances in Signal Processing},
  volume={2008},
  pages={1--10},
  year={2007},
  publisher={Springer}
}

@article{zheng2017super,
  title={Super-resolution delay-{Doppler} estimation for {OFDM} passive radar},
  author={Zheng, Le and Wang, Xiaodong},
  journal={IEEE Trans. Signal Processing},
  volume={65},
  number={9},
  pages={2197--2210},
  year={2017},
  publisher={IEEE}
}

@article{zhang2021overview,
  title={An overview of signal processing techniques for joint communication and radar sensing},
  author={Zhang, J Andrew and Liu, Fan and Masouros, Christos and Heath, Robert W and Feng, Zhiyong and Zheng, Le and Petropulu, Athina},
  journal={IEEE J. Sel. Topics Signal Processing},
  volume={15},
  number={6},
  pages={1295--1315},
  year={2021},
  publisher={IEEE}
}

@article{wu2022super,
  title={Super-resolution {TOA} and {AOA} estimation for {OFDM} radar systems based on compressed sensing},
  author={Wu, Min and Hao, Chengpeng},
  journal={IEEE Trans. Aerosp. Electron. Syst.},
  volume={58},
  number={6},
  pages={5730--5740},
  year={2022},
  publisher={IEEE}
}

@article{zhang2024input,
  title={Input Distribution Optimization in {OFDM} Dual-Function Radar-Communication Systems},
  author={Zhang, Yumeng and Aditya, Sundar and Clerckx, Bruno},
  journal={IEEE Trans. Signal Processing},
  pages={5258-5273},
  volume={72},
  year={2024},
  publisher={IEEE}
}

@inproceedings{lellouch2014processing,
  title={Processing alternatives in {OFDM} radar},
  author={Lellouch, Gabriel and Mishra, Amit and Inggs, Mike},
  booktitle={2014 International Radar Conference},
  pages={1--6},
  year={2014},
  organization={IEEE}
}

@article{bissiri2016general,
  title={A general framework for updating belief distributions},
  author={Bissiri, Pier Giovanni and Holmes, Chris C and Walker, Stephen G},
  journal={Journal of the Royal Statistical Society Series B: Statistical Methodology},
  volume={78},
  number={5},
  pages={1103--1130},
  year={2016},
  publisher={Oxford University Press}
}

@article{baek2023generalized,
  title={Generalized {Bayes} approach to inverse problems with model misspecification},
  author={Baek, Youngsoo and Aquino, Wilkins and Mukherjee, Sayan},
  journal={Inverse Problems},
  volume={39},
  number={10},
  pages={105011},
  year={2023},
  publisher={IOP Publishing}
}

@article{davey2013snr,
  title={{SNR} limits on {Kalman} filter detect-then-track},
  author={Davey, Samuel J},
  journal={IEEE Signal Processing Lett.},
  volume={20},
  number={8},
  pages={767--770},
  year={2013},
  publisher={IEEE}
}

@incollection{martin2022direct,
  title={Direct {Gibbs} posterior inference on risk minimizers: Construction, concentration, and calibration},
  author={Martin, Ryan and Syring, Nicholas},
  booktitle={Handbook of Statistics},
  volume={47},
  pages={1--41},
  year={2022},
  publisher={Elsevier}
}

@article{liu2025uncovering,
  title={Uncovering the Iceberg in the Sea: Fundamentals of Pulse Shaping and Modulation Design for Random {ISAC} Signals},
  author={Liu, Fan and Xiong, Yifeng and Lu, Shihang and Li, Shuangyang and Yuan, Weijie and Masouros, Christos and Jin, Shi and Caire, Giuseppe},
  journal={IEEE Trans. Signal Processing},
  pages={2511-2526},
  volume={73},
  year={2025}
}

@article{liu2025cpofdm,
  title={{CP-OFDM} Achieves the Lowest Average Ranging Sidelobe Under {QAM/PSK} Constellations},
  author={Liu, Fan and Zhang, Ying and Xiong, Yifeng and Li, Shuangyang and Yuan, Weijie and Gao, Feifei and Jin, Shi and Caire, Giuseppe},
  journal={IEEE Trans. Inform. Theory},
  pages={6950-6967},
  volume={71},
  issue={9},
  year={2025}
}

@article{berger2010signal,
  title={Signal processing for passive radar using {OFDM} waveforms},
  author={Berger, Christian R and Demissie, Bruno and Heckenbach, J{\"o}rg and Willett, Peter and Zhou, Shengli},
  journal={IEEE J. Sel. Topics Signal Processing},
  volume={4},
  number={1},
  pages={226--238},
  year={2010},
  publisher={IEEE}
}

@article{xu2022experimental,
  title={An experimental proof of concept for integrated sensing and communications waveform design},
  author={Xu, Tongyang and Liu, Fan and Masouros, Christos and Darwazeh, Izzat},
  journal={IEEE Open J. Commun. Soc.},
  volume={3},
  pages={1643--1655},
  year={2022},
  publisher={IEEE}
}

@inproceedings{gu2023beam,
  title={Beam Tracking for Distributed Millimeter-Wave {MIMO-OFDM} Systems with Unscented Particle Filter},
  author={Gu, Fei and Zhu, Pengcheng and Jiang, Yanxiang and Wang, Dongming and Wang, Yan},
  booktitle={2023 IEEE 23rd International Conference on Communication Technology (ICCT)},
  pages={534--539},
  year={2023},
  organization={IEEE}
}

@inproceedings{kwon2019particle,
  title={Particle filter based track-before-detect method in the range-{Doppler} domain},
  author={Kwon, Jihoon and Kwak, Nojun and Yang, Eunjung and Kim, Kwansung},
  booktitle={2019 IEEE Radar Conference (RadarConf)},
  pages={1--5},
  year={2019},
  organization={IEEE}
}

@article{ito2020multi,
  title={A multi-target track-before-detect particle filter using superpositional data in non-{Gaussian} noise},
  author={Ito, Nobutaka and Godsill, Simon},
  journal={IEEE Signal Processing Lett.},
  volume={27},
  pages={1075--1079},
  year={2020},
  publisher={IEEE}
}

@book{simon2006optimal,
  title={Optimal State Estimation: {Kalman}, {H}-Infinity, and Nonlinear Approaches},
  author={Simon, Dan},
  year={2006},
  publisher={John Wiley \& Sons}
}

@article{lu2024random,
  title={Random {ISAC} signals deserve dedicated precoding},
  author={Lu, Shihang and Liu, Fan and Dong, Fuwang and Xiong, Yifeng and Xu, Jie and Liu, Ya-Feng and Jin, Shi},
  journal={IEEE Trans. Signal Processing},
  pages={3453-3469},
  volume={72},
  year={2024},
  publisher={IEEE}
}

@article{xiong2023fundamental,
  title={On the fundamental tradeoff of integrated sensing and communications under {Gaussian} channels},
  author={Xiong, Yifeng and Liu, Fan and Cui, Yuanhao and Yuan, Weijie and Han, Tony Xiao and Caire, Giuseppe},
  journal={IEEE Trans. Inform. Theory},
  volume={69},
  number={9},
  pages={5723--5751},
  year={2023},
  publisher={IEEE}
}

@inproceedings{hu2022low,
  title={Low-{PAPR} {DFRC MIMO-OFDM} waveform design for integrated sensing and communications},
  author={Hu, Xiaoyan and Masouros, Christos and Liu, Fan and Nissel, Ronald},
  booktitle={ICC 2022-IEEE International Conference on Communications},
  pages={1599--1604},
  year={2022},
  organization={IEEE}
}

@article{ahmadipour2022information,
  title={An information-theoretic approach to joint sensing and communication},
  author={Ahmadipour, Mehrasa and Kobayashi, Mari and Wigger, Michele and Caire, Giuseppe},
  journal={IEEE Trans. Inform. Theory},
  volume={70},
  number={2},
  pages={1124--1146},
  year={2022},
  publisher={IEEE}
}

@article{olson2023coverage,
  title={Coverage and rate of joint communication and parameter estimation in wireless networks},
  author={Olson, Nicholas R and Andrews, Jeffrey G and Heath, Robert W},
  journal={IEEE Trans. Inform. Theory},
  volume={70},
  number={1},
  pages={206--243},
  year={2023},
  publisher={IEEE}
}

@article{li2024mimo,
  title={{MIMO-OFDM ISAC} waveform design for range-{Doppler} sidelobe suppression},
  author={Li, Peishi and Li, Ming and Liu, Rang and Liu, Qian and Swindlehurst, A Lee},
  journal={IEEE Trans. Wireless Commun.},
  pages={1001-1015},
  volume={24},
  year={2025},
  publisher={IEEE}
}

@article{liu2021cramer,
  title={{Cram{\'e}r-Rao} bound optimization for joint radar-communication beamforming},
  author={Liu, Fan and Liu, Ya-Feng and Li, Ang and Masouros, Christos and Eldar, Yonina C},
  journal={IEEE Trans. Signal Processing},
  volume={70},
  pages={240--253},
  year={2021},
  publisher={IEEE}
}

@article{xiao2024novel,
  title={A novel joint angle-range-velocity estimation method for {MIMO-OFDM ISAC} systems},
  author={Xiao, Zichao and Liu, Rang and Li, Ming and Liu, Qian and Swindlehurst, A Lee},
  journal={IEEE Trans. Signal Processing},
  year={2024},
  pages={3805-3818},
  volume={72},
  publisher={IEEE}
}

@article{du2024reshaping,
  title={Reshaping the {ISAC} tradeoff under {OFDM} signaling: A probabilistic constellation shaping approach},
  author={Du, Zhen and Liu, Fan and Xiong, Yifeng and Han, Tony Xiao and Eldar, Yonina C and Jin, Shi},
  journal={IEEE Trans. Signal Processing},
  volume={72},
  pages={4782-4797},
  year={2024},
  publisher={IEEE}
}

@article{heath2016overview,
  title={An overview of signal processing techniques for millimeter wave {MIMO} systems},
  author={Heath, Robert W and Gonzalez-Prelcic, Nuria and Rangan, Sundeep and Roh, Wonil and Sayeed, Akbar M},
  journal={IEEE J. Sel. Topics Signal Processing},
  volume={10},
  number={3},
  pages={436--453},
  year={2016},
  publisher={IEEE}
}

@inproceedings{godsill2019particle,
  title={Particle filtering: the first 25 years and beyond},
  author={Godsill, Simon},
  booktitle={ICASSP 2019-2019 IEEE International Conference on Acoustics, Speech and Signal Processing (ICASSP)},
  pages={7760--7764},
  year={2019},
  organization={IEEE}
}

@article{vermaak2005monte,
  title={{Monte Carlo} filtering for multi target tracking and data association},
  author={Vermaak, Jaco and Godsill, Simon J and Perez, Patrick},
  journal={IEEE Trans. Aerosp. Electron. Syst.},
  volume={41},
  number={1},
  pages={309--332},
  year={2005},
  publisher={IEEE}
}

@article{schon2005marginalized,
  title={Marginalized particle filters for mixed linear/nonlinear state-space models},
  author={Schon, Thomas and Gustafsson, Fredrik and Nordlund, P-J},
  journal={IEEE Trans. Signal Processing},
  volume={53},
  number={7},
  pages={2279--2289},
  year={2005},
  publisher={IEEE}
}

@article{cappe2007overview,
  title={An overview of existing methods and recent advances in sequential {Monte Carlo}},
  author={Capp{\'e}, Olivier and Godsill, Simon J and Moulines, Eric},
  journal={Proc. IEEE},
  volume={95},
  number={5},
  pages={899--924},
  year={2007},
  publisher={IEEE}
}

@article{wang2017survey,
  title={A survey of recent advances in particle filters and remaining challenges for multitarget tracking},
  author={Wang, Xuedong and Li, Tiancheng and Sun, Shudong and Corchado, Juan M},
  journal={Sensors},
  volume={17},
  number={12},
  pages={2707},
  year={2017},
  publisher={MDPI}
}

@article{conti2019soft,
  title={Soft information for localization-of-things},
  author={Conti, Andrea and Mazuelas, Santiago and Bartoletti, Stefania and Lindsey, William C and Win, Moe Z},
  journal={Proc. IEEE},
  volume={107},
  number={11},
  pages={2240--2264},
  year={2019},
  publisher={IEEE}
}

\newpage

\onecolumn

\setcounter{page}{1}
\begin{figure}
\centering
{\large \bf Supplementary Materials}
\end{figure}

\appendices

\section{Proof of Proposition \captext{\ref{prop:h1}}}\label{append:h1}
\begin{proof}
Solve \eqref{eq:cost-func-time-domain} with respect to $\gamma$, yielding
\begin{equation}\label{eq:h1-simplified-2}
    \begin{array}{cl}
    \rho_1(\vec y, \vec \varphi) &= 
    \displaystyle \sum^{N_p - 1}_{p = 0} \sum^{L - 1}_{l = 0} \vec y^\H_p(l) \vec y_p(l)  -  
             \frac{
            \displaystyle
            \left| \sum^{N_p - 1}_{p = 0} \sum^{L - 1}_{l = 0} \vec y^\H_p(l) \vec b(\theta) \vec a^\H(\theta) \vec s_p(l - \tau) e^{j 2 \pi p T_r \nu} \right|^2
            }{\displaystyle
            \sum^{N_p - 1}_{p = 0} \sum^{L - 1}_{l = 0} \vec a^\H(\theta) \vec s_p(l - \tau) \vec s^\H_p(l - \tau) \vec a(\theta) \vec b^\H(\theta) \vec b(\theta)
            } \\
     &\ge 0.
    \end{array}
\end{equation}
The denominator in the second line equals 
\[
    N_r N_p \vec a^\H(\theta) \sum^{L_{\text{ss}} - 1}_{l = 0}  \vec s_p(l) \vec s^\H_p(l) \vec a(\theta) = N_r N_p L_{\text{ss}} \vec a^\H(\theta) \math R_s \vec a(\theta).
\]
Note that $L_{\text{ss}} \le L$ and $\vec s_p(l) = \vec 0$ if $L_{\text{ss}} \le l \le L-1$, for every $p \in [N_p]$. This completes the proof. \stp
\end{proof}

\section{Proof of Proposition \captext{\ref{prop:h4-property}}}\label{append:h4-property}

\begin{proof}
When scatterers (including targets of non-interest) are well-separated from the ToI, the interference energy due to scatterers, after matched filtering, can be negligible within the specified tracking gate \cite{richards2022fundamentals,michael2013mimo,san2007mimo}. Hence, it is equivalent to examine the simplified (i.e., scatterer-free) version of the frequency-domain signal model \eqref{eq:general-radar-measurement-compact-freq}:
\[
    \begin{array}{cl}
        \rvecb y_p(n) &= \beta_0 \vec b(\theta_0) \vec a^\H(\theta_0) \rvecb s_p(n) e^{-j 2 \pi n \Delta \tau_0} e^{j 2 \pi p T_r \nu_0} + \rvecb n_{p}(n),
    \end{array}
\]
where $\rvecb n_{p}(n)$ denotes the pure (i.e., scatterer-free) channel noise; $\rvecb n_{p}(n)$ is uncorrelated with 
$\rvecb s_p(n)$ for every $n \in [N_{\text{cr}}]$ and $p \in [N_p]$. After matched filtering, we have \eqref{eq:tracking-gate-1}-\eqref{eq:tracking-gate-3}.

The first nulls of the Dirichlet function, or the periodic sinc function, in \eqref{eq:tracking-gate-3} is given by the following conditions in each independent dimension: $\frac{N_t}{2} (\sin{\theta_0} - \sin{\theta}) = \pm 1$, $\frac{N_r}{2} (\sin{\theta_0} - \sin{\theta}) = \pm 1$, $B (\tau_0 - \tau) = \pm 1$, and $T_i (\nu - \nu_0) = \pm 1$. Therefore, within the tracking gate
\[
\begin{array}{l}
    \displaystyle \Big[\tau_0 - \frac{1}{B}, \tau_0 + \frac{1}{B}\Big] \times \Big[\nu_0 - \frac{1}{T_i}, \nu_0 + \frac{1}{T_i}\Big] \times
    \displaystyle  \Big[\arcsin\Big(\sin \theta_0 - \frac{2}{N_{tr}}\Big), \arcsin\Big(\sin \theta_0 + \frac{2}{N_{tr}}\Big)\Big],
\end{array}
\]
the Dirichlet function is nonzero and satisfies Rule \ref{rule:cost-func-h}. This completes the proof. \stp

\begin{figure*}[!htbp]
\begin{equation}\label{eq:tracking-gate-1}
    \begin{array}{cl}
        \displaystyle \rvecb y^\H_p(n) \vec b(\theta) \vec a^\H(\theta) \rvecb s_p(n) e^{-j 2 \pi n \Delta \tau} e^{j 2 \pi p T_r \nu} &= \beta^\H_0 \vec b^\H(\theta) \vec b(\theta_0) \vec a^\H(\theta_0) \rvecb s_p(n) \rvecb s^\H_p(n) \vec a(\theta) e^{j 2 \pi n \Delta (\tau_0 - \tau)} e^{j 2 \pi p T_r (\nu - \nu_0)} + \\ 
        
        & \quad \quad \quad \quad \rvecb n^\H_{p}(n) \vec b(\theta) \vec a^\H(\theta) \rvecb s_p(n) e^{-j 2 \pi n \Delta \tau} e^{j 2 \pi p T_r \nu}
    \end{array} 
\end{equation}
\hrule~\\
\begin{equation}\label{eq:tracking-gate-2}
    \begin{array}{l}
        \displaystyle \E_{\{\rvecb y_p(n), \rvecb s_p(n)\}_{\forall p, n}} \left\{ \sum^{N_p - 1}_{p = 0} \sum^{N_{\text{cr}} - 1}_{n = 0} \Big[\rvecb y^\H_p(n) \vec b(\theta) \vec a^\H(\theta) \rvecb s_p(n) e^{-j 2 \pi n \Delta \tau} e^{j 2 \pi p T_r \nu} \Big] \right\} \\
        
        \quad \quad =  \displaystyle \beta^\H_0 \vec b^\H(\theta) \vec b(\theta_0) \vec a^\H(\theta_0) \mat R_s \vec a(\theta) \sum^{N_{\text{cr}} - 1}_{n = 0} e^{j 2 \pi n \Delta (\tau_0 - \tau)} \sum^{N_p - 1}_{p = 0} e^{j 2 \pi p T_r (\nu - \nu_0)} + \\

        \quad \quad \quad \quad \quad \quad \quad \quad  \displaystyle \E_{\{\rvecb n_p(n), \rvecb s_p(n)\}_{\forall p, n}} \left\{\vec a^\H(\theta) \sum^{N_p - 1}_{p = 0} \sum^{N_{\text{cr}} - 1}_{n = 0} \rvecb s_p(n) \rvecb n^\H_{p}(n) e^{-j 2 \pi n \Delta \tau} e^{j 2 \pi p T_r \nu} \vec b(\theta) \right\} \\

        \quad \quad = \displaystyle \beta^\H_0 P_t \vec b^\H(\theta) \vec b(\theta_0) \vec a^\H(\theta_0) \vec a(\theta) \sum^{N_{\text{cr}} - 1}_{n = 0} e^{j 2 \pi n \Delta (\tau_0 - \tau)} \sum^{N_p - 1}_{p = 0} e^{j 2 \pi p T_r (\nu - \nu_0)} + 0 \\

        \quad \quad = \displaystyle \beta^\H_0 P_t \sum^{N_t - 1}_{t = 0} e^{j \pi t (\sin{\theta} - \sin{\theta_0})} \sum^{N_r - 1}_{r = 0} e^{j \pi r (\sin{\theta_0} - \sin{\theta})} \sum^{N_{\text{cr}} - 1}_{n = 0} e^{j 2 \pi n \Delta (\tau_0 - \tau)} \sum^{N_p - 1}_{p = 0} e^{j 2 \pi p T_r (\nu - \nu_0)} \\

        \quad \quad = \displaystyle \beta^\H_0 P_t \sum^{N_t - 1}_{t = 0} e^{j 2 \pi \frac{t}{N_t} \frac{N_t}{2} (\sin{\theta} - \sin{\theta_0})} \sum^{N_r - 1}_{r = 0} e^{j 2 \pi \frac{r}{N_r} \frac{N_r}{2} (\sin{\theta_0} - \sin{\theta})} \sum^{N_{\text{cr}} - 1}_{n = 0} e^{j 2 \pi \frac{n}{N_{\text{cr}}} B (\tau_0 - \tau)} \sum^{N_p - 1}_{p = 0} e^{j 2 \pi \frac{p}{N_p} T_i (\nu - \nu_0)} \\
    \end{array} 
\end{equation}
\hrule~\\
\begin{equation}\label{eq:tracking-gate-3}
    \begin{array}{l}
        \displaystyle \big| \eqref{eq:tracking-gate-2} \big| = \displaystyle |\beta_0| \cdot P_t \cdot \left|\frac{\sin[\pi\cdot(\frac{N_t}{2} (\sin{\theta} - \sin{\theta_0}))]}{\sin[\pi\cdot(\frac{N_t}{2} (\sin{\theta} - \sin{\theta_0}))/N_t]}\right| \cdot  \left|\frac{\sin[\pi\cdot(\frac{N_r}{2} (\sin{\theta_0} - \sin{\theta}))]}{\sin[\pi\cdot(\frac{N_r}{2} (\sin{\theta_0} - \sin{\theta}))/N_r]}\right| \cdot \\
        
        \quad \quad \quad \quad \quad \quad \quad \quad \bigg|\frac{\sin[\pi\cdot(B (\tau_0 - \tau))]}{\sin[\pi\cdot(B (\tau_0 - \tau))/N_{\text{cr}}]}\bigg| \cdot \bigg|\frac{\sin[\pi\cdot(T_i (\nu - \nu_0))]}{\sin[\pi\cdot(T_i (\nu - \nu_0))/N_p]}\bigg| \\
    \end{array} 
\end{equation}
\hrule
\end{figure*}
\end{proof}

\section{Proof of Theorem \captext{\ref{thm:role-xi}}}\label{append:role-xi}
\begin{proof}
We have
\[
\begin{array}{cl}
- \displaystyle \sum^{N_{\text{par}} - 1}_{i = 0} \alpha^{(\xi)}_i \ln \alpha^{(\xi)}_i &= \displaystyle \sum^{N_{\text{par}} - 1}_{i = 0} - \frac{\alpha^\xi_i}{C_\xi} \ln \frac{\alpha^\xi_i}{C_\xi} = \displaystyle \sum^{N_{\text{par}} - 1}_{i = 0} - \frac{\alpha^\xi_i}{C_\xi} \ln {\alpha^\xi_i} + \frac{\alpha^\xi_i}{C_\xi} \ln C_{\xi} \\

&= \displaystyle \frac{C_{\xi} \ln C_{\xi}}{C_{\xi} } - \displaystyle \frac{\xi }{C_\xi} \sum^{N_{\text{par}} - 1}_{i = 0} \alpha^\xi_i \ln {\alpha_i}.
\end{array}
\]
The derivative with respect to $\xi$ is
\[
\begin{array}{l}
\displaystyle \frac{\xi}{C^2_\xi} \left\{\left[\sum^{N_{\text{par}} - 1}_{i = 0} \alpha^\xi_i \ln \alpha_i \right]^2 - \right. \left. \displaystyle \left[\sum^{N_{\text{par}} - 1}_{i = 0} \alpha^\xi_i \right]\left[\sum^{N_{\text{par}} - 1}_{i = 0} \alpha^\xi_i \ln \alpha_i \ln \alpha_i \right] \right\} \le 0,
\end{array}
\]
where the inequality is due to the Cauchy inequality, and the equality holds if and only if $\vec \alpha$ is a uniform distribution. Therefore, the entropy function $\xi \mapsto - \sum^{N_{\text{par}} - 1}_{i = 0} \alpha^{(\xi)}_i \ln \alpha^{(\xi)}_i$ is monotonically decreasing in $\xi \ge 0$.  \stp
\end{proof}

\section{Additional Experimental Results}\label{sec:additional-experiments}
\subsection{CW MIMO-OFDM Case}
The tracking results of KF-ILTR and PF-SLTR are visually illustrated in Fig. \ref{fig:CW_trajectory_tracking_results}. As indicated, PF-SLTR demonstrates a clear advantage over KF-ILTR, especially visible in trajectory and DoA tracking, which agrees with the observations from Table \ref{tab:scenario-cw}.

The visualization of Tables \ref{tab:scenario-cw-par}-\ref{tab:scenario-cw-xi} is shown in Fig. \ref{fig:CW_table_visualizaiton}.

Fig. \ref{fig:CW_trajectory_tracking_failed_results} shows the tracking results of PF-SLTR under the breakdown conditions in Tables \ref{tab:scenario-cw-par}, \ref{tab:scenario-cw-snr}, \ref{tab:scenario-cw-tx}, and \ref{tab:scenario-cw-sub}, corresponding to $N_\text{par} \defeq 12$, the transmit $\text{SNR} \defeq -50$dB, $N_t \defeq N_r \defeq 21$, and $N_c \defeq 32$, respectively. The results suggest that PF-SLTR can initially follow the true trajectory under all four scenarios, but the estimated trajectory eventually deviates from the ground truth. This is because, as the target range increases, the receive SNR quickly decreases due to the range-dependent path loss; see \eqref{eq:path-gain}. Once the range exceeds a certain threshold, the received echoes become too weak to provide sufficiently reliable measurement information for particle weighting, especially when the number of particles is small. Consequently, the particle set drifts away from the true trajectory, leading to tracking failures.

We close this subsection with the following remark on the relation between tracking accuracy and radar resolution.
\begin{remark}
     In the main body of the article, it states that the bandwidth is $B=51.2$ MHz, so the sampling time interval and the round-trip delay resolution are $T_s=1/B=0.019531\mu$s. The corresponding range resolution is therefore $C/(2B)\approx 2.93$m. Moreover, for the $64$-element half-wavelength receive array, the resolution in the $\sin$-DoA domain is $2/64=1/32$. These resolution values can be used as reference scales for interpreting the tracking results. To be specific, results in Tables \ref{tab:scenario-cw-par}$-$\ref{tab:scenario-cw-Np} are consistent with radar intuition: Larger bandwidth improves range discrimination, larger antenna aperture improves angular discrimination and spatial processing capability, and larger time-on-target improves Doppler discrimination. However, we clarify that tracking accuracy and radar resolution are related but not identical concepts. Radar resolution characterizes the ability to separate two closely spaced targets, whereas the tracking error reported in the experiments measures the estimation accuracy of a single ToI over time. Therefore, when the ToI is sufficiently isolated and temporal filtering is exploited, the tracking RMSE (root-MSE) can be smaller than the nominal range/angle/Doppler resolution, which does not contradict classical radar-resolution theory. \stp
\end{remark}

\begin{figure}[!htbp]
    \centering

    \subfigure[Trajectory]{
        \centering
        \includegraphics[height=2.9cm]{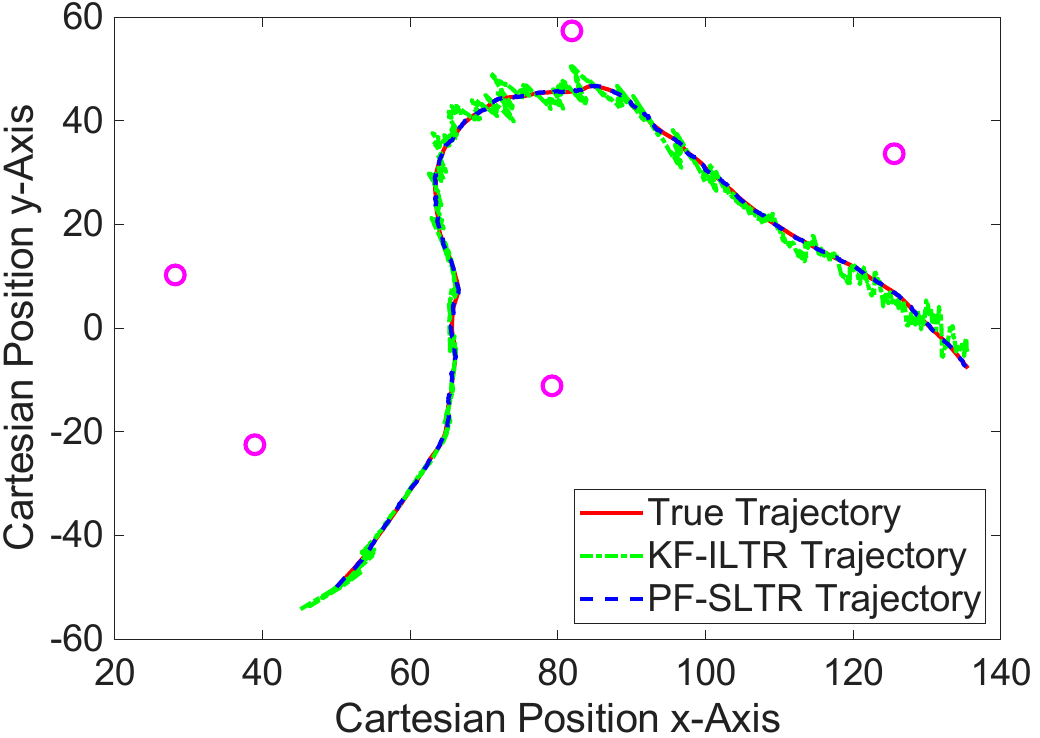}
    }
    \subfigure[Range]{
        \centering
        \includegraphics[height=2.9cm]{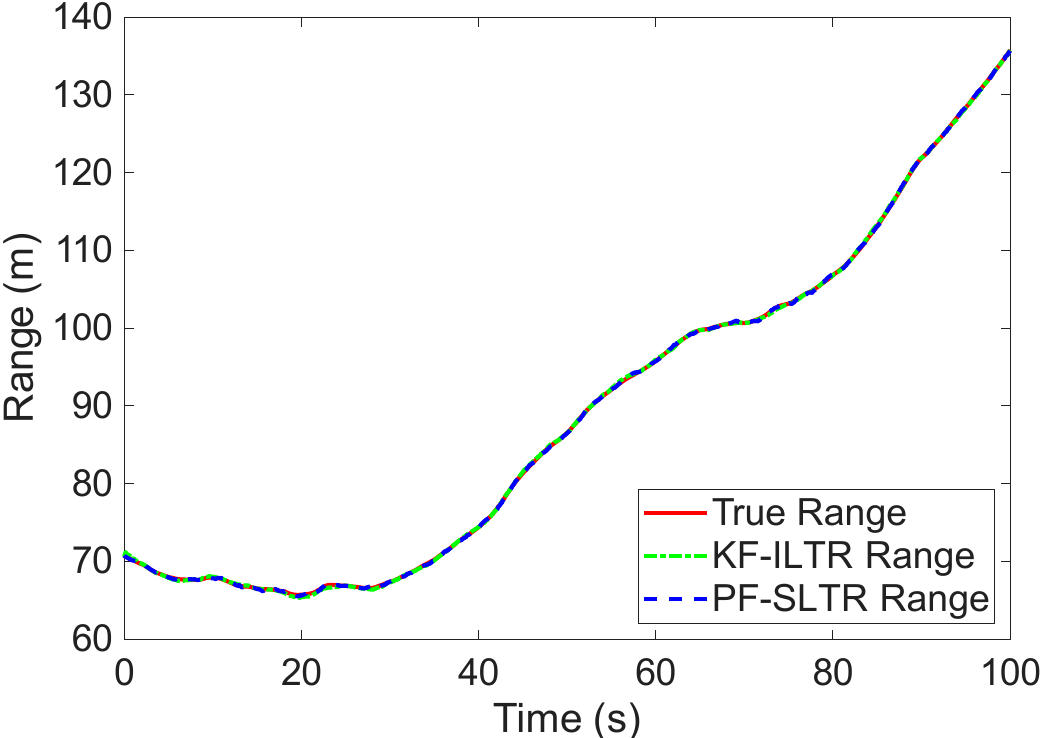}
    }
    \subfigure[Radial Velocity]{
        \centering
        \includegraphics[height=2.9cm]{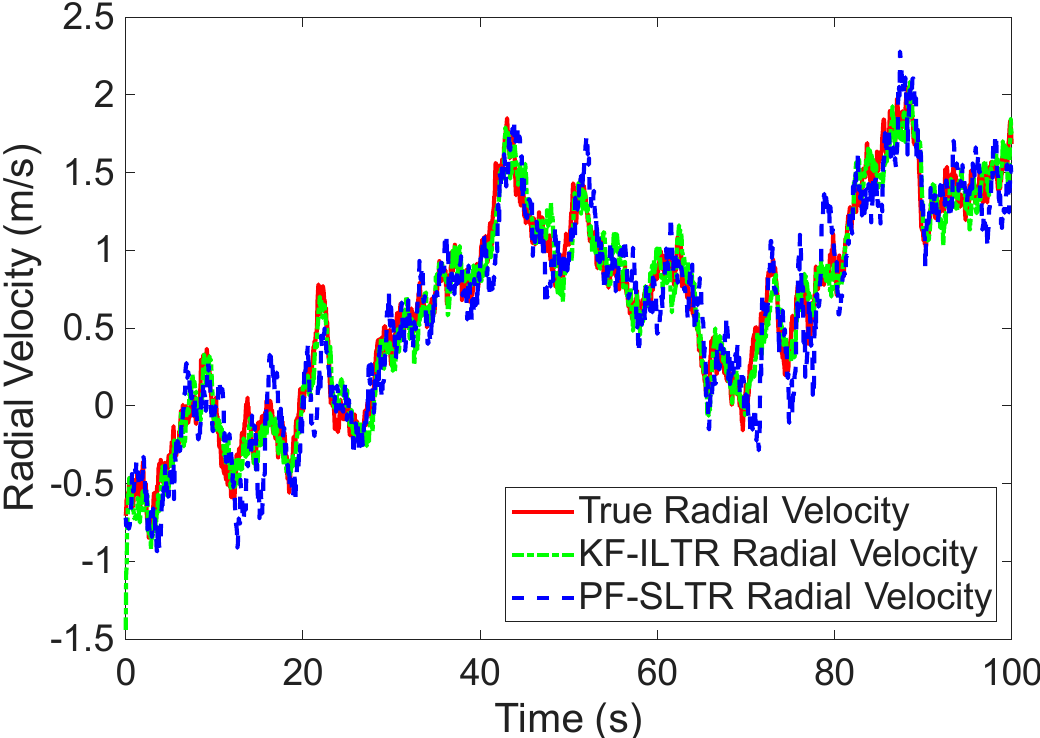}
    }
    \subfigure[DoA]{
        \centering
        \includegraphics[height=2.9cm]{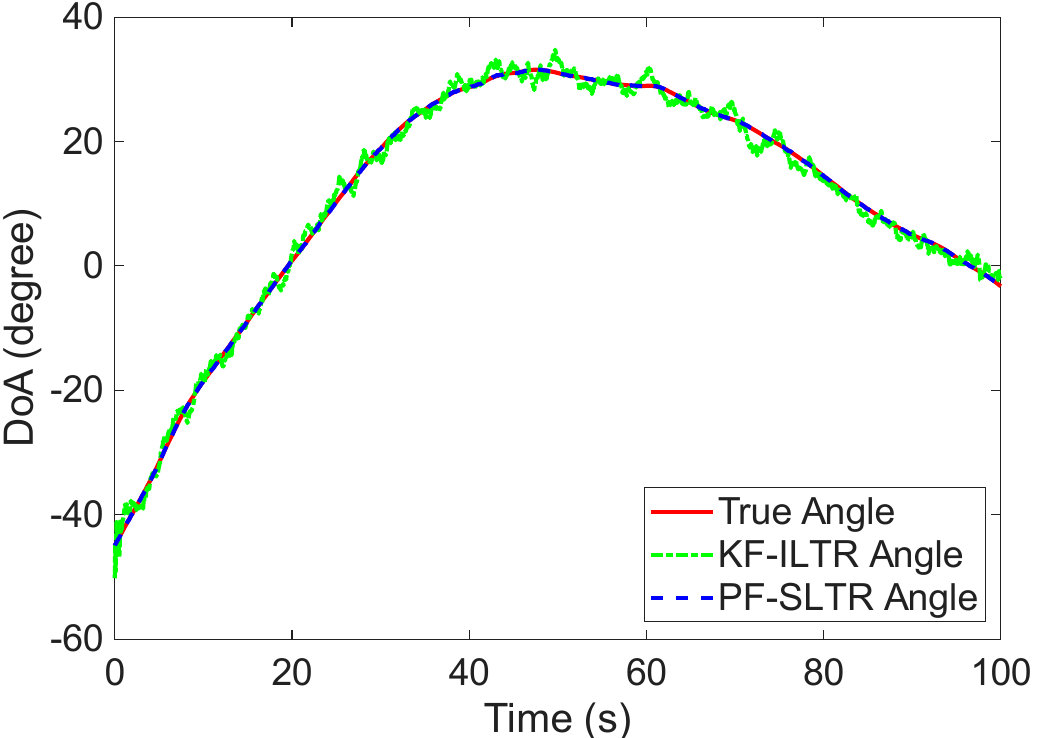}
    }
    
    \caption{Tracking results of KF-ILTR and PF-SLTR. Consistent with the results in Table \ref{tab:scenario-cw}, PF-SLTR demonstrates a clear advantage over KF-ILTR.}
    \label{fig:CW_trajectory_tracking_results}
\end{figure}

\begin{figure}[!htbp]
    \centering

    \subfigure[Table \ref{tab:scenario-cw-par}: $N_{\mathrm{par}}$ vs MSE]{
        \centering
        \includegraphics[height=2.9cm,width=4.8cm]{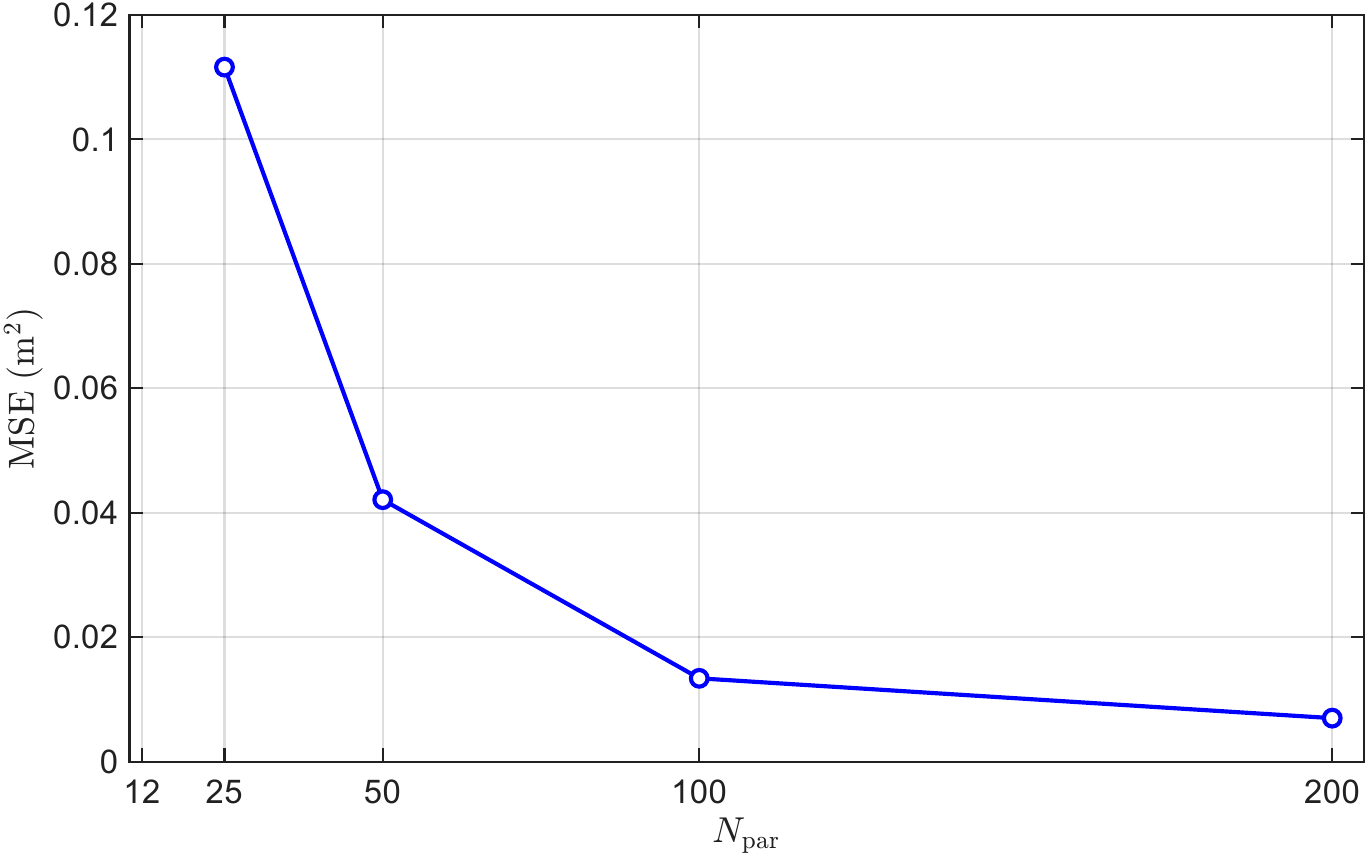}
    }
    \subfigure[Table \ref{tab:scenario-cw-snr}: $\mathrm{SNR}_{\mathrm{tx}}\;(\mathrm{dB})$ vs MSE]{
        \centering
        \includegraphics[height=2.9cm,width=4.8cm]{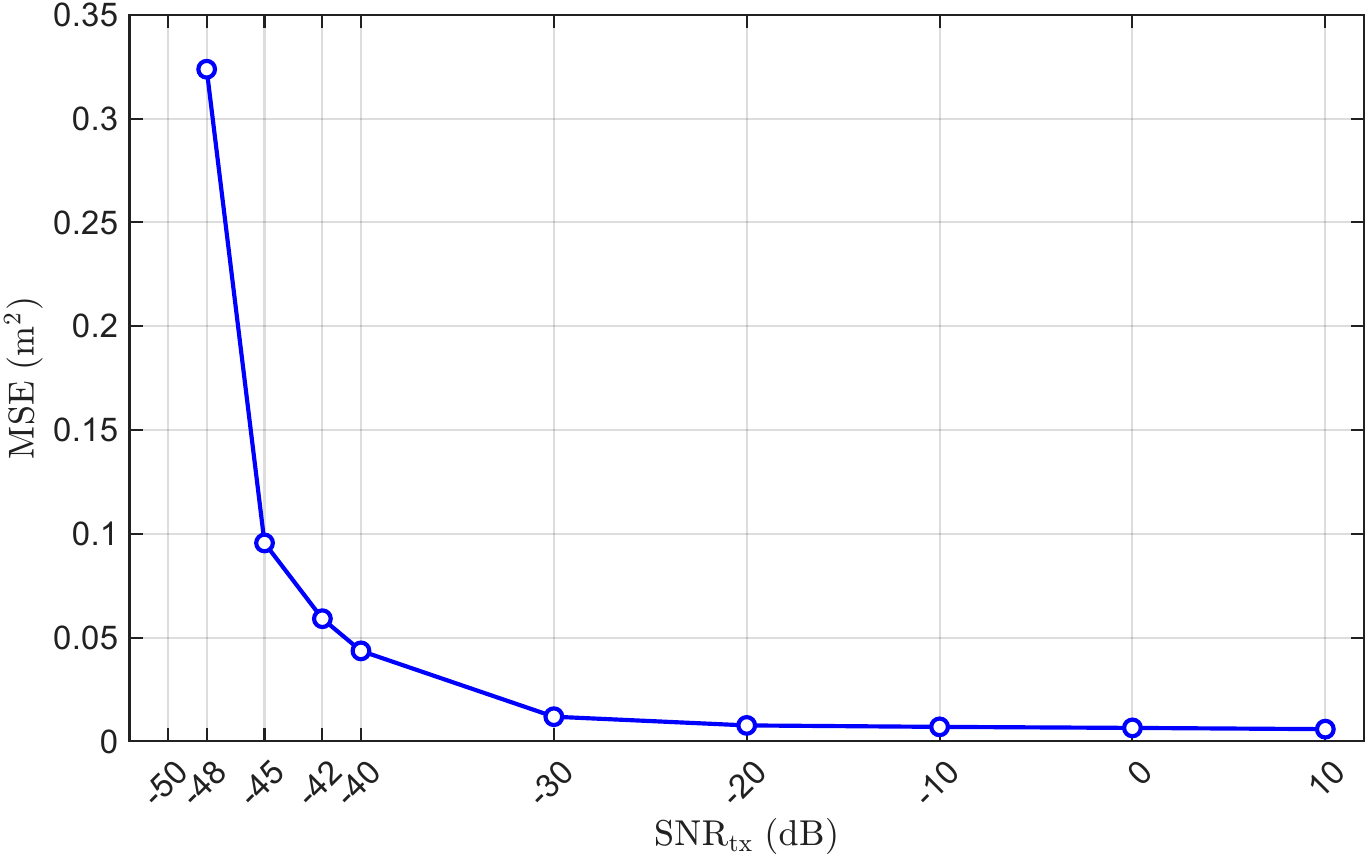}
    }
    \subfigure[Table \ref{tab:scenario-cw-tx}: $N_t=N_r$ vs MSE]{
        \centering
        \includegraphics[height=2.9cm,width=4.8cm]{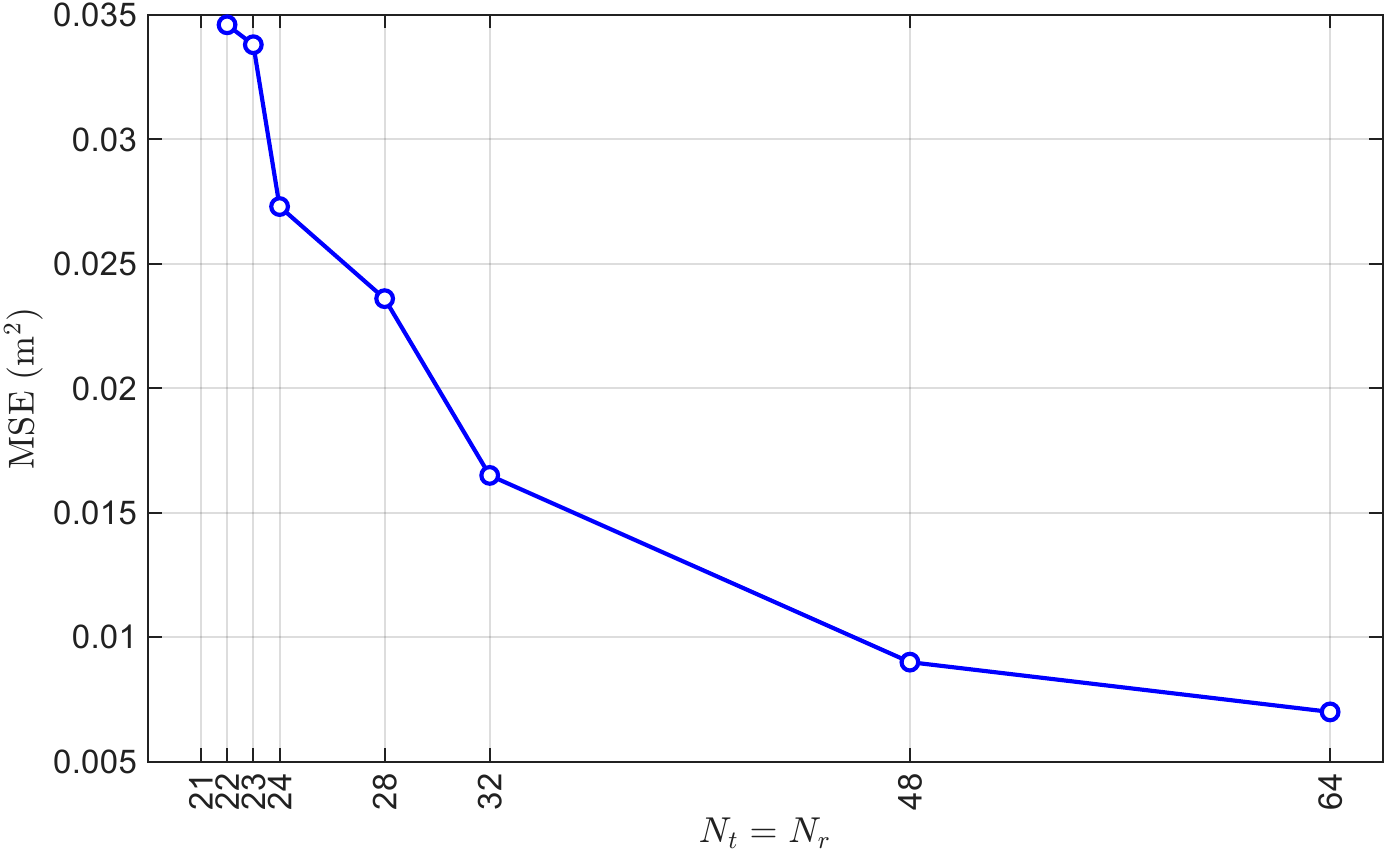}
    }
    
    \subfigure[Table \ref{tab:scenario-cw-sub}: $N_c$ vs MSE]{
        \centering
        \includegraphics[height=2.9cm,width=4.8cm]{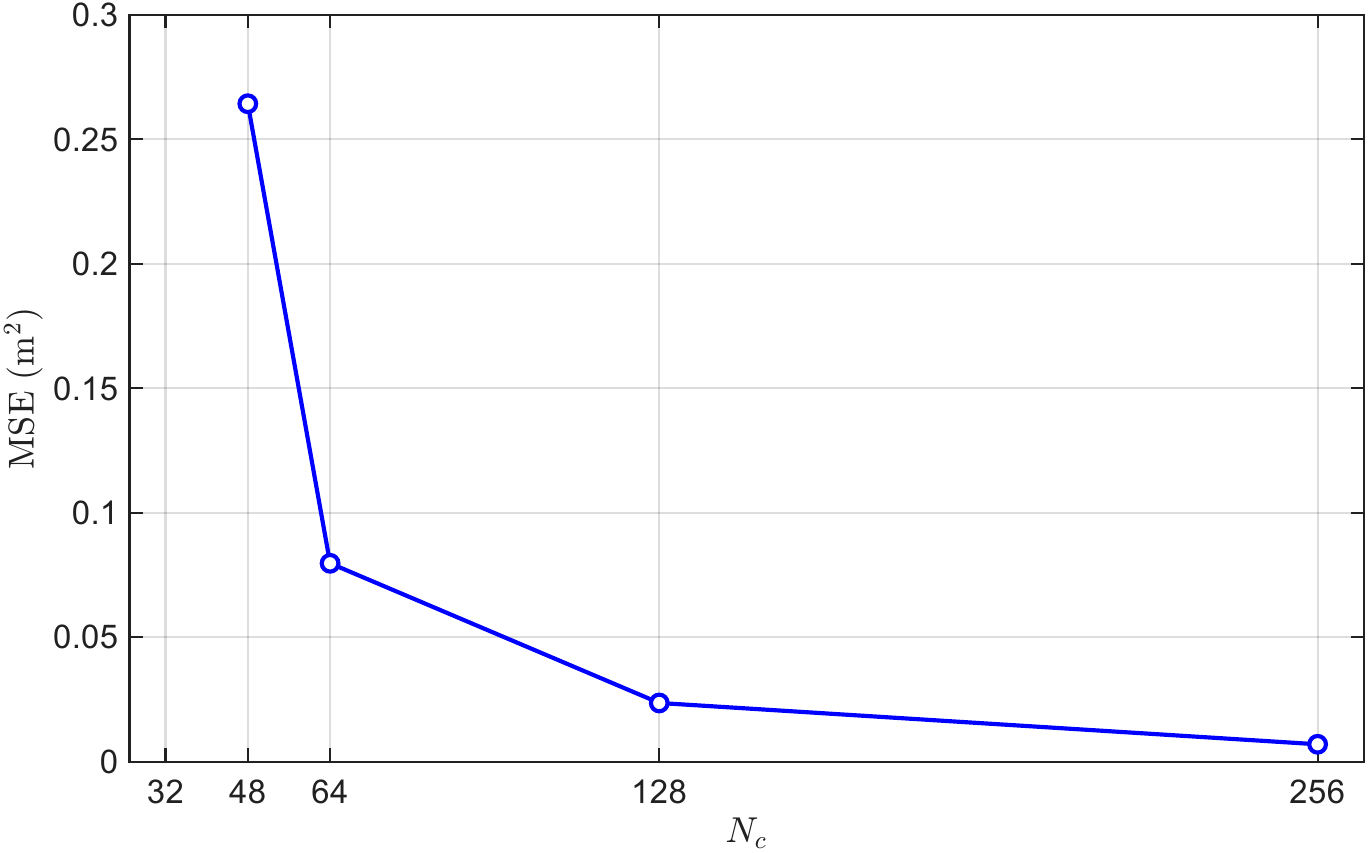}
    }
    \subfigure[Table \ref{tab:scenario-cw-Np}: $N_p$ vs MSE]{
        \centering
        \includegraphics[height=2.9cm,width=4.8cm]{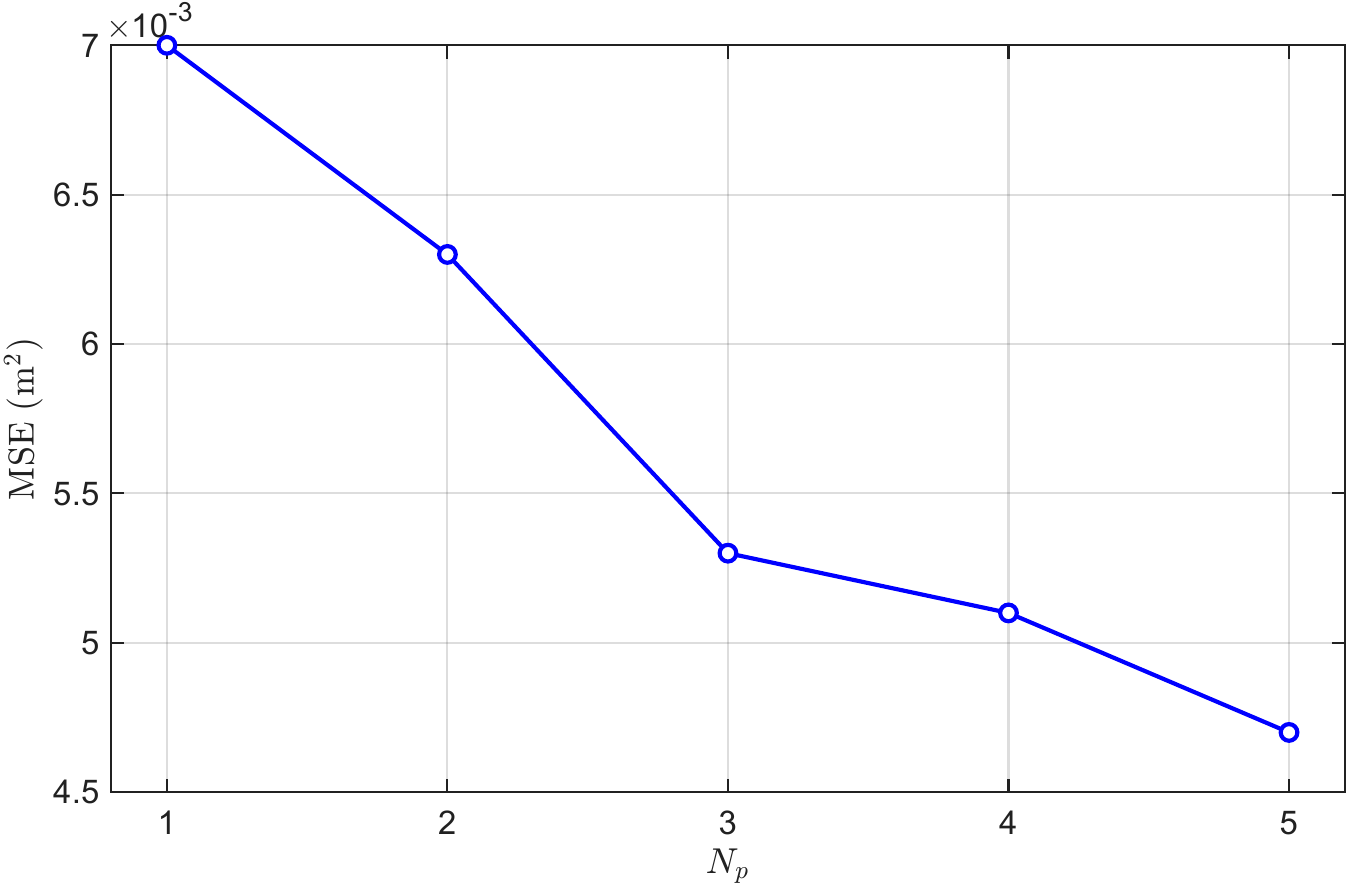}
    }
    \subfigure[Table \ref{tab:scenario-cw-xi}: $\xi$ vs MSE]{
        \centering
        \includegraphics[height=2.9cm,width=4.8cm]{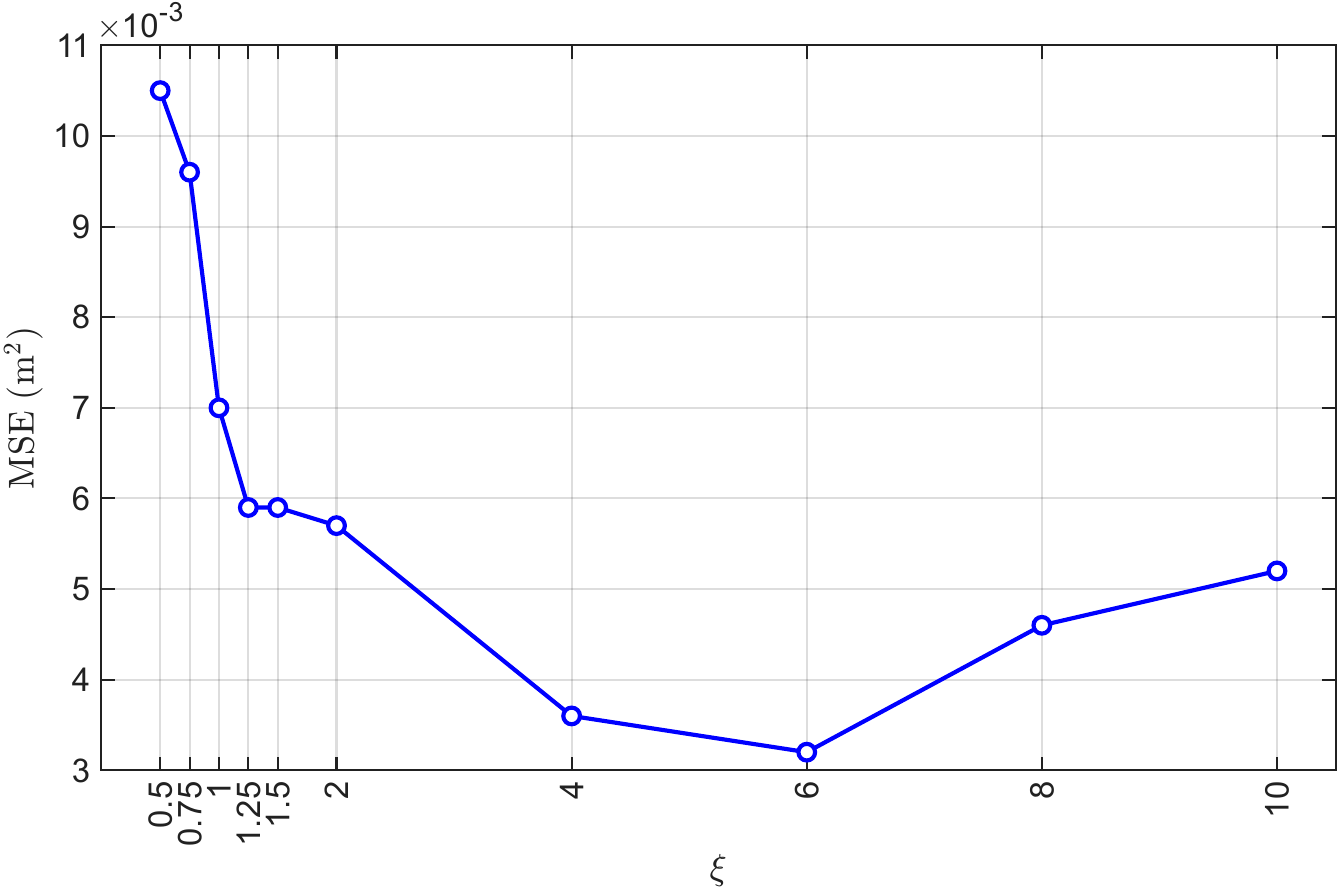}
    }
    
    \caption{Visualization of the tracking MSEs versus configuration parameters in Tables \ref{tab:scenario-cw-par}-\ref{tab:scenario-cw-xi}.}
    \label{fig:CW_table_visualizaiton}
\end{figure}

\begin{figure}[!htbp]
    \centering

    \subfigure[$N_\text{par} = 12$]{
        \centering
        \includegraphics[height=2.9cm]{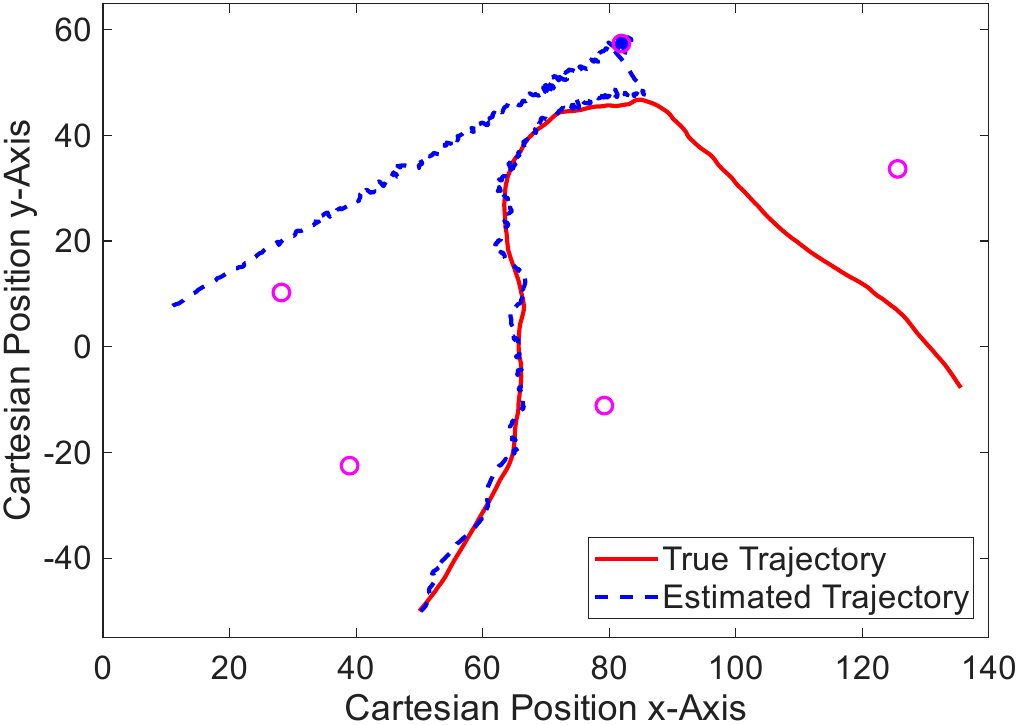}
    }
    \subfigure[Transmit SNR = -50dB]{
        \centering
        \includegraphics[height=2.9cm]{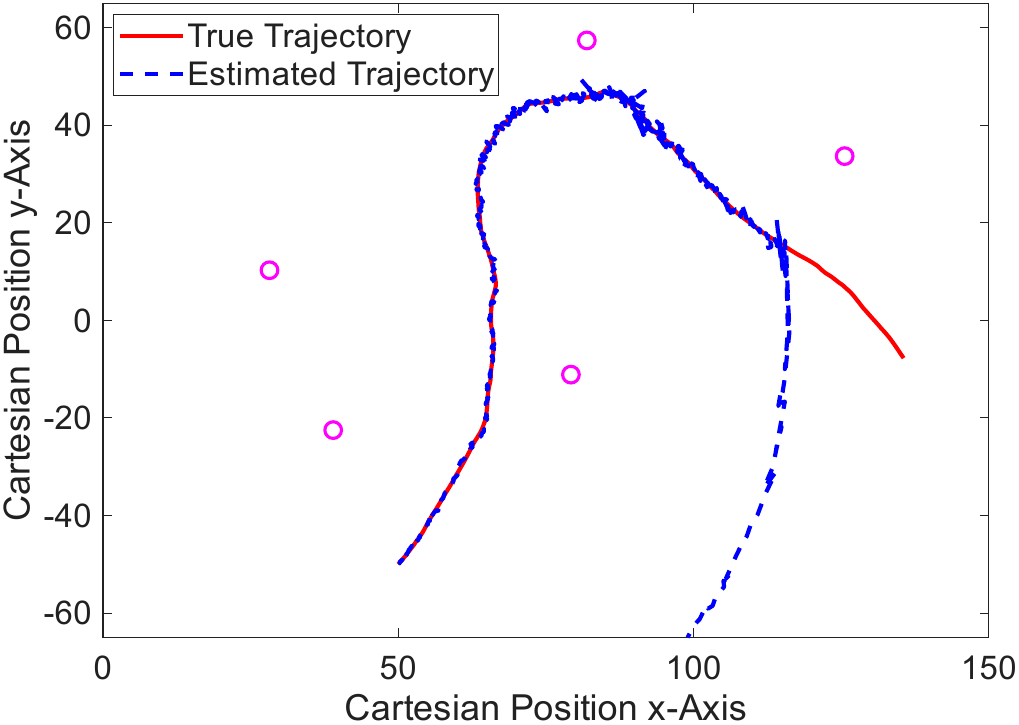}
    }
    \subfigure[$N_t = N_r = 21$]{
        \centering
        \includegraphics[height=2.9cm]{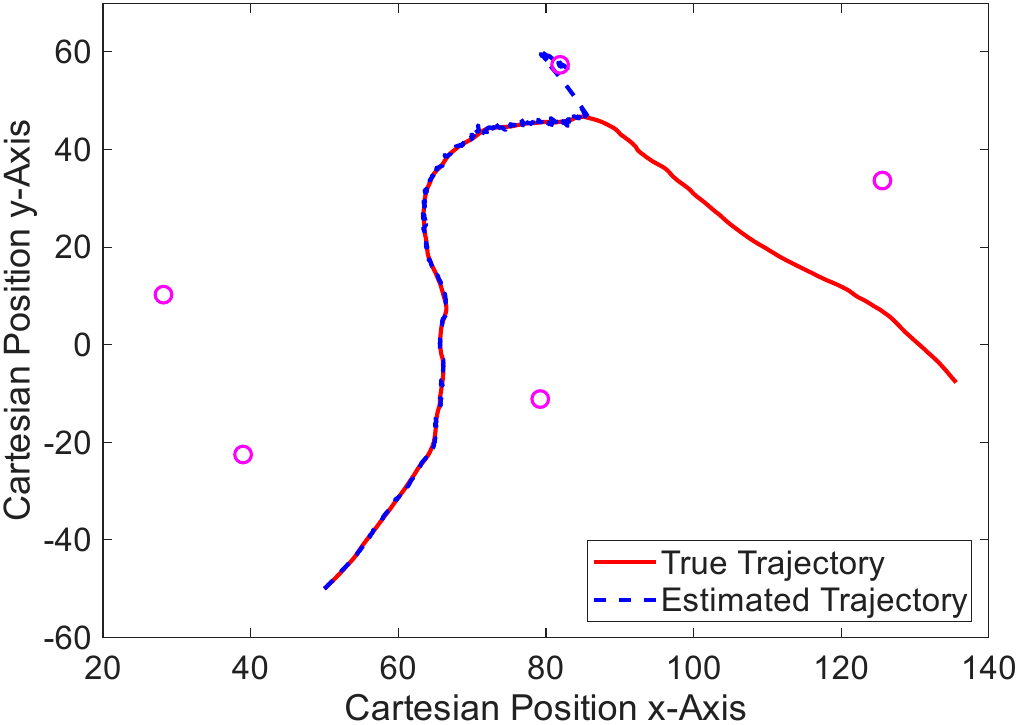}
    }
    \subfigure[$N_c = 32$]{
        \centering
        \includegraphics[height=2.9cm]{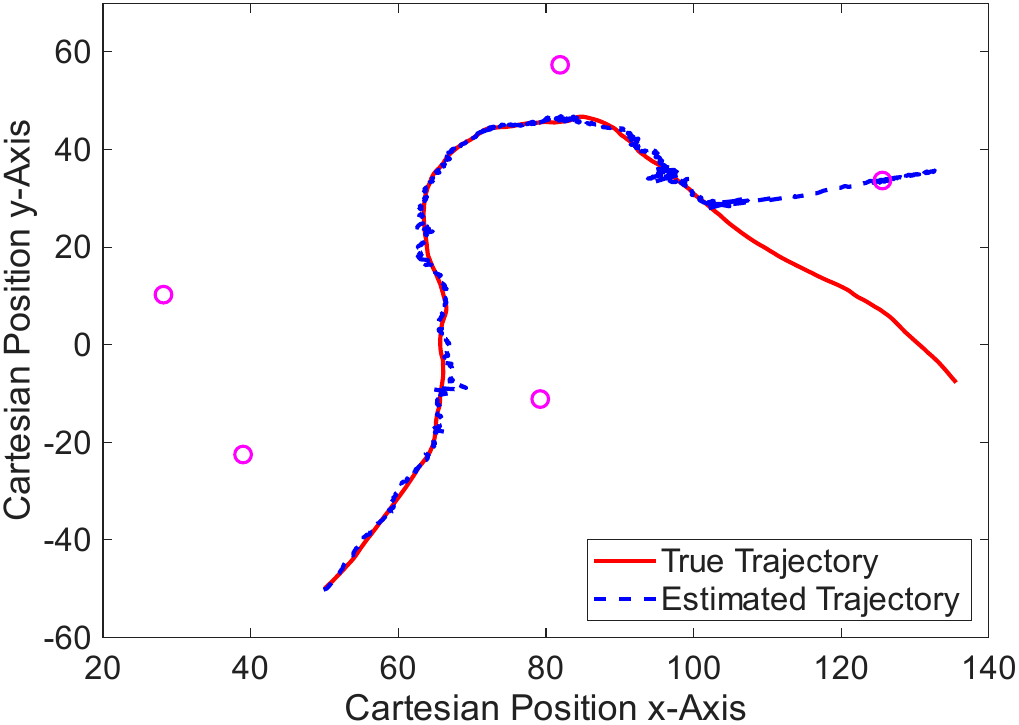}
    }
    
    \caption{Tracking results of PF-SLTR under breakdown conditions; cf. Tables \ref{tab:scenario-cw-par}, \ref{tab:scenario-cw-snr}, \ref{tab:scenario-cw-tx}, and \ref{tab:scenario-cw-sub}.}
    \label{fig:CW_trajectory_tracking_failed_results}
\end{figure}

\subsection{Pulsed MIMO-OFDM Case}
We consider a pulsed MIMO-OFDM ISAC system with center frequency $f_c = 10$GHz for tracking a \textit{far} and \textit{fast} ToI. The full-duplex working mode is assumed. Unless otherwise stated, all parameter settings remain the same as the continuous-wave case in Fig. \ref{fig:CW_trajectory}.

The experimental scenario is displayed in Fig. \ref{fig:pulsed_trajectory}. Therein, the ToI is moving along a curved trajectory, and five static scatterers are assumed, with ranges $[300,450,1000,800,1300]$ and DoAs $[20,-30,35,-8,15]$, respectively. In this case, the ToI and scatterers can be far away from the ISAC station, and the ToI moves quickly; cf. Fig. \ref{fig:CW_trajectory}(c) and Fig. \ref{fig:pulsed_trajectory}(c).
\begin{figure}[!htbp]
    \centering

    \subfigure[Trajectory]{
        \centering
        \includegraphics[height=2.9cm]{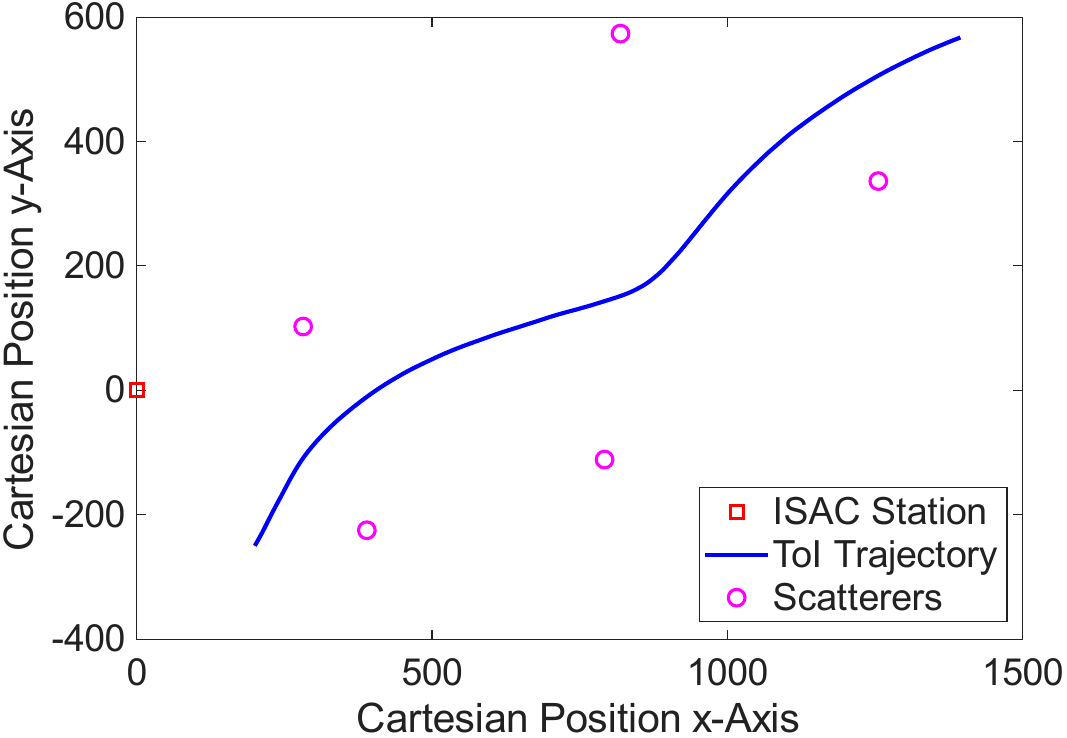}
    }
    \subfigure[Range]{
        \centering
        \includegraphics[height=2.9cm]{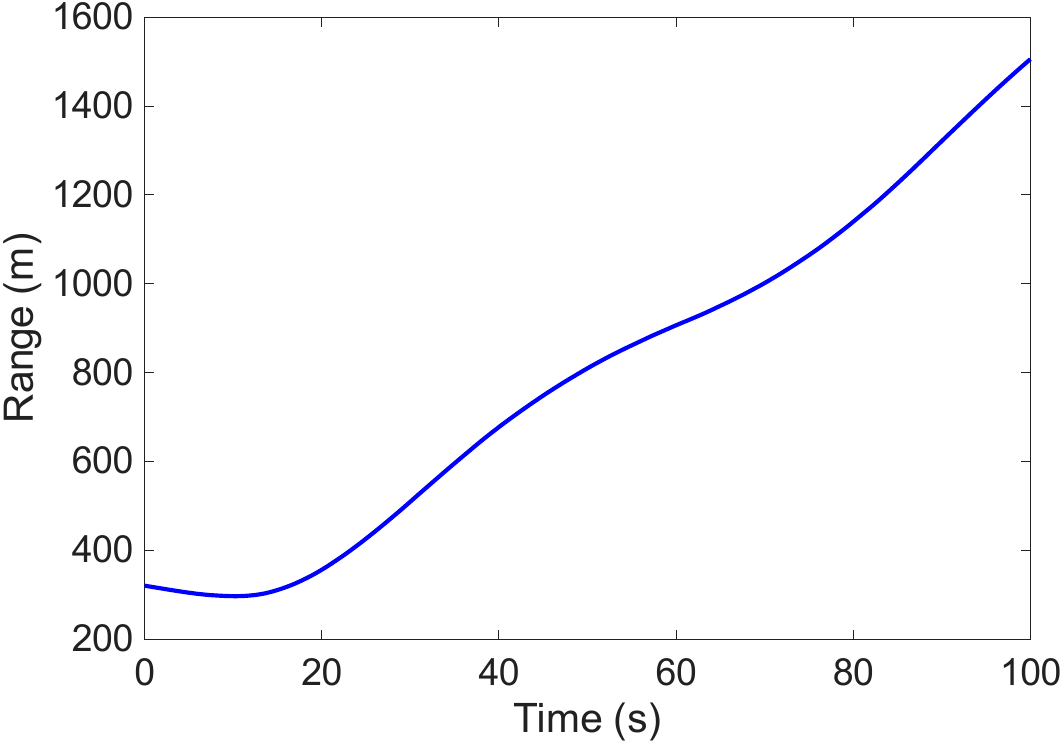}
    }
    \subfigure[Radial Velocity]{
        \centering
        \includegraphics[height=2.9cm]{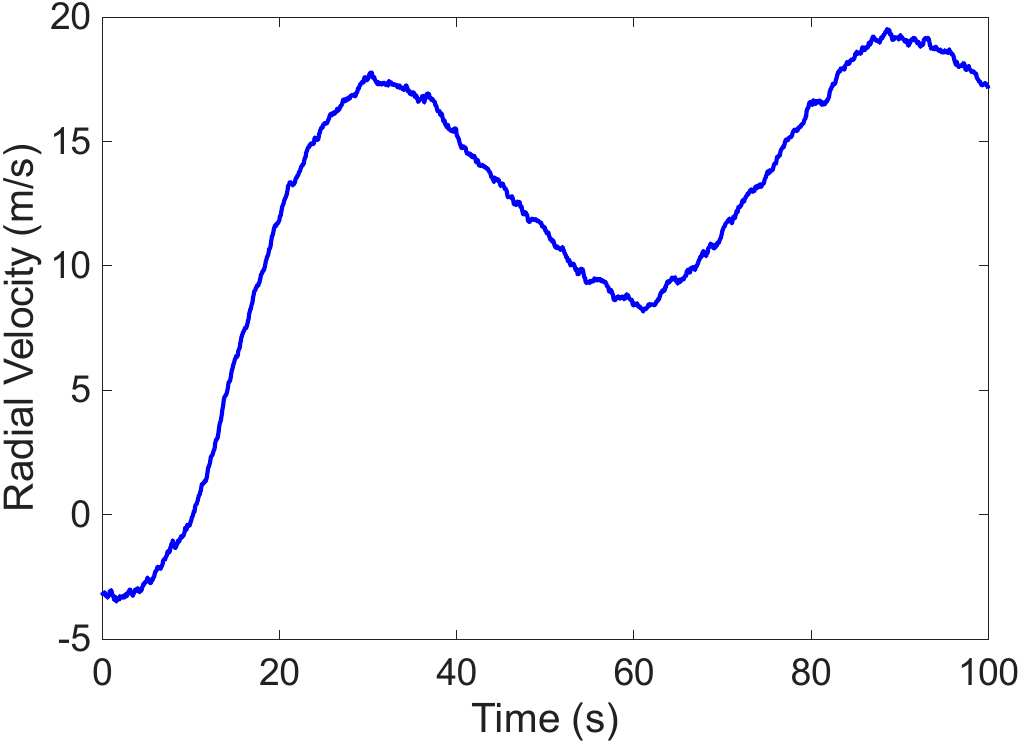}
    }
    \subfigure[DoA]{
        \centering
        \includegraphics[height=2.9cm]{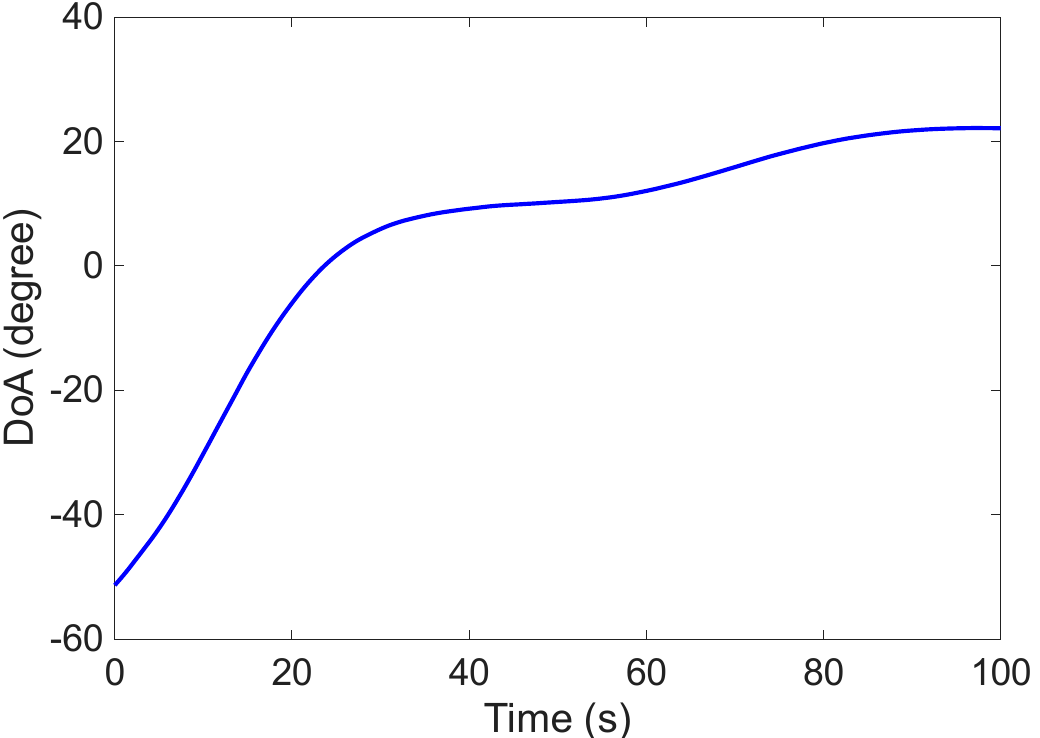}
    }
    
    \caption{Tracking the maneuvering ToI using pulsed MIMO-OFDM ISAC.}
    \label{fig:pulsed_trajectory}
\end{figure}

We suppose that the pulse repetition interval $T_r \defeq 20\mu \text{s}$. Because the pulse length is $T_p = 6\mu$s, as set in the main body of the article, the silent time is $14\mu$s and the duty cycle is $6/20$. Consequently, the maximum detectable round-trip delay is $20\mu$s in the full-duplex mode, corresponding to the maximum detectable distance of $3.0$Km. To enable PF-SLTR for the pulsed case, the sensing-centric waveform structure in Fig. \ref{fig:OFDM_ISAC_frame}(b) is used, and the cost function in \eqref{eq:cost-func-freq-domain-AF} is employed. Since no data symbols are transmitted during the silent interval, the average communication throughput is reduced by the duty-cycle factor $T_p/T_r=6/20=0.3$.

The tracking results of KF-ILTR and PF-SLTR are visually illustrated in Fig. \ref{fig:pulsed_trajectory_tracking_results}. As indicated, PF-SLTR demonstrates a clear advantage over KF-ILTR, especially visible in trajectory and DoA tracking. The results associated with PF-SLTR-A and RBPF-SLTR-A are not displayed because they diverge.
\begin{figure}[!htbp]
    \centering

    \subfigure[Trajectory]{
        \centering
        \includegraphics[height=2.9cm]{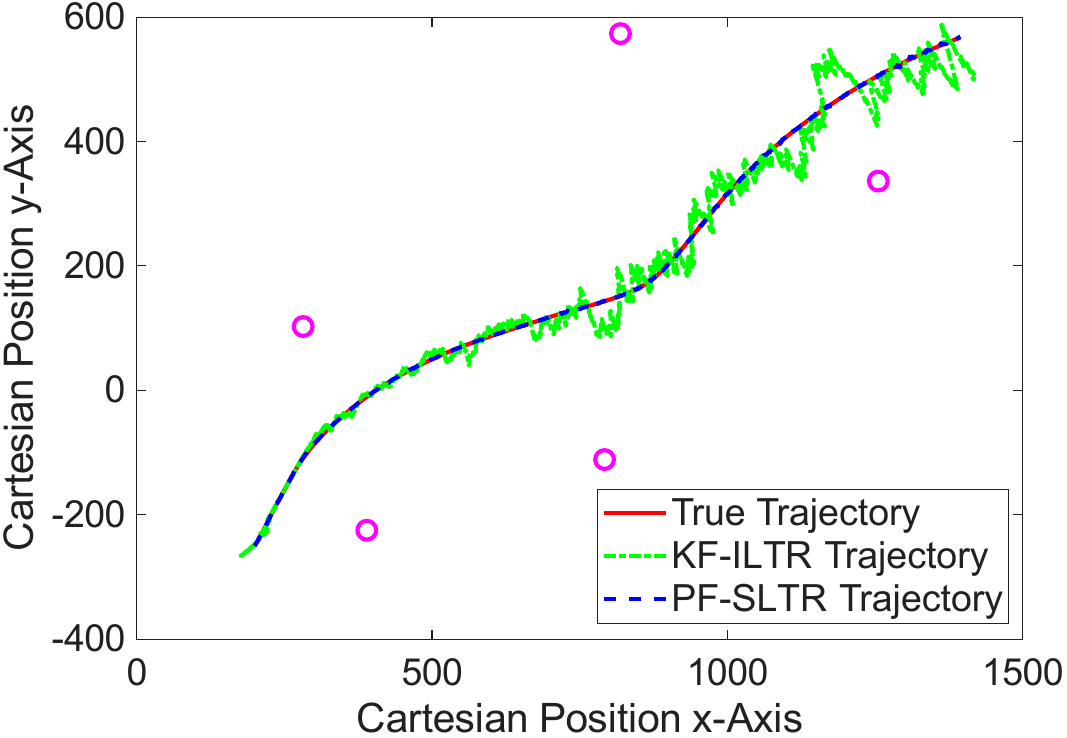}
    }
    \subfigure[Range]{
        \centering
        \includegraphics[height=2.9cm]{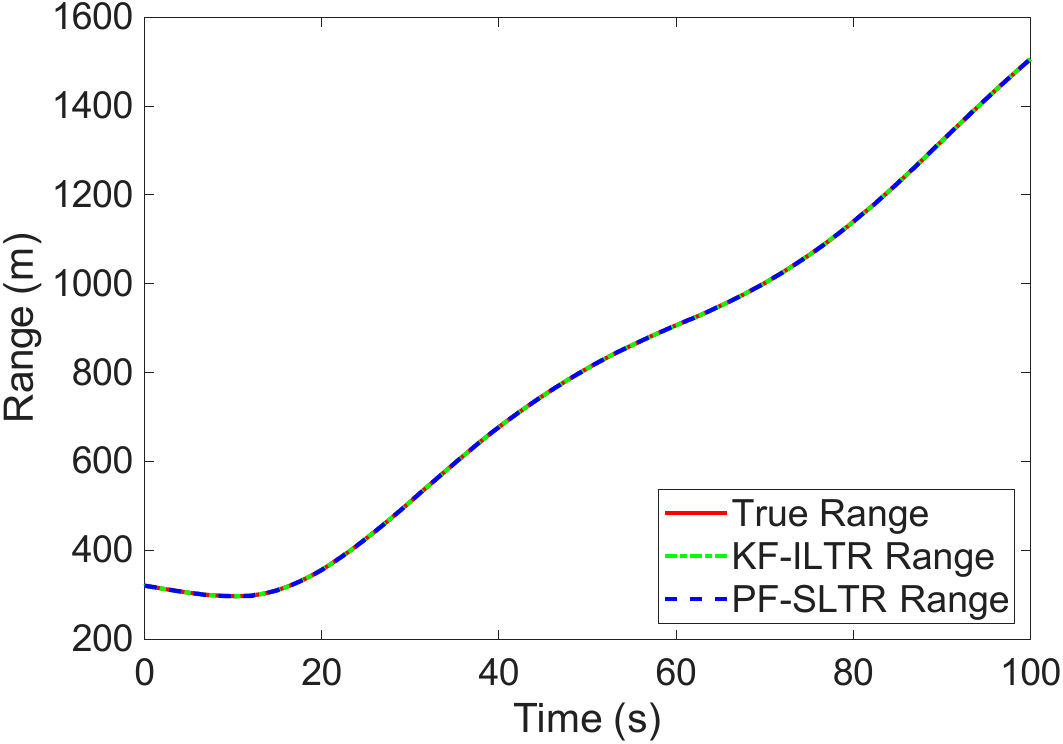}
    }
    \subfigure[Radial Velocity]{
        \centering
        \includegraphics[height=2.9cm]{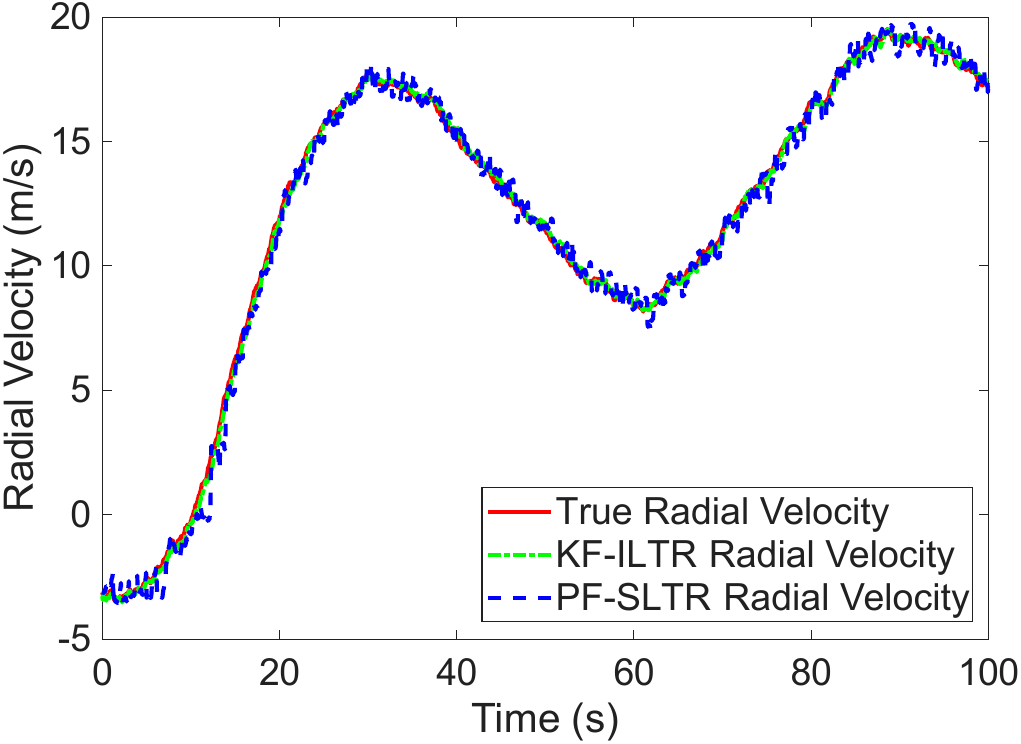}
    }
    \subfigure[DoA]{
        \centering
        \includegraphics[height=2.9cm]{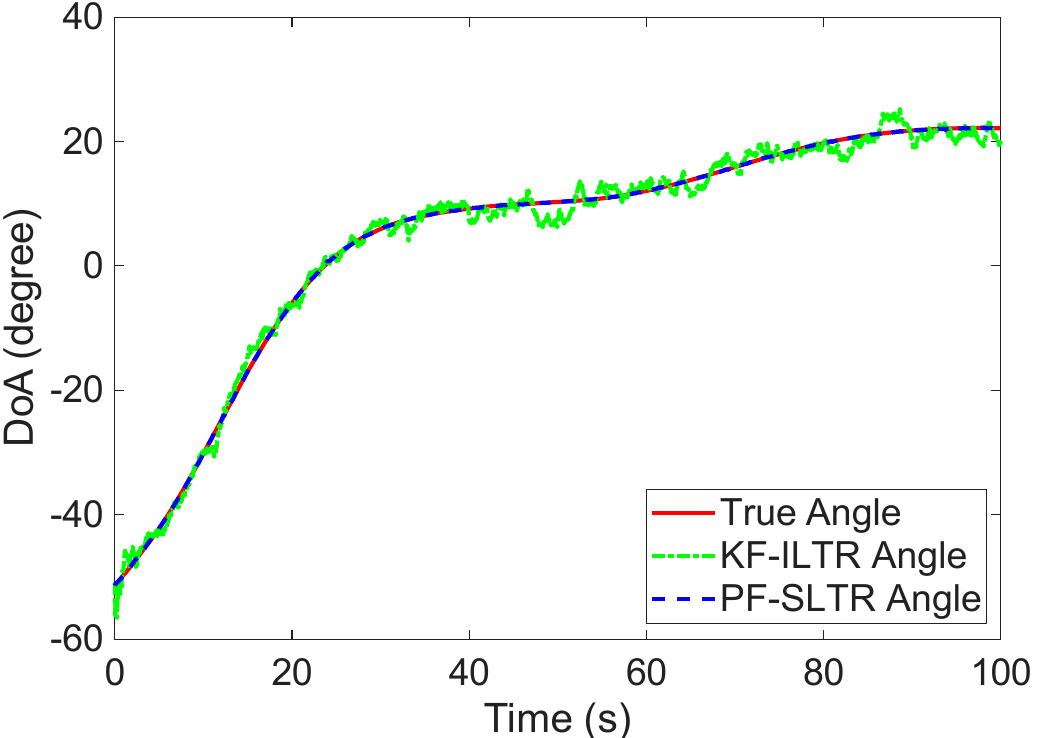}
    }
    
    \caption{Tracking results of KF-ILTR and PF-SLTR in the pulsed case. As the range increases, the tracking performance deteriorates (particularly visible for KF-ILTR) because the receive SNR decreases rapidly due to range-dependent path loss; see \eqref{eq:path-gain}.}
    \label{fig:pulsed_trajectory_tracking_results}
\end{figure}

We do not report other supplementary experimental results for the pulsed case, as they are qualitatively similar to those obtained for the CW case. Interested readers can reproduce the corresponding results (cf. Tables \ref{tab:scenario-cw-par}$-$\ref{tab:scenario-cw-xi}) using the source code available on GitHub: \url{https://github.com/Spratm-Asleaf/PF-SLTR-MIMO-OFDM}.


\end{document}